\ifdefined\pdfsuppressptexinfo
\fi

\ifdefined\pdfinfoomitdate
\fi

\ifdefined\pdftrailerid
  \pdftrailerid{}
\fi

\documentclass{article}
\usepackage{iclr2027_conference,times}

\usepackage{amsmath,amsfonts,bm}

\def\eqref#1{equation~\ref{#1}}
\def\Eqref#1{Equation~\ref{#1}}

\def\1{\bm{1}}

\DeclareMathAlphabet{\mathsfit}{\encodingdefault}{\sfdefault}{m}{sl}
\SetMathAlphabet{\mathsfit}{bold}{\encodingdefault}{\sfdefault}{bx}{n}

\usepackage[utf8]{inputenc}
\usepackage[T1]{fontenc}
\usepackage[english]{babel}
\usepackage{microtype}

\usepackage{amsmath}
\usepackage{amssymb}
\usepackage{amsfonts}
\usepackage{mathtools}
\usepackage{amsthm}
\usepackage{eucal}
\usepackage{bm}
\usepackage{nicefrac}

\usepackage{graphicx}
\usepackage{tikz}
\usepackage{subfigure}
\usepackage{wrapfig}
\usepackage{setspace}

\usepackage{booktabs}
\usepackage{colortbl}
\usepackage{multirow}
\usepackage{tabularx}
\usepackage{enumitem}

\usepackage{algpseudocode}

\usepackage{xcolor}
\usepackage{url}
\usepackage{xspace}
\usepackage{pifont}
\usepackage{stackengine}
\usepackage{etoc}
\usepackage{blindtext}

\usepackage{hyperref}
\hypersetup{
    breaklinks=true,
    colorlinks=true,
    citecolor=blue,
    linkcolor=blue,
    urlcolor=blue,
    hypertexnames=false
}

\newtheorem{lemma}{Lemma}

\newtheorem{proposition}{Proposition}

\newcolumntype{C}[1]{>{\centering\arraybackslash}p{#1}}
\newcolumntype{L}[1]{>{\raggedright\arraybackslash}p{#1}}

\def\ourSystem{\text{CoRF}\xspace}
\newcommand{\ourSystemPR}{\ensuremath{\mathrm{CoRF}_{\mathrm{PR}}}\xspace}
\newcommand{\ourSystemCAL}{\ensuremath{\mathrm{CoRF}_{\mathrm{CAL}}}\xspace}
\newcommand{\ourSystemPRCAL}{\ensuremath{\mathrm{CoRF}_{\mathrm{PR,CAL}}}\xspace}

\def\versus{\textit{vs.}\@\xspace}

\newlength{\mylen}
\graphicspath{{figures/}}

\title{CoRF: Cross-Scene RF Synthesis by Learning Propagation and Preserving Array Physics}

\author{
  Kang Yang\textsuperscript{1} \hspace{2em}
  Duaa Nakshbandi\textsuperscript{2} \hspace{2em}
  Wan Du\textsuperscript{2} \hspace{2em}
  Mani Srivastava\textsuperscript{1}\thanks{Mani Srivastava holds concurrent appointments as a Professor of ECE and CS (joint) at the University of California, Los Angeles.}
  \\[0.5em]
  \textsuperscript{1}Department of Electrical and Computer Engineering, University of California, Los Angeles\\
  \textsuperscript{2}Department of Computer Science and Engineering, University of California, Merced\\
  {\footnotesize\texttt{kyang73@g.ucla.edu \quad \{duaanakshbandi,wdu3\}@ucmerced.edu \quad mbs@ucla.edu}}
}

\iclrfinalcopy

\begin{document}

\maketitle

\lhead{}

\addtocontents{toc}{\protect\setcounter{tocdepth}{-1}}

\begin{abstract}

Existing radio-frequency~(RF) neural fields fit each scene separately, making new-scene deployment measurement- and optimization-intensive.
This work studies amortized cross-scene spatial spectrum synthesis, where a shared model learns propagation across scenes and instantiates an unseen scene from sparse target-scene measurements without scene-specific training.
To achieve this,~\ourSystem separates learned scene-dependent propagation from the known receiver-array observation model.
An unordered set of spectrum-only references conditions a canonical anchor field, producing arrival directions and query-dependent component powers for each query.
Analytic array physics maps these components to the receiver covariance and then to the spatial spectrum.
This factorization keeps the pretrained propagation model frozen, enabling synthesis at arbitrary query locations after a single reference-conditioning pass.
A necessary local reference-capacity bound and an error decomposition further characterize the formulation.
Across~\(35\) simulated scenes spanning seven categories,~\ourSystem outperforms the strongest baseline by~\(7.09\,\mathrm{dB}\)~PSNR on unseen variants of represented scene categories and by~\(6.97\,\mathrm{dB}\) on entirely unseen scene categories.

\end{abstract}

\vspace{\mylen}
\section{Introduction}
\label{sec:intro}
\vspace{\mylen}
    
Spatial spectrum synthesis predicts the angular distribution of received signal power at a receiver antenna array for transmitters at unmeasured locations~\citep{krim1996twodecades,zhao2023nerf2}.
This directional information supports wireless communication and sensing tasks including beam management, direction-of-arrival estimation, localization, and environment-aware communication~\citep{krim1996twodecades,schmidt1986music,guo2025nbf,wang2026radiodiffv2,zeng2024ckm}.
Recent radio-frequency~(RF) neural fields and Gaussian representations synthesize spectra at unmeasured locations~\citep{zhao2023nerf2,lu2024newrf,wen2025wrfgs,zhang2024rf3dgs,yang2025gsrf}, but fit each scene separately.
Deploying these methods in a new scene therefore requires a large target-scene measurement set and a new scene-specific optimization.
In contrast, we study amortized cross-scene synthesis, where a shared model learns a propagation prior across scenes and instantiates an unseen scene from sparse target-scene spectra without scene-specific training.
GRaF~\citep{yang2026graf} is the closest prior work under this input setting, but synthesizes each query from geographically nearby measurements and therefore depends on local reference coverage.

Cross-scene RF synthesis raises two challenges.
\emph{I. Reference correspondence.}
Generalizable rendering in vision establishes query-reference correspondences by projecting query rays into reference views using scene geometry~\citep{yu2021pixelnerf,wang2021ibrnet,charatan2024pixelsplat,chen2024mvsplat}.
RF spatial spectra lack such direct correspondence because each spectrum superposes propagation paths whose directions and powers change with the transmitter location~\citep{zhao2023nerf2,lu2024newrf,orekondy2023winert}.
\emph{II. Propagation and array entanglement.}
Each measured spectrum combines scene-dependent propagation with the receiver-array response~\citep{krim1996twodecades,schmidt1986music,stoica2011spice}.
Propagation varies with scene geometry and material properties~\citep{orekondy2023winert,hoydis2023sionnart,lu2024newrf}, whereas the observation operator is determined by the receiver configuration.
Jointly learning both factors expends model capacity on reproducing known array physics rather than modeling cross-scene propagation.

We introduce~\textbf{\ourSystem} for cross-scene RF spatial spectrum synthesis from sparse measurements.
\ourSystem is trained across source scenes to learn scene-dependent propagation while retaining the receiver-array response as an analytic operator.
Scene propagation is represented by canonical anchors, which are persistent representation slots at fixed coordinates in a normalized volume.
At deployment, an unordered set of target-scene spectra, termed \emph{references}, conditions these anchors once to instantiate the propagation representation of an unseen scene.
The resulting anchor field is reused across query locations to predict arrival directions, mean component powers, and power variances.
An analytic decoder constructs the receiver covariance from these components and then renders the spatial spectrum through the Bartlett response~\citep{krim1996twodecades}, with the power variances providing a query-level reliability estimate.
Thus, feed-forward~\ourSystem operates in an unseen scene through reference conditioning rather than scene-specific parameter fitting, keeping the pretrained propagation model frozen.
The formulation is characterized by a necessary local reference-capacity bound and a propagation-to-spectrum error bound linking component and renderer errors to spectrum error.

Beyond feed-forward~\ourSystem, we develop two optional deployment extensions.
\ourSystemPR refines only the mean component powers using the target-scene reference set, while~\ourSystemCAL calibrates only the analytic receiver response and keeps the learned propagation model fixed.
Cross-scene generalization is evaluated on~\(35\) simulated scenes generated with Sionna RT~\citep{hoydis2023sionnart}, spanning seven categories with five variants each.
Using~\(M=32\) references, feed-forward~\ourSystem outperforms the strongest baseline by~\(7.09\,\mathrm{dB}\)~PSNR on unseen variants of categories represented during training and by~\(6.97\,\mathrm{dB}\) on categories excluded entirely from training.
\textit{Our contributions are as follows:}
\begin{itemize}[label=\textbullet, leftmargin=1em, itemsep=0.2em, topsep=-0.2em, parsep=0pt, partopsep=0pt]

\item We formulate cross-scene spatial spectrum synthesis as amortized propagation inference from an unordered set of spectrum-only target-scene references.

\item We introduce~\ourSystem, which conditions canonical anchors on sparse references to infer scene-dependent propagation and renders these components using analytic receiver-array physics.

\item We demonstrate cross-scene generalization to unseen scene variants and entirely unseen scene categories using sparse target-scene references.

\end{itemize}

\vspace{\mylen}
\section{Related Work}
\label{sec:related_work}
\vspace{\mylen}

\textbf{Spatial Spectrum Synthesis.}
Existing RF neural fields and Gaussian representations reconstruct propagation from scene-specific measurements and typically require separate optimization for each scene~\citep{zhao2023nerf2,lu2024newrf,orekondy2023winert,wen2025wrfgs,wen2024wrfgs,zhang2024rf3dgs,yang2025gsrf,sudo2025swiftwrf,rfpgs2025,nukapotula2025gsparc,huang2024dart}.
Related work explores receiver generalization, sparse differentiable fitting, physics-informed generalizable architectures, and transfer using explicit scene geometry~\citep{chen2024rfcanvas,bian2025genert,shen2026gainerf,zhang2026radtwin}.
Geometry-conditioned methods such as~RadTwin require a point cloud or mesh~\citep{zhang2026radtwin}, whereas our setting uses measured spectra when such a scene model is unavailable.
GRaF~\citep{yang2026graf} is the closest prior work with the same input modality because it performs feed-forward synthesis from reference measurements.
However, each query relies on geographically nearby spectra and therefore requires local reference coverage.
In contrast,~\ourSystem constructs a reusable scene representation from a sparse reference set collected once for the scene and applies it across query locations.

\textbf{Feed-Forward Reconstruction in Vision.}
Generalizable neural rendering predicts new views by conditioning on reference observations~\citep{yu2021pixelnerf,chen2021mvsnerf,wang2021ibrnet,wang2023gnt,sajjadi2022srt,charatan2024pixelsplat,chen2024mvsplat,wang2025vggt,jin2025lvsm}.
Many of these methods exploit geometric structure and cross-view correspondence between reference and query views.
RF spatial spectra lack comparable correspondence because their multipath directions and powers vary across transmitter locations~\citep{zhao2023nerf2,lu2024newrf,orekondy2023winert}.
Permutation-invariant set models and conditional neural processes provide mechanisms for aggregating unordered reference observations~\citep{zaheer2017deepsets,lee2019settransformer,garnelo2018cnp,kim2019anp,nguyen2022tnp}.
Accordingly,~\ourSystem conditions a canonical anchor field on unordered RF reference spectra to form a shared scene representation.

\textbf{Analytic Array Processing.}
Classical array processing recovers angular structure from observed measurements through beamforming, subspace, sparse, and Bayesian methods~\citep{krim1996twodecades,stoica2011spice,yang2013offgridsbl}.
Neural methods further support direction estimation and spectral reconstruction~\citep{shmuel2023subspacenet,shmuel2024transmusic,merkofer2024damusic,spnet2025doa}.
\ourSystem instead predicts propagation for unmeasured transmitter locations and renders the spatial spectrum using the analytic Bartlett response.
Keeping the receiver response explicit allows~\ourSystem to support calibration without retraining propagation~\citep{pan2023insitu}.
Learned compensation can handle broader hardware mismatch but requires receiver-specific data~\citep{liu2018dnnimperfections}.

\begin{figure}[t]
\centering
\includegraphics[width=1.0\linewidth]{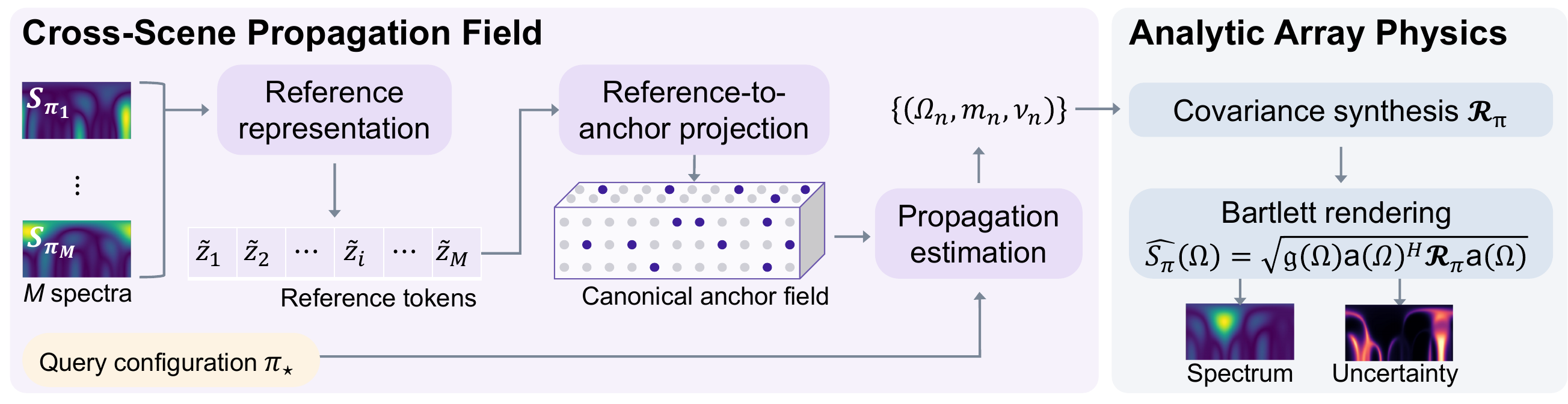}
\caption{
\textit{Framework of \ourSystem.}
Unordered spectrum-only references condition canonical anchors to predict scene-dependent propagation, which analytic array physics renders as the spatial spectrum.
}
\label{fig:overview}
\end{figure}

\vspace{\mylen}
\section{Method}
\label{sec:method}
\vspace{\mylen}

\textbf{Problem Setting.}
We consider a fixed receiver array at position~\(\mathbf t^{\mathrm{rx}}\in\mathbb R^3\) with orientation~\(\mathbf Q\in\mathrm{SO}(3)\).
For a transmitter at~\(\mathbf t_i^{\mathrm{tx}}\in\mathbb R^3\), the measurement configuration is~\(\pi_i=\left(\mathbf t_i^{\mathrm{tx}};\mathbf t^{\mathrm{rx}},\mathbf Q\right)\).
At each configuration, the receiver reports a spatial spectrum~\(S_{\pi_i}\in\mathbb R_+^{H\times W}\) on an~\(H\times W\) angular grid.

Let~\(\mathcal I_{\mathrm{tr}}\) denote the set of training scenes, where each scene provides measurements~\(\left\{\left(\pi_j,S_{\pi_j}\right)\right\}_j\).
For each training instance, we sample a scene~\(s\in\mathcal I_{\mathrm{tr}}\), an unordered reference set~\(\mathcal C_s=\left\{S_{\pi_i}\right\}_{i=1}^{M}\), and a query configuration~\(\pi_\star\) outside the reference set.
A single model~\(f_\theta\) is trained across scenes:
\begin{equation}
\widehat S_{\pi_\star}
=
f_\theta\left(
\mathcal C_s,
\pi_\star
\right),
\qquad
\min_\theta
\mathbb E_{s\in\mathcal I_{\mathrm{tr}}}
\mathbb E_{\mathcal C_s,\pi_\star}
\left[
\ell_{\mathrm{spec}}\left(
\widehat S_{\pi_\star},
S_{\pi_\star}
\right)
\right].
\label{eq:task}
\end{equation}
At test time,~\(M\) references from an unseen scene are processed once to construct a reusable scene representation for synthesizing spectra at arbitrary query transmitter locations.

\textbf{Overview.}
As shown in~\autoref{fig:overview},~\ourSystem first conditions a shared canonical anchor field on the reference spectra to construct the scene representation.
Given a query configuration, the conditioned field predicts~\(K\) propagation components~\(\left\{\left(\Omega_n,m_n,v_n\right)\right\}_{n=1}^{K}\) in~\S\ref{sec:cond}.
Analytic array physics then converts these components into the spatial spectrum and uncertainty in~\S\ref{sec:decoder}.
The formulation is analyzed in~\S\ref{sec:theory}, and training and inference are described in~\S\ref{sec:training}.

\vspace{\mylen}
\subsection{Cross-Scene Propagation Field}
\label{sec:cond}
\vspace{\mylen}

The cross-scene propagation field maps the reference spectra and query configuration to query-specific propagation parameters~\(\left\{\left(\Omega_n,m_n,v_n\right)\right\}_{n=1}^{K}\).
It first constructs a scene representation from the unordered references and then conditions this representation on the query to estimate propagation.

\textbf{(I)~Reference Representation.}
Because the reference measurements form an unordered set, their representation should be invariant to input ordering.
Each reference spectrum is encoded independently, and the resulting tokens are then contextualized across the full reference set:
\begin{equation}
\mathbf z_i
=
E_\phi\left(
S_{\pi_i}
\right),
\qquad
\left\{
\widetilde{\mathbf z}_i
\right\}_{i=1}^{M}
=
T_\psi\left(
\left\{
\mathbf z_i
\right\}_{i=1}^{M}
\right).
\label{eq:enc}
\end{equation}
Here,~\(E_\phi\) is a~2D CNN, while~\(T_\psi\) is a Transformer encoder without positional encoding and is therefore permutation-equivariant over the reference set~\citep{lee2019settransformer}.
The contextualized tokens capture each reference in relation to the full set.

\textbf{(II)~Reference-to-Anchor Projection.}
We define~\(K\) anchors~\(\left\{\mathbf x_n\right\}_{n=1}^{K}\subset\mathbb R^3\) at fixed coordinates in a normalized canonical volume, providing shared spatial support across scenes.
After reference conditioning, each anchor predicts a position offset, an occupancy weight, and an emission feature that parameterize a candidate propagation component.
The conditioning combines global context from the reference set with direction-specific evidence aligned with each anchor.

\emph{Reference-token conditioning.}
Let~\(\boldsymbol\eta\left(\mathbf x_n\right)\) denote the positional embedding of anchor~\(n\).
Each anchor cross-attends to the contextualized reference tokens, and the attended feature modulates its position-dependent representation:
\begin{equation}
\left(
\boldsymbol\gamma_n,
\boldsymbol\beta_n
\right)
=
\mathrm{MLP}\left(
\mathrm{CrossAttn}\left(
\boldsymbol\eta\left(\mathbf x_n\right),
\left\{
\widetilde{\mathbf z}_i
\right\}_{i=1}^{M}
\right)
\right),
\qquad
\mathbf h_n
=
\left(
\mathbf 1+\boldsymbol\gamma_n
\right)
\odot
\mathbf h_n^0
+
\boldsymbol\beta_n.
\label{eq:gate}
\end{equation}
Here,~\(\mathbf h_n^0=\mathrm{ReLU}\left(\mathbf W_0\boldsymbol\eta\left(\mathbf x_n\right)\right)\) is the anchor representation before reference conditioning.
The resulting~\(\mathbf h_n\) captures global reference context associated with anchor~\(n\).

\emph{Bearing-aligned sampling.}
Because each reference token compresses a full spatial spectrum, we additionally preserve direction-specific evidence from its angular feature map~\(\mathbf F_i\).
For anchor~\(n\), we sample each feature map along its bearing~\(\bar{\Omega}_n\) in the receiver array frame:
\begin{equation}
\bar{\Omega}_n
=
\angle\left(
\mathbf Q^{\mathsf T}
\frac{
\mathbf x_n-\mathbf t^{\mathrm{rx}}
}{
\left\|
\mathbf x_n-\mathbf t^{\mathrm{rx}}
\right\|
}
\right),
\qquad
\mathbf b_n
=
\left[
\sum_{i=1}^{M}
w_{n,i}
\mathbf W_{\mathrm v}
\mathbf F_i\left(
\bar{\Omega}_n
\right)
\;;\;
\max_i
\mathbf F_i\left(
\bar{\Omega}_n
\right)
\right].
\label{eq:bearing_sampling}
\end{equation}
Here,~\(w_{n,i}\) are attention weights over the references and~\(\mathbf W_{\mathrm v}\) is a value projection.
The attention pool aggregates evidence across references, while the element-wise maximum retains the strongest response along the anchor bearing.
Together,~\(\left\{\left(\mathbf h_n,\mathbf b_n\right)\right\}_{n=1}^{K}\) form the scene-conditioned canonical anchor field used for query-specific propagation estimation.

\textbf{(III)~Propagation Estimation.}
The canonical anchor field is query independent and captures reusable scene structure.
For each query, we decode the anchors into refined geometry and query-conditioned propagation parameters.
This separation allows the anchor geometry to represent scene structure, while the propagation strength adapts to the query transmitter location.

\emph{Anchor decoding.}
We first fuse the global and bearing-aligned features:
\begin{equation}
\left(
\Delta\mathbf x_n,\,
\alpha_n,\,
\mathbf e_n^{0}
\right)
=
\mathcal D_\omega\left(
\mathbf h_n
+
\mathbf W_{\mathrm b}\mathbf b_n
\right),
\qquad
\mathbf x_n'
=
\mathbf x_n+\Delta\mathbf x_n.
\label{eq:anchor}
\end{equation}
Here,~\(\mathcal D_\omega\) is a shared anchor decoder,~\(\Delta\mathbf x_n\) is a learned position offset that refines the canonical anchor location,~\(\alpha_n\) is its occupancy logit, and~\(\mathbf e_n^{0}\) is its query-independent base emission.

\emph{Propagation parameters.}
The refined anchor position~\(\mathbf x_n'\) determines the final arrival direction~\(\Omega_n\) using the same receiver-frame bearing transformation as~\Eqref{eq:bearing_sampling}, with~\(\mathbf x_n\) replaced by~\(\mathbf x_n'\).

In contrast, the component power depends on the query-anchor geometry, allowing the same anchor field to support different transmitter locations.
Let~\(\mathbf q_\star=\mathbf Q^{\mathsf T}\left(\mathbf t_\star^{\mathrm{tx}}-\mathbf t^{\mathrm{rx}}\right)\) and~\(\mathbf x_n^{\mathrm{rx}}=\mathbf Q^{\mathsf T}\left(\mathbf x_n'-\mathbf t^{\mathrm{rx}}\right)\) denote the query and anchor positions in the receiver array frame.
A query-conditioned FiLM head~\citep{perez2018film} modulates the base emission:
\begin{equation}
\left(
\boldsymbol\gamma_n^\star,
\boldsymbol\beta_n^\star
\right)
=
F_{\mathrm q}\left(
\mathbf q_\star,
\mathbf x_n^{\mathrm{rx}}
\right),
\qquad
\mathbf e_n
=
\left(
\mathbf 1+\boldsymbol\gamma_n^\star
\right)
\odot
\mathbf e_n^{0}
+
\boldsymbol\beta_n^\star.
\label{eq:film}
\end{equation}
Using the zeroth-order, direction-independent component of~\(\mathbf e_n\), denoted by~\(e_n\in\mathbb C\), we compute the mean power and power variance as
\begin{equation}
m_n
=
\operatorname{sigmoid}\left(
\alpha_n
\right)
\left|
e_n
\right|^2,
\qquad
v_n
=
\frac{
m_n^2
}{
\kappa_n
},
\label{eq:propagation_power}
\end{equation}
where~\(\kappa_n>0\) is predicted by an uncertainty head.
The resulting propagation parameters~\(\left\{\left(\Omega_n,m_n,v_n\right)\right\}_{n=1}^{K}\) are passed to the analytic array model.

\vspace{\mylen}
\subsection{Analytic Array Physics}
\label{sec:decoder}
\vspace{\mylen}

Given the estimated propagation parameters~\(\left\{\left(\Omega_n,m_n,v_n\right)\right\}_{n=1}^{K}\), the analytic array model first synthesizes the receiver covariance, then renders the spatial spectrum through the Bartlett response, and finally propagates component power uncertainty to the spectrum.

\textbf{(I)~Covariance Synthesis.}
For an~\(N\)-element receiver array, let~\(\mathbf a\left(\Omega\right)\in\mathbb C^N\) denote the steering vector at arrival direction~\(\Omega\).
Under the narrowband and far-field model~\citep{schmidt1986music}, a component with direction~\(\Omega_n\) and mean power~\(m_n\) contributes the rank-one covariance~\(m_n\mathbf a\left(\Omega_n\right)\mathbf a\left(\Omega_n\right)^{\mathsf H}\).
Assuming mutually incoherent components, their covariance contributions add without cross terms.
With an explicit line-of-sight~(LOS) component, the query covariance is
\begin{equation}
\mathbf R_{\pi_\star}
=
\sum_{n=1}^{K}
m_n
\mathbf a\left(\Omega_n\right)
\mathbf a\left(\Omega_n\right)^{\mathsf H}
+
m_{\mathrm{los}}
\mathbf a\left(\Omega_{\mathrm{los}}\right)
\mathbf a\left(\Omega_{\mathrm{los}}\right)^{\mathsf H}.
\label{eq:decoder}
\end{equation}
Here,~\(\Omega_{\mathrm{los}}\) is the arrival direction of the direct~LOS path.
The term~\(m_{\mathrm{los}}\) provides a shared baseline for its normalized power.
Because each target spectrum is normalized independently, this term does not represent absolute received power.
The conditioned non-LOS components model the remaining variation across scenes and query locations.

\textbf{(II)~Bartlett Rendering.}
The covariance is converted to a spatial spectrum by evaluating the Bartlett response over the angular grid~\citep{krim1996twodecades}:
\begin{equation}
\widetilde S_{\pi_\star}\left(\Omega\right)
=
\sqrt{
g\left(\Omega\right)
\mathbf a\left(\Omega\right)^{\mathsf H}
\mathbf R_{\pi_\star}
\mathbf a\left(\Omega\right)
}.
\label{eq:render}
\end{equation}
Here,~\(g\left(\Omega\right)\) denotes the receiver element power pattern, with~\(g\equiv 1\) for the simulated arrays.
Simulated targets are generated from coherent single snapshots, whereas the decoder in~\Eqref{eq:decoder} sums mutually incoherent covariance components without path cross terms.
Accordingly, the renderer is a structured approximation under the narrowband and far-field assumptions of classical array processing~\citep{krim1996twodecades,schmidt1986music}.
The representation-mismatch term in~Proposition~\ref{prop:rendering_error} captures the residual error under perfect propagation estimates.

\textbf{(III)~Component-Power Uncertainty.}
The analytic renderer also propagates uncertainty in the component powers to the spatial spectrum.
Because Bartlett power is linear in the component powers, their variances propagate through the fixed array response in closed form without stochastic sampling.
Let~\(p_n\ge0\) be the random power of component~\(n\), with~\(\mathbb E\left[p_n\right]=m_n\) and~\(\operatorname{Var}\left[p_n\right]=v_n\).
For~\(\mathbf a_n=\mathbf a\left(\Omega_n\right)\), define the unit-power array response~\(\chi_n\left(\Omega\right)=\left|\mathbf a\left(\Omega\right)^{\mathsf H}\mathbf a_n\right|^2\).
With uncorrelated component powers and fixed arrival directions, the random Bartlett power and variance are
\begin{equation}
P_{\pi_\star}^{\mathrm{rand}}\left(\Omega\right)
=
g\left(\Omega\right)
\sum_{n=0}^{K}
p_n
\chi_n\left(\Omega\right),
\qquad
\operatorname{Var}\left[
P_{\pi_\star}^{\mathrm{rand}}\left(\Omega\right)
\right]
=
g\left(\Omega\right)^2
\sum_{n=0}^{K}
v_n
\chi_n\left(\Omega\right)^2.
\label{eq:power_uncertainty}
\end{equation}
The deterministic~LOS component is included as~\(n=0\), with~\(p_0=m_{\mathrm{los}}\) and~\(v_0=0\), so the mean Bartlett power equals~\(\widetilde S_{\pi_\star}\left(\Omega\right)^2\).
Because the synthesized spectrum is the square root of the Bartlett power, we propagate its variance through the square root using a first-order approximation:
\begin{equation}
\operatorname{Var}\left[
S_{\pi_\star}^{\mathrm{rand}}\left(\Omega\right)
\right]
\approx
\frac{
\operatorname{Var}\left[
P_{\pi_\star}^{\mathrm{rand}}\left(\Omega\right)
\right]
}{
4
\widetilde S_{\pi_\star}\left(\Omega\right)^2
}.
\label{eq:spectrum_uncertainty}
\end{equation}
This quantity conditions on fixed arrival directions and analytically propagates component-power uncertainty into a query-level reliability signal.
The full derivation is provided in~Appendix~\ref{app:Power}.

\vspace{\mylen}
\subsection{Theoretical Analysis}
\label{sec:theory}
\vspace{\mylen}
Two bounds characterize the information supplied by reference spectra and the effect of propagation errors on rendering.
For reference~\(i\),~\(\overline{\mathbf J}_i\) measures scene sensitivity after projecting out pose effects.
\begin{proposition}[Reference Capacity]
\label{prop:reference_budget}
Let~\(r_\star\) be the locally observable scene dimension,~\(M^\star\) the minimum number of references to span it, and~\(N\) the number of array elements.
Then,
\begin{equation}
\operatorname{rank}\overline{\mathbf J}_i
\le
N-1,
\qquad
M^\star
\ge
\left\lceil
\frac{r_\star}{N-1}
\right\rceil.
\label{eq:reference_budget_bound}
\end{equation}
\end{proposition}
This necessary local bound shows why complementary references matter, although it does not determine the reference count needed in practice.
The derivation is provided in~Appendix~\ref{app:Information}.

\begin{proposition}[Propagation-to-Spectrum Error]
\label{prop:rendering_error}
Let~\(P\) and~\(\widehat P\) be the target and rendered power spectra,~\(\varepsilon_{\mathrm{phys}}\) the best representable spectrum error, and~\(E_m\) and~\(E_{\Omega}\) the matched power and power-weighted direction errors.
If the array response is bounded and Lipschitz, then
\begin{equation}
\left\|
\widehat P-P
\right\|_\infty
\le
\varepsilon_{\mathrm{phys}}
+
C_m E_m
+
C_{\Omega}E_{\Omega},
\label{eq:rendering_error_bound}
\end{equation}
where~\(C_m\) and~\(C_{\Omega}\) depend on the receiver array.
\end{proposition}
This bound separates renderer mismatch from errors in predicted component powers and arrival directions.
The derivation is provided in~Appendix~\ref{app:errobound}.

\vspace{\mylen}
\subsection{Training and Inference}
\label{sec:training}
\vspace{\mylen}

\textbf{Training.}
We pretrain the spectrum encoder~\(E_\phi\) and reference contextualizer~\(T_\psi\) with supervised contrastive learning over scene labels, encouraging consistent representations for references from the same scene.
With both modules frozen, we train the reference-to-anchor projection and propagation estimation modules to predict spectra at query locations outside the reference set.
The loss combines spectrum reconstruction, reference-use regularization, and heteroscedastic uncertainty learning, as detailed in~Appendix~\ref{app:implementation:training}.
Model architecture is described in~Appendix~\ref{app:implementation:model}.

\textbf{Inference.}
For an unseen scene,~\ourSystem conditions the canonical anchors on its reference set once and reuses the resulting field across queries.
For each query configuration~\(\pi_\star\), it predicts~\(\left\{\left(\Omega_n,m_n,v_n\right)\right\}_{n=1}^{K}\) and analytically renders the spatial spectrum and uncertainty.
Unless stated otherwise,~\ourSystem denotes this feed-forward setting.

\textbf{Power Refinement.}
We also study~\ourSystemPR, a lightweight, power-only target-scene adaptation that fits a correction to the mean component powers~\(\left\{m_n\right\}_{n=1}^{K}\).
It uses the same~\(M\) references without additional measurements and leaves the pretrained encoder and propagation field frozen.
Arrival directions remain fixed during refinement.
The correction is fitted once per target scene and applied to the predicted powers at each query, as detailed in~Appendix~\ref{app:implementation:power_adaptation}.

\textbf{Receiver Calibration.}
Array geometry may come from hardware specifications, and antenna metadata from AISG Antenna Location and Orientation Sensor~(ALS)~\citep{aisg2024als,3gpp38455}.
If the resulting analytic response is accurate, calibration is unnecessary.
Otherwise,~\ourSystemCAL fits a correction to the nominal response using~\(N_{\mathrm C}\) additional spectra, disjoint from the~\(M\) references and test queries, while keeping the propagation model and scene representation fixed.
The fitted response is reused across queries; details appear in~Appendix~\ref{app:implementation:array_calibration}.

\vspace{\mylen}
\section{Evaluation}
\label{sec:evaluation}
\vspace{\mylen}

\textbf{Dataset.}
We use~\(35\) simulated scenes generated with~Sionna RT~\citep{hoydis2023sionnart} and one measured scene from~NeRF$^2$~\citep{zhao2023nerf2}.
The simulated scenes span seven categories with five variants each: \texttt{conference room}, \texttt{classroom}, \texttt{bedroom}, \texttt{corridor}, \texttt{open office}, \texttt{laboratory}, and \texttt{lounge}.
Each variant contains thousands of transmitter measurements.
All scenes use a fixed~\(4\times4\) receiver array with~\(\lambda/2\) spacing and produce~\(90\times360\) spatial spectra at~\(920\,\mathrm{MHz}\).
Further dataset details are provided in~Appendix~\ref{app:implementation:data}.

\textbf{Protocol.}
Each scene uses a~\(70/10/20\) train, validation, and test split.
Cross-scene evaluations use~\(M=32\) target-scene references from the training split and queries from the test split.
\textit{(I)~For per-scene evaluation}, each model is trained and tested within the same scene variant.
\textit{(II)~For generalization within represented scene categories}, five folds exclude one variant per category.
\textit{(III)~For generalization to excluded scene categories}, seven folds exclude an entire category.

\textbf{Baselines.}
We compare with GRaF~\citep{yang2026graf}, the closest generalizable RF field.
GRaF is trained on the same source scenes as~\ourSystem and receives the same~\(M\) references from the target scene.
For each query, its~\(L\) nearest neighbors are selected from these~\(M\) references.

We also compare with three per-scene RF fields: NeRF$^2$~\citep{zhao2023nerf2}, GSRF~\citep{yang2025gsrf}, and WRF-GS~\citep{wen2025wrfgs} using its deformable WRF-GS+ variant~\citep{wen2024wrfgs}.
For each baseline~\(X\), \textit{X-FS} trains from scratch on the~\(M\) target-scene references.
\textit{X-TR} adapts a model trained on one source scene using the same references, with results averaged over all eligible source scenes.

\textbf{\ourSystem Variants.}
\ourSystem denotes the feed-forward version in the results.
Power refinement, receiver calibration, and their combination are labeled~\ourSystemPR,~\ourSystemCAL, and~\ourSystemPRCAL, respectively.

\textbf{Metrics.}
We report peak signal-to-noise ratio~(PSNR)\(\uparrow\), structural similarity index~(SSIM)\(\uparrow\), normalized mean squared error~(NMSE)\(\downarrow\), and angle-of-arrival~(AoA) error\(\downarrow\).
AoA error measures the angular distance between the peak directions of the predicted and target spectra.

\vspace{\mylen}
\subsection{Per-Scene Performance}
\label{sec:eval:perscene}
\vspace{\mylen}

\begin{wraptable}{r}{0.41\linewidth}
\centering
\small
\setlength{\tabcolsep}{3.5pt}
\caption{\textit{Per-scene performance.}}
\label{tab:perscene}
\resizebox{\linewidth}{!}{%
\begin{tabular}{lllll}
\toprule
Method & PSNR$\uparrow$ & SSIM$\uparrow$ & NMSE$\downarrow$ & AoA$\downarrow$ \\
\midrule
NeRF$^2$ & 21.57\,{\scriptsize$\pm$3.57} & 0.698\,{\scriptsize$\pm$0.104} & 0.130\,{\scriptsize$\pm$0.086} & 11.02\,{\scriptsize$\pm$7.81} \\
GSRF & 20.45\,{\scriptsize$\pm$2.43} & 0.677\,{\scriptsize$\pm$0.066} & 0.151\,{\scriptsize$\pm$0.070} & 11.58\,{\scriptsize$\pm$6.90} \\
WRF-GS & 21.93\,{\scriptsize$\pm$3.06} & 0.705\,{\scriptsize$\pm$0.081} & 0.117\,{\scriptsize$\pm$0.064} & 9.89\,{\scriptsize$\pm$6.67} \\
GRaF & 21.96\,{\scriptsize$\pm$3.11} & 0.683\,{\scriptsize$\pm$0.083} & 0.113\,{\scriptsize$\pm$0.060} & 8.91\,{\scriptsize$\pm$6.48} \\
\midrule
\textbf{\ourSystemPR} & \textbf{24.27}\,{\scriptsize$\pm$3.38} & \textbf{0.791}\,{\scriptsize$\pm$0.071} & \textbf{0.083}\,{\scriptsize$\pm$0.052} & \textbf{7.68}\,{\scriptsize$\pm$5.82} \\
\bottomrule
\end{tabular}
}
\end{wraptable}
We first evaluate the conventional per-scene setting, where each method is trained and tested within the same scene variant.
For this comparison,~\ourSystemPR uses the power-refinement protocol.
Across~\(36\) scenes,~\ourSystemPR achieves the best average performance on all four metrics in~\autoref{tab:perscene}.
Relative to the strongest baseline for each metric, PSNR improves from~\(21.96\) to~\(24.27\,\mathrm{dB}\) and SSIM from~\(0.705\) to~\(0.791\).
NMSE decreases from~\(0.113\) to~\(0.083\), and AoA error from~\(8.91^\circ\) to~\(7.68^\circ\).
Together, these gains show improved spectral fidelity and angular accuracy with component-based propagation estimation and analytic array rendering.
Qualitative examples and per-scene breakdowns appear in~Appendices~\ref{app:qualitative:perscene} and~\ref{app:breakdown_perscene}.

\vspace{\mylen}
\subsection{Generalization to Unseen Variants of Known Scene Categories}
\label{sec:eval:seen}
\vspace{\mylen}
We evaluate five folds, each excluding one scene variant per category; all methods use the same~\(M=32\) references from each test scene.
As shown in~\autoref{tab:seen},~\ourSystemPR achieves~\(23.02\,\mathrm{dB}\) PSNR,~\(0.749\) SSIM,~\(0.119\) NMSE, and~\(8.81^\circ\) AoA error.
Compared with~WRF-GS-FS, the strongest baseline on all four metrics, it improves PSNR by~\(7.16\,\mathrm{dB}\) and SSIM by~\(0.254\), while reducing NMSE by~\(0.279\) and AoA error by~\(13.42^\circ\).
Transfer adaptation is inconsistent: NeRF$^2$-TR and WRF-GS-TR perform worse than training from scratch, while GSRF-TR improves only some metrics.
Sparse target-scene references may be insufficient to adapt a representation fitted to another scene.
GRaF selects nearby references for each query, so sparse random sampling may leave some queries without local coverage.
This comparison therefore measures performance under a matched sparse-reference budget, not under denser local sampling.
In contrast,~\ourSystemPR conditions a reusable anchor field on the full reference set and refines component powers using those same measurements.
Its gains in spectral and AoA metrics show improved spectrum synthesis and dominant-direction accuracy on unseen variants.
Qualitative examples and scene-level breakdowns appear in~Appendices~\ref{app:qualitative:variants} and~\ref{app:breakdown_unseen}, respectively.

\vspace{\mylen}
\subsection{Generalization to Unseen Scene Categories}
\label{sec:eval:heldout}
\vspace{\mylen}
We further evaluate category-level generalization in seven folds, each excluding all five variants of one category from training.
All methods use the same~\(M=32\) references from each test scene.
As shown in~\autoref{tab:heldout},~\ourSystemPR achieves~\(22.89\,\mathrm{dB}\) PSNR,~\(0.743\) SSIM,~\(0.122\) NMSE, and~\(8.89^\circ\) AoA error.
Compared with~WRF-GS-FS, the strongest baseline on all four metrics, it improves PSNR by~\(7.03\,\mathrm{dB}\) and SSIM by~\(0.248\), while reducing NMSE by~\(0.276\) and AoA error by~\(13.36^\circ\).
Under this sparse reference budget, GRaF selects nearby measurements for each query, whereas~\ourSystemPR conditions a reusable anchor field on the full set and applies a target-scene power correction at each query.
Performance remains close to the unseen-variant setting: PSNR changes from~\(23.02\) to~\(22.89\,\mathrm{dB}\), SSIM from~\(0.749\) to~\(0.743\), NMSE from~\(0.119\) to~\(0.122\), and AoA error from~\(8.81^\circ\) to~\(8.89^\circ\).
These results support transfer to excluded categories within the same simulation pipeline, with target-scene references supplying scene-specific information.
Qualitative examples and scene-level breakdowns appear in~Appendices~\ref{app:qualitative:categories} and~\ref{app:breakdown_unseentype}, respectively.

\begin{table*}[t]
\centering
\begin{minipage}[t]{0.49\textwidth}
\centering
\small
\caption{\textit{Generalization to excluded scene variants.} Mean~\(\pm\) standard deviation over~\(5\) folds.}
\label{tab:seen}
\resizebox{\linewidth}{!}{%
\begin{tabular}{lllll}
\toprule
Method & PSNR$\uparrow$ & SSIM$\uparrow$ & NMSE$\downarrow$ & AoA$\downarrow$ \\
\midrule
NeRF$^2$-FS & 13.74\,{\scriptsize$\pm$0.72} & 0.414\,{\scriptsize$\pm$0.032} & 0.636\,{\scriptsize$\pm$0.091} & 24.65\,{\scriptsize$\pm$7.10} \\
NeRF$^2$-TR & 13.00\,{\scriptsize$\pm$0.61} & 0.382\,{\scriptsize$\pm$0.025} & 0.734\,{\scriptsize$\pm$0.101} & 31.20\,{\scriptsize$\pm$7.27} \\
\midrule

GSRF-FS & 12.37\,{\scriptsize$\pm$0.24} & 0.335\,{\scriptsize$\pm$0.004} & 0.717\,{\scriptsize$\pm$0.023} & 51.40\,{\scriptsize$\pm$3.96} \\
GSRF-TR & 12.71\,{\scriptsize$\pm$0.71} & 0.413\,{\scriptsize$\pm$0.026} & 0.766\,{\scriptsize$\pm$0.127} & 43.68\,{\scriptsize$\pm$8.73} \\
\midrule

WRF-GS-FS & 15.86\,{\scriptsize$\pm$0.83} & 0.495\,{\scriptsize$\pm$0.034} & 0.398\,{\scriptsize$\pm$0.053} & 22.23\,{\scriptsize$\pm$5.52} \\
WRF-GS-TR & 13.49\,{\scriptsize$\pm$1.03} & 0.448\,{\scriptsize$\pm$0.027} & 0.624\,{\scriptsize$\pm$0.138} & 34.08\,{\scriptsize$\pm$9.23} \\
\midrule

GRaF & 13.86\,{\scriptsize$\pm$0.38} & 0.441\,{\scriptsize$\pm$0.020} & 0.578\,{\scriptsize$\pm$0.044} & 28.32\,{\scriptsize$\pm$5.63} \\
\midrule
\textbf{\ourSystemPR} & \textbf{23.02}\,{\scriptsize$\pm$2.18} & \textbf{0.749}\,{\scriptsize$\pm$0.055} & \textbf{0.119}\,{\scriptsize$\pm$0.054} & \textbf{8.81}\,{\scriptsize$\pm$4.15} \\
\bottomrule
\end{tabular}
}
\end{minipage}
\hfill
\begin{minipage}[t]{0.49\textwidth}
\centering
\small
\caption{\textit{Generalization to excluded scene categories.} Mean~\(\pm\) standard deviation over~\(7\) folds.}
\label{tab:heldout}
\resizebox{\linewidth}{!}{%
\begin{tabular}{lllll}
\toprule
Method & PSNR$\uparrow$ & SSIM$\uparrow$ & NMSE$\downarrow$ & AoA$\downarrow$ \\
\midrule
NeRF$^2$-FS & 13.75\,{\scriptsize$\pm$0.63} & 0.414\,{\scriptsize$\pm$0.027} & 0.635\,{\scriptsize$\pm$0.078} & 24.65\,{\scriptsize$\pm$7.04} \\
NeRF$^2$-TR & 12.98\,{\scriptsize$\pm$1.15} & 0.379\,{\scriptsize$\pm$0.056} & 0.745\,{\scriptsize$\pm$0.139} & 32.87\,{\scriptsize$\pm$10.89} \\
\midrule
GSRF-FS & 12.39\,{\scriptsize$\pm$0.63} & 0.335\,{\scriptsize$\pm$0.020} & 0.715\,{\scriptsize$\pm$0.094} & 51.29\,{\scriptsize$\pm$6.91} \\
GSRF-TR & 12.77\,{\scriptsize$\pm$0.47} & 0.413\,{\scriptsize$\pm$0.015} & 0.772\,{\scriptsize$\pm$0.100} & 39.73\,{\scriptsize$\pm$3.82} \\
\midrule
WRF-GS-FS & 15.86\,{\scriptsize$\pm$0.81} & 0.495\,{\scriptsize$\pm$0.030} & 0.398\,{\scriptsize$\pm$0.043} & 22.25\,{\scriptsize$\pm$5.19} \\
WRF-GS-TR & 12.75\,{\scriptsize$\pm$1.12} & 0.426\,{\scriptsize$\pm$0.031} & 0.740\,{\scriptsize$\pm$0.202} & 35.04\,{\scriptsize$\pm$5.13} \\
\midrule
GRaF & 13.67\,{\scriptsize$\pm$0.47} & 0.435\,{\scriptsize$\pm$0.025} & 0.607\,{\scriptsize$\pm$0.093} & 29.14\,{\scriptsize$\pm$4.86} \\
\midrule
\textbf{\ourSystemPR} & \textbf{22.89}\,{\scriptsize$\pm$2.72} & \textbf{0.743}\,{\scriptsize$\pm$0.071} & \textbf{0.122}\,{\scriptsize$\pm$0.062} & \textbf{8.89}\,{\scriptsize$\pm$5.43} \\
\bottomrule
\end{tabular}
}
\end{minipage}
\end{table*}

\begin{wrapfigure}{r}{0.4\linewidth}
    \centering
    \includegraphics[width=\linewidth]{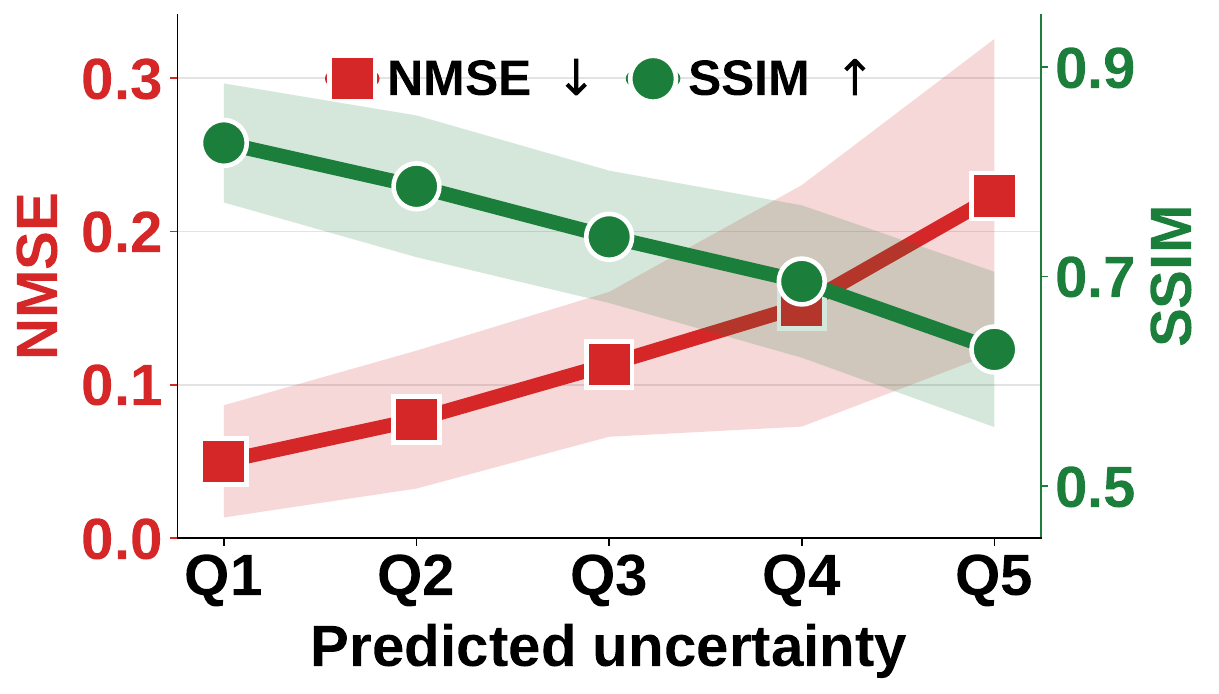}
    \caption{\textit{Reliability analysis.}}
    \label{fig:uncertainty}
\end{wrapfigure}

\vspace{\mylen}
\subsection{Component-Power Reliability}
\label{sec:eval:uncertainty}
\vspace{\mylen}

\ourSystem analytically propagates component-power variance to the spatial spectrum while holding arrival directions fixed.
For each query, the reliability score~\(U_\pi\) is the mean propagated variance over the~\(H\times W\) angular grid.
Because variance scales differ across folds, we group test queries into within-fold quintiles of~\(U_\pi\) before aggregating results over the~\(12\) generalization folds.
As shown in~\autoref{fig:uncertainty}, NMSE increases and SSIM decreases monotonically across these quintiles.
Across~\(68{,}232\) test queries, mean NMSE rises from~\(0.050\) in the lowest-uncertainty quintile to~\(0.223\) in the highest, a~\(4.5\times\) increase, while mean SSIM falls from~\(0.828\) to~\(0.630\).
Mean within-fold Spearman correlations are~\(\rho=0.53\) for NMSE and~\(\rho=-0.55\) for SSIM; the NMSE correlation is positive in every fold, ranging from~\(0.42\) to~\(0.65\).
Therefore, propagated power variance provides a useful query-level ranking of synthesis reliability, though it is not calibrated pixel-level predictive uncertainty.
These results use feed-forward~\ourSystem; the corresponding power-refinement analysis, qualitative examples, and further details appear in~Appendix~\ref{app:reliability_analysis}.

\begin{wraptable}{r}{0.4\linewidth}
\centering
\small
\caption{\textit{Adaptation to unseen arrays.}}
\label{tab:array_adaptation}
\setlength{\tabcolsep}{4pt}
\resizebox{\linewidth}{!}{%
\begin{tabular}{llccc}
\toprule
Array & Array model & PSNR$\uparrow$ & SSIM$\uparrow$ & NMSE$\downarrow$ \\
\midrule
\multirow{2}{*}{\(\mathcal A_1\) (\(3\times3\))}
& Training array & 15.64 & 0.528 & 0.190 \\
& Calibrated     & 20.20 & 0.751 & 0.095 \\
\midrule
\multirow{2}{*}{\(\mathcal A_2\) (\(8\times2\))}
& Training array & 14.96 & 0.458 & 0.327 \\
& Calibrated     & 22.09 & 0.697 & 0.097 \\
\midrule
\multirow{2}{*}{\(\mathcal A_3\) (\(4\times4\) narrow)}
& Training array & 16.47 & 0.553 & 0.179 \\
& Calibrated     & 22.60 & 0.795 & 0.073 \\
\bottomrule
\end{tabular}}
\end{wraptable}

\vspace{\mylen}
\subsection{Receiver Array Adaptation}
\label{sec:eval:array}
\vspace{\mylen}
Separating propagation from array physics allows~\ourSystem to update the analytic response for a new receiver without retraining the propagation model.
When the receiver configuration, including its geometry and pose, is accurately known, this update requires no calibration fit.
We test three arrays unseen during training on the seven held-out scene variants of one fold:~\(\mathcal A_1\), a~\(3\times3\) array at the training spacing;~\(\mathcal A_2\), an~\(8\times2\) array at the same spacing; and~\(\mathcal A_3\), a~\(4\times4\) array with spacing reduced from~\(0.16\,\mathrm m\) to~\(0.12\,\mathrm m\).
Transmitter locations, data splits, and reference sets remain unchanged.
The analytic response uses array geometry, wavelength, position, and orientation from receiver specifications or deployment metadata~\citep{aisg2024als,3gpp38455,3gpp37355,3gpp38901}.

To evaluate imperfect specifications, we start from a perturbed array response and fit an in-plane rotation and two axis scalings using measured spectra.
The pretrained neural modules remain frozen.
This optional calibration uses~\(N_{\mathrm C}=64\) measurements per scene, separate from the~\(M=32\) references and test queries.
In~\autoref{tab:array_adaptation}, retaining the training-array response yields~\(0.458\)--\(0.553\) SSIM, whereas calibration reaches~\(0.697\)--\(0.795\) SSIM and improves PSNR by~\(4.6\)--\(7.1\,\mathrm{dB}\).
Because the imposed errors match the three-parameter calibration model, these results demonstrate correction of the tested specification errors, not robustness to arbitrary receiver mismatch.
Calibration procedures and parameter-error sensitivity are detailed in~Appendix~\ref{app:array_error}.

\vspace{\mylen}
\subsection{Ablation Study}
\label{sec:eval:ablation}
\vspace{\mylen}
\autoref{tab:ablation} examines three design choices of feed-forward~\ourSystem under both generalization settings.
Each variant is retrained on the same data with the same objective.

\begin{table*}[h]
\centering
\small
\vspace{0.05in}
\caption{\textit{Ablation study.} Each variant changes one design choice of~\ourSystem.}
\label{tab:ablation}
\setlength{\tabcolsep}{3.5pt}
\begin{tabular}{lcccccccc}
\toprule
& \multicolumn{4}{c}{Unseen Scene Variants}
& \multicolumn{4}{c}{Unseen Scene Categories} \\
\cmidrule(lr){2-5}
\cmidrule(lr){6-9}
Variant
& PSNR$\uparrow$
& SSIM$\uparrow$
& NMSE$\downarrow$
& AoA$\downarrow$
& PSNR$\uparrow$
& SSIM$\uparrow$
& NMSE$\downarrow$
& AoA$\downarrow$ \\
\midrule
Learned renderer
& 16.09 & 0.544 & 0.369 & 16.27
& 16.00 & 0.542 & 0.396 & 15.42 \\
$-$ bearing-aligned sampling
& 21.07 & 0.697 & 0.159 & 9.95
& 20.79 & 0.709 & 0.149 & 9.67 \\
$-$ query conditioning
& 20.75 & 0.683 & 0.224 & 13.57
& 18.32 & 0.620 & 0.291 & 17.21 \\
\midrule
{\ourSystem}
& {22.95} & {0.740} & {0.121} & {8.40}
& {22.83} & {0.733} & {0.124} & {8.41} \\
\bottomrule
\end{tabular}
\vspace{0.05in}
\end{table*}
\textbf{Analytic Array Physics.}
Replacing the analytic covariance and Bartlett renderer with a learned CNN produces the largest degradation.
For unseen variants, PSNR falls from~\(22.95\) to~\(16.09\,\mathrm{dB}\), SSIM from~\(0.740\) to~\(0.544\), and AoA error rises from~\(8.40^\circ\) to~\(16.27^\circ\); the same pattern holds for unseen categories.
The analytic renderer applies the specified array response to each predicted component, whereas the CNN must learn this mapping from training examples.
The comparison therefore supports retaining the known array mapping while learning scene-dependent propagation.

\textbf{Bearing-Aligned Sampling.}
Removing bearing-aligned sampling reduces PSNR by~\(1.88\) and~\(2.04\,\mathrm{dB}\) in the two settings and increases AoA error by~\(1.55^\circ\) and~\(1.26^\circ\), respectively.
SSIM and NMSE also worsen in both settings.
These results suggest that set-level reference aggregation captures broad scene context, while bearing-aligned sampling supplies complementary directional evidence to the anchors and improves query-specific propagation estimates.

\textbf{Query Conditioning.}
Removing query conditioning reduces PSNR by~\(2.20\,\mathrm{dB}\) for unseen variants and~\(4.51\,\mathrm{dB}\) for unseen categories.
For unseen categories, SSIM also falls from~\(0.733\) to~\(0.620\), NMSE rises from~\(0.124\) to~\(0.291\), and AoA error increases from~\(8.41^\circ\) to~\(17.21^\circ\).
These losses exceed those from removing bearing-aligned sampling, particularly when the target category is excluded from training.
The result highlights the importance of specializing the shared scene representation to each query transmitter location when estimating propagation components.

\vspace{\mylen}
\subsection{Case Study: Beam Management with Synthesized Spatial Spectra}
\label{sec:eval:beam}
\vspace{\mylen}

Beam management aims to find a high-power beam from a codebook, but probing every candidate incurs measurement overhead~\citep{xue2024beammanagement,alkhateeb2017beamassociation}.
For each query location,~\ourSystem uses its synthesized spatial spectrum to rank the candidate beams.
We first evaluate the loss from probing only the highest-ranked beams.
We then test whether probing more beams for less reliable predictions reduces loss at the same average search budget.

\textbf{Spectrum-Based Beam Search.}
For a~\(B=64\)-beam codebook~\(\mathcal B=\left\{\Omega_b\right\}_{b=1}^{B}\), spectrum-based search selects the~\(K_{\mathrm b}\) beams with the highest synthesized values~\(\widehat S_\pi\left(\Omega_b\right)\), rather than searching the full codebook.
Let~\(\widehat{\mathcal B}_{\pi,K_{\mathrm b}}\) denote the selected set.
Using target spatial power~\(P_\pi\left(\Omega\right)=S_\pi\left(\Omega\right)^2\) as the beam-power proxy, we measure loss relative to exhaustive search:
\begin{equation}
L_{\pi,K_{\mathrm b}}
=
10\log_{10}\left(
\frac{
\max_{\Omega_b\in\mathcal B}P_\pi\left(\Omega_b\right)
}{
\max_{\Omega_b\in\widehat{\mathcal B}_{\pi,K_{\mathrm b}}}P_\pi\left(\Omega_b\right)
}
\right).
\label{eq:beam_loss}
\end{equation}
A loss of~\(0\,\mathrm{dB}\) means the selected set contains a beam maximizing the target spatial power.
We evaluate~\(K_{\mathrm b}\in\left\{1,2,4,8,16\right\}\) with the same codebook, queries, and reference budget for all methods.

\textbf{Reliability-Aware Beam Search.}
A fixed~\(K_{\mathrm b}\) assigns every query the same search budget despite differences in synthesis reliability.
We rank queries within each fold by~\(U_\pi\) from~\S\ref{sec:eval:uncertainty} and divide them into five groups.
Groups with higher propagated variance receive more beams while preserving a target average budget~\(\overline K_{\mathrm b}\).
The allocation is chosen on validation data and fixed for testing, so fixed and adaptive policies use the same average number of beam probes.

\begin{wrapfigure}{r}{0.41\linewidth}
\centering
\includegraphics[width=\linewidth]{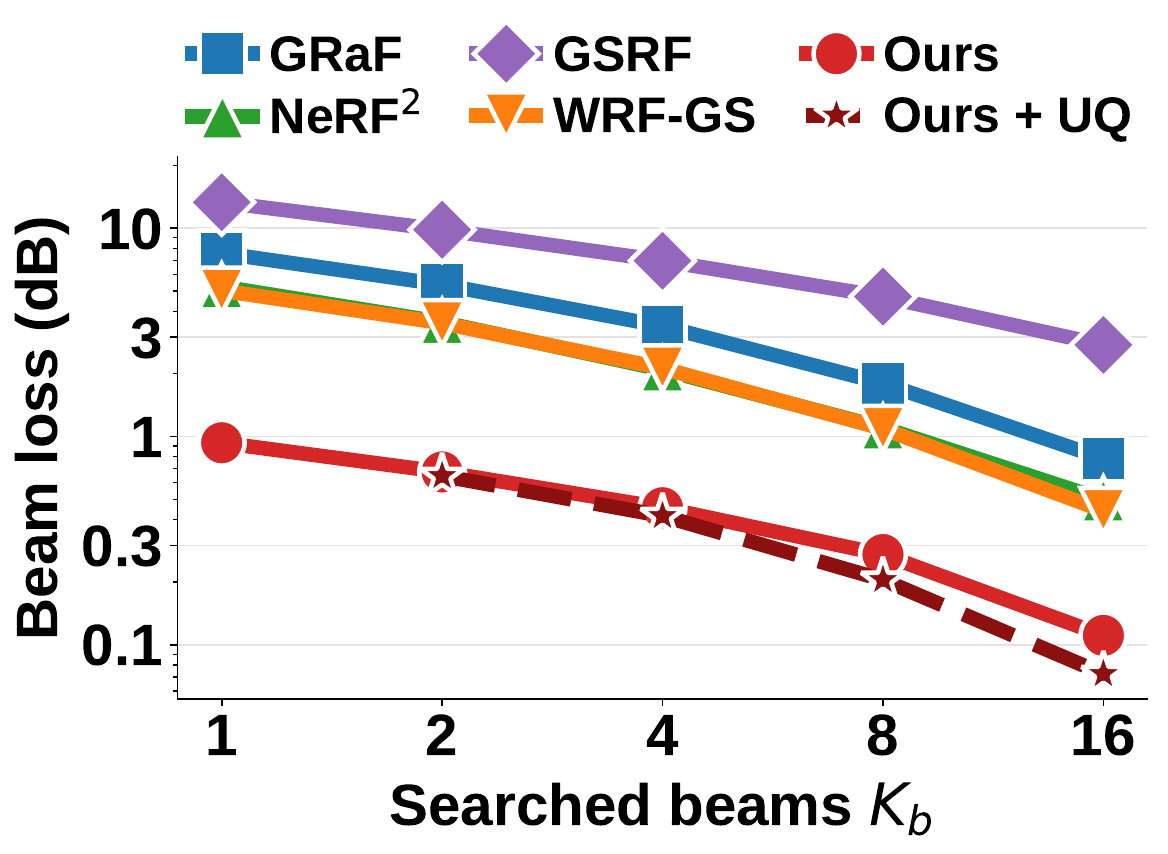}
\caption{\textit{Beam management study.}
}
\label{fig:beam}
\end{wrapfigure}

In~\autoref{fig:beam},~\ourSystem achieves~\(0.93\,\mathrm{dB}\) beam loss with one selected beam, compared with~\(5.02\,\mathrm{dB}\) for the strongest baseline.
It selects a beam maximizing target spatial power in~\(63.7\%\) of queries, versus~\(21.8\%\) for that baseline.
With~\(K_{\mathrm b}=4\),~\ourSystem reaches~\(0.46\,\mathrm{dB}\), comparable to the baseline's~\(0.44\,\mathrm{dB}\) at~\(K_{\mathrm b}=16\), using one quarter as many beams.
For the unseen-variant generalization setting, reliability-aware allocation further reduces beam loss from~\(0.271\) to~\(0.205\,\mathrm{dB}\) at~\(\overline K_{\mathrm b}=8\) and from~\(0.111\) to~\(0.072\,\mathrm{dB}\) at~\(\overline K_{\mathrm b}=16\).
For unseen scene categories, validation selects a nearly uniform allocation and the adaptive gain disappears.
Therefore, spectrum synthesis reduces the beam budget needed for low loss under this proxy, while reliability-aware allocation helps when~\(U_\pi\) meaningfully ranks query difficulty.
A localization case study using the same synthesized spectra appears in~Appendix~\ref{app:localization}.

\vspace{\mylen}
\subsection{Generalization to Real Measurements}
\label{sec:eval:real}
\vspace{\mylen}

\begin{wraptable}{r}{0.41\linewidth}
\centering
\small
\caption{\textit{Simulation-to-real transfer on the measured-scene dataset~(S36).}}
\label{tab:real_matched}
\setlength{\tabcolsep}{4pt}
\resizebox{\linewidth}{!}{%
\begin{tabular}{lcccc}
\toprule
Method & PSNR$\uparrow$ & SSIM$\uparrow$ & NMSE$\downarrow$ & AoA$\downarrow$ \\
\midrule
NeRF$^2$       & 12.94 & 0.457 & 0.287 & 16.24 \\
GSRF           & 16.85 & 0.587 & 0.145 & 12.09 \\
WRF-GS+        & 16.96 & 0.592 & 0.140 & 9.26 \\
GRaF           & 14.24 & 0.510 & 0.219 & \textbf{7.87} \\
\midrule
\ourSystemCAL   & 16.74 & 0.620 & 0.125 & 8.12 \\
\ourSystemPRCAL & \textbf{17.45} & \textbf{0.653} & \textbf{0.118} & 7.88 \\
\bottomrule
\end{tabular}}
\end{wraptable}

We compare~\ourSystem trained on~\(35\) simulated scenes with baselines on measured scene~S36 under a matched budget of~\(96\) target-domain spectra.
\ourSystemCAL and~\ourSystemPRCAL use~\(M=32\) references and~\(N_{\mathrm C}=64\) disjoint calibration spectra;~\ourSystemPRCAL reuses both for power refinement.
Calibration improves PSNR by~\(1.95\,\mathrm{dB}\) over uncalibrated~\ourSystem, and refinement adds~\(0.71\,\mathrm{dB}\) without further measurements.
In~\autoref{tab:real_matched},~\ourSystemPRCAL leads in PSNR, SSIM, and NMSE, while GRaF has a marginally lower AoA error.
The smaller PSNR margin partly reflects WRF-GS+ fitting~S36 directly, whereas~\ourSystem transfers a simulation-trained propagation prior.
Held-out simulations share the material library and propagation pipeline; S36 may have unfamiliar reflected and diffuse paths and receiver mismatch.
Sparse references may recover dominant directions but miss weaker paths, reducing spectral fidelity while leaving AoA nearly unchanged.
Calibration partly corrects receiver mismatch and refinement adjusts powers, but neither recovers missing directions.
Further analysis appears in~Appendix~\ref{app:real}.

Additional experiments examine reference evidence~(\S\ref{app:reference_analysis}), representation and renderer choices~(\S\ref{app:representation_choices}), deployment efficiency~(\S\ref{app:deployment_efficiency}), and material-property shifts~(\S\ref{app:physical_stress}).

\vspace{\mylen}
\section{Discussion and Conclusion}
\label{sec:conclusion}
\vspace{\mylen}

We introduce~\ourSystem for cross-scene RF spatial spectrum synthesis from sparse target-scene measurements.
Unordered references condition a canonical anchor field whose predicted propagation components are rendered through analytic receiver-array physics.
Feed-forward~\ourSystem requires no scene-specific optimization, while~\ourSystemPR optionally fits a correction limited to component powers using the same reference set.
Across~\(35\) simulated scenes,~\ourSystemPR exceeds the strongest baseline by about~\(7\,\mathrm{dB}\) PSNR on both unseen scene variants and excluded categories.

Three limitations remain.
First, the performance gain narrows under simulation-to-real transfer, where physical receiver responses can differ from the nominal array model.
Extending analytic calibration beyond geometry to account for element-wise gain and phase errors may reduce this gap.
Second, the method relies on spatial-spectrum references; extending it to other measurement modalities, such as received signal strength indicator~(RSSI) or channel state information~(CSI), may require modality-specific representations and reference-acquisition strategies.
Third,~\ourSystem assumes static scenes; temporal representations and online reference updates could support dynamic scenes.

\section*{Acknowledgments}
The research reported in this paper was sponsored in part by the DEVCOM Army Research Laboratory~(award \# W911NF1720196) and the National Science Foundation~(awards \# 2325956 and 2525614).
Wan Du and Duaa Nakshbandi were funded in part by the National Science Foundation through grant~CCSS-2525613.
The views and conclusions contained in this document are those of the authors and should not be interpreted as representing the official policies, either expressed or implied, of the funding agencies.


\newpage

\setcounter{tocdepth}{2}
\renewcommand{\contentsname}{Appendix Contents}

{\small\tableofcontents}

\clearpage

\addtocontents{toc}{\protect\setcounter{tocdepth}{2}}

\appendix

\section{Analytic Component Power Uncertainty}
\label{app:Power}

We derive how uncertainty in the predicted component powers propagates through the analytic decoder to the reported spatial spectrum.
We first exploit the linearity of~Bartlett power in the component powers to obtain its mean and variance exactly.
We then analyze the square root transformation, provide a rigorous bound on the resulting magnitude uncertainty, and derive the first order delta approximation used by~\ourSystem.
Finally, we propagate this approximation through the fixed normalization used by the model and obtain the reported variance proxy after global calibration.

\textbf{Exact Power Domain Moments.}
Let
\begin{equation*}
\mathbf a_n
=
\mathbf a\left(
\Omega_n
\right),
\qquad
n=1,\ldots,K.
\end{equation*}
The predicted arrival directions are treated as deterministic.
We include the~LOS component as~\(n=0\), with
\begin{equation*}
\mathbf a_0
=
\mathbf a\left(
\Omega_{\mathrm{los}}
\right),
\qquad
m_0
=
m_{\mathrm{los}},
\qquad
v_0
=
0.
\end{equation*}

Let~\(p_n\ge0\) almost surely denote the random power of component~\(n\), with
\begin{equation*}
\mathbb E\left[
p_n
\right]
=
m_n,
\qquad
\operatorname{Var}\left[
p_n
\right]
=
v_n
<
\infty.
\end{equation*}
We assume that the component powers are pairwise uncorrelated.
The deterministic~LOS contribution is represented by~\(p_0=m_{\mathrm{los}}\) almost surely.

Define
\begin{equation*}
\chi_n\left(
\Omega
\right)
=
\left|
\mathbf a\left(
\Omega
\right)^{\mathsf H}
\mathbf a_n
\right|^2.
\end{equation*}
The corresponding random~Bartlett power is
\begin{equation}
P_\pi^{\mathrm{rand}}
\left(
\Omega
\right)
=
g\left(
\Omega
\right)
\sum_{n=0}^{K}
p_n
\chi_n\left(
\Omega
\right).
\label{eq:app-random-power}
\end{equation}
Since~\(g\left(\Omega\right)\ge0\),~\(p_n\ge0\), and~\(\chi_n\left(\Omega\right)\ge0\), we have~\(P_\pi^{\mathrm{rand}}\left(\Omega\right)\ge0\) almost surely.

\begin{proposition}[Exact Power Domain Moments]
\label{prop:moments}
The mean and variance of~\eqref{eq:app-random-power} are
\begin{equation}
\begin{aligned}
\mu_P\left(
\Omega
\right)
&\coloneqq
\mathbb E\left[
P_\pi^{\mathrm{rand}}
\left(
\Omega
\right)
\right]
=
g\left(
\Omega
\right)
\sum_{n=0}^{K}
m_n
\chi_n\left(
\Omega
\right),
\\
\sigma_P^2\left(
\Omega
\right)
&\coloneqq
\operatorname{Var}\left[
P_\pi^{\mathrm{rand}}
\left(
\Omega
\right)
\right]
=
g\left(
\Omega
\right)^2
\sum_{n=0}^{K}
v_n
\chi_n\left(
\Omega
\right)^2.
\end{aligned}
\label{eq:app-power-moments}
\end{equation}
\end{proposition}

\begin{proof}
The first identity in~\eqref{eq:app-power-moments} follows directly from linearity of expectation and~\(\mathbb E\left[p_n\right]=m_n\).
Since the component powers are pairwise uncorrelated,
\begin{equation*}
\operatorname{Var}\left[
\sum_{n=0}^{K}
p_n
\chi_n\left(
\Omega
\right)
\right]
=
\sum_{n=0}^{K}
v_n
\chi_n\left(
\Omega
\right)^2,
\end{equation*}
which gives the second identity in~\eqref{eq:app-power-moments}.
\end{proof}

The deterministic decoder substitutes the mean powers~\(m_n\) into the same analytic renderer.
Therefore,
\begin{equation*}
\widetilde S_\pi
\left(
\Omega
\right)
=
\sqrt{
\mu_P\left(
\Omega
\right)
}.
\end{equation*}
Thus, introducing the component power uncertainty leaves the deterministic point prediction unchanged.

\textbf{Magnitude Domain Uncertainty.}
Define
\begin{equation*}
S_\pi^{\mathrm{rand}}
\left(
\Omega
\right)
=
\sqrt{
P_\pi^{\mathrm{rand}}
\left(
\Omega
\right)
}.
\end{equation*}
The deterministic prediction~\(\sqrt{\mu_P}\) is generally not equal to~\(\mathbb E\left[S_\pi^{\mathrm{rand}}\right]\).

\begin{lemma}[Magnitude Uncertainty]
\label{lem:magnitude_bounds}
At any angular location with~\(\mu_P>0\),
\begin{equation*}
\mathbb E\left[
S_\pi^{\mathrm{rand}}
\right]
\le
\sqrt{
\mu_P
},
\qquad
\operatorname{Var}\left[
S_\pi^{\mathrm{rand}}
\right]
=
\mu_P
-
\mathbb E\left[
S_\pi^{\mathrm{rand}}
\right]^2,
\end{equation*}
and
\begin{equation}
\operatorname{Var}\left[
S_\pi^{\mathrm{rand}}
\right]
\le
\frac{
\sigma_P^2
}{
\mu_P
}.
\label{eq:app-magnitude-var-bound}
\end{equation}
\end{lemma}

\begin{proof}
The first inequality follows from the concavity of the square root and~Jensen's inequality.
Since~\(\left(S_\pi^{\mathrm{rand}}\right)^2=P_\pi^{\mathrm{rand}}\),
\begin{equation*}
\operatorname{Var}\left[
S_\pi^{\mathrm{rand}}
\right]
=
\mathbb E\left[
P_\pi^{\mathrm{rand}}
\right]
-
\mathbb E\left[
S_\pi^{\mathrm{rand}}
\right]^2.
\end{equation*}
For the variance bound,
\begin{align*}
\operatorname{Var}\left[
S_\pi^{\mathrm{rand}}
\right]
&\le
\mathbb E\left[
\left(
\sqrt{
P_\pi^{\mathrm{rand}}
}
-
\sqrt{
\mu_P
}
\right)^2
\right]
\\
&=
\mathbb E\left[
\frac{
\left(
P_\pi^{\mathrm{rand}}
-
\mu_P
\right)^2
}{
\left(
\sqrt{
P_\pi^{\mathrm{rand}}
}
+
\sqrt{
\mu_P
}
\right)^2
}
\right]
\\
&\le
\frac{
\sigma_P^2
}{
\mu_P
},
\end{align*}
which gives~\eqref{eq:app-magnitude-var-bound}.
\end{proof}

If~\(\mu_P\left(\Omega\right)=0\), nonnegativity and zero expectation imply~\(P_\pi^{\mathrm{rand}}\left(\Omega\right)=0\) almost surely.
Hence the corresponding component induced magnitude variance is also zero.

\textbf{Delta Approximation.}
The exact magnitude variance cannot in general be determined from only~\(\mu_P\) and~\(\sigma_P^2\).
For a lightweight uncertainty proxy,~\ourSystem applies the first order delta method to~\(\varphi\left(P\right)=\sqrt{P}\).
For~\(\mu_P\left(\Omega\right)>0\),
\begin{equation}
\sigma_{S,\Delta}^2
\left(
\Omega
\right)
\coloneqq
\frac{
\sigma_P^2\left(
\Omega
\right)
}{
4
\mu_P\left(
\Omega
\right)
}.
\label{eq:app-delta}
\end{equation}
This is the approximation used in~\eqref{eq:spectrum_uncertainty}.
It is expected to be accurate when
\begin{equation*}
\frac{
\sigma_P^2\left(
\Omega
\right)
}{
\mu_P\left(
\Omega
\right)^2
}
\ll
1,
\end{equation*}
but it is not a distribution free upper or lower bound on the true magnitude variance.
In particular,~\eqref{eq:app-delta} equals one quarter of the upper bound expression in~\eqref{eq:app-magnitude-var-bound}, which does not imply that the true variance differs from the approximation by a constant factor.
For~\(\mu_P\left(\Omega\right)=0\), we set~\(\sigma_{S,\Delta}^2\left(\Omega\right)=0\).

\textbf{Normalization and Reported Uncertainty.}
The deterministic prediction is compared with measurements after per spectrum min max normalization.
Define
\begin{equation*}
\ell_\pi
=
\min_{\Omega}
\widetilde S_\pi
\left(
\Omega
\right),
\qquad
r_\pi
=
\max_{\Omega}
\widetilde S_\pi
\left(
\Omega
\right)
-
\ell_\pi,
\qquad
r_\pi
>
0.
\end{equation*}
The reported deterministic spectrum is
\begin{equation*}
\widehat S_\pi
\left(
\Omega
\right)
=
\frac{
\widetilde S_\pi
\left(
\Omega
\right)
-
\ell_\pi
}{
r_\pi
}.
\end{equation*}

Exactly normalizing every realization of~\(S_\pi^{\mathrm{rand}}\) would make its minimum and maximum random and would couple all angular locations.
Instead, the uncertainty propagation holds~\(\ell_\pi\) and~\(r_\pi\) fixed at the deterministic prediction.
Under this fixed affine approximation,
\begin{equation}
\sigma_{\mathrm{norm},\Delta}^2
\left(
\Omega
\right)
=
\frac{
\sigma_P^2\left(
\Omega
\right)
}{
4
\mu_P\left(
\Omega
\right)
r_\pi^2
},
\qquad
\mu_P\left(
\Omega
\right)
>
0.
\label{eq:app-normalized-var}
\end{equation}
For~\(\mu_P\left(\Omega\right)=0\), we set~\(\sigma_{\mathrm{norm},\Delta}^2\left(\Omega\right)=0\).

The final reported variance proxy is
\begin{equation}
\widehat\sigma_\pi^2
\left(
\Omega
\right)
=
\tau_{\mathrm u}^2
\sigma_{\mathrm{norm},\Delta}^2
\left(
\Omega
\right)
+
\varepsilon_0.
\label{eq:app-reported-var}
\end{equation}
The temperature~\(\tau_{\mathrm u}\) rescales the uncertainty proxy, while~\(\varepsilon_0\) provides a residual variance floor.
Both are fitted by the heteroscedastic objective in~\S\ref{sec:training}, with their gradients isolated from the deterministic mean branch.
The reported quantity~\(\widehat\sigma_\pi^2\) is a component power uncertainty proxy rather than the complete predictive variance.
It is conditional on the predicted component means and arrival directions.
It does not account for errors in those means or directions, scene state estimation error, renderer mismatch, or distribution shift.

\section{Observability from Reference Spectra}
\label{app:Information}

We study how many reference spectra are needed to constrain the scene-dependent propagation state.
Because each reference depends on both the scene and its unknown transmitter pose, only scene variations that cannot be explained by pose variation provide effective scene information.
We use this observation to bound the information provided by each reference and characterize how overlap among references affects the required reference budget.

\textbf{Observation Model and Effective Scene Information.}
Let~\(\mathbf p\in\mathbb R^{d_{\mathrm p}}\) denote a local parameterization of the scene-dependent, query-independent propagation state, and~\(\boldsymbol\xi_i\in\Xi\subset\mathbb R^{d_{\pi}}\) the unknown physical pose that generated reference~\(i\).
The vectorized reference spectrum is
\begin{equation*}
\mathbf y_i
=
\mathcal F\left(
\mathbf p,
\boldsymbol\xi_i
\right)
\in
\mathbb R^{HW}.
\end{equation*}
We assume that~\(\mathcal F\) is differentiable at the operating points considered and that the reference poses can vary independently.
For the per spectrum min max normalization, this excludes zero range spectra and points where the active extrema change under an infinitesimal perturbation.

Around~\(\left(\mathbf p_0,\boldsymbol\xi_i\right)\),
\begin{align}
&
\mathcal F\left(
\mathbf p_0+\Delta\mathbf p,
\boldsymbol\xi_i+\Delta\boldsymbol\xi_i
\right)
-
\mathcal F\left(
\mathbf p_0,
\boldsymbol\xi_i
\right)
\nonumber
\\
&\qquad
=
\mathbf J_{\mathrm p,i}
\Delta\mathbf p
+
\mathbf J_{\xi,i}
\Delta\boldsymbol\xi_i
+
o\left(
\left\|
\begin{bmatrix}
\Delta\mathbf p
\\
\Delta\boldsymbol\xi_i
\end{bmatrix}
\right\|_2
\right),
\label{eq:app_ref_linearization}
\end{align}
where
\begin{equation*}
\mathbf J_{\mathrm p,i}
=
D_{\mathbf p}
\mathcal F\left(
\mathbf p_0,
\boldsymbol\xi_i
\right),
\qquad
\mathbf J_{\xi,i}
=
D_{\boldsymbol\xi}
\mathcal F\left(
\mathbf p_0,
\boldsymbol\xi_i
\right).
\end{equation*}

A scene perturbation is informative only through the part of its spectral effect that cannot be reproduced by changing the unknown reference pose.
Define
\begin{equation}
\boldsymbol\Pi_i^{\perp}
=
\mathbf I
-
\mathbf J_{\xi,i}
\mathbf J_{\xi,i}^{\dagger},
\qquad
\overline{\mathbf J}_i
=
\boldsymbol\Pi_i^{\perp}
\mathbf J_{\mathrm p,i},
\label{eq:app_effective_jacobian}
\end{equation}
where~\(\left(\cdot\right)^\dagger\) denotes the~Moore--Penrose pseudoinverse.
The matrix~\(\boldsymbol\Pi_i^{\perp}\) projects onto the orthogonal complement of~\(\operatorname{col}\mathbf J_{\xi,i}\), so~\(\overline{\mathbf J}_i\) retains only scene variations that cannot be absorbed by an infinitesimal pose change.

\begin{lemma}[Effective Scene Information]
\label{lem:effective_scene_information}
A scene perturbation~\(\Delta\mathbf p\) is indistinguishable to first order from a perturbation of the unknown reference pose if and only if
\begin{equation*}
\overline{\mathbf J}_i
\Delta\mathbf p
=
\mathbf 0.
\end{equation*}
\end{lemma}

\begin{proof}
By~\eqref{eq:app_ref_linearization}, the effect of~\(\Delta\mathbf p\) can be canceled by some~\(\Delta\boldsymbol\xi_i\) if and only if
\begin{equation*}
\mathbf J_{\mathrm p,i}
\Delta\mathbf p
\in
\operatorname{col}
\mathbf J_{\xi,i}.
\end{equation*}
This is equivalent to
\begin{equation*}
\boldsymbol\Pi_i^{\perp}
\mathbf J_{\mathrm p,i}
\Delta\mathbf p
=
\mathbf 0,
\end{equation*}
which gives the result by~\eqref{eq:app_effective_jacobian}.
\end{proof}

For a reference configuration~\(\mathcal C=\left\{\boldsymbol\xi_i\right\}_{i=1}^{M}\), define the stacked effective Jacobian
\begin{equation*}
\overline{\mathbf J}_{\mathcal C}
=
\begin{bmatrix}
\overline{\mathbf J}_1
\\
\vdots
\\
\overline{\mathbf J}_M
\end{bmatrix}.
\end{equation*}
Permuting the references only permutes its block rows, so the analysis is invariant to their ordering.

\textbf{Observable Dimension and Reference Sufficiency.}
At~\(\mathbf p_0\), define the total locally observable scene subspace and its dimension as
\begin{equation}
\mathcal O_{\star}\left(\mathbf p_0\right)
=
\operatorname{span}
\left\{
\operatorname{row}
\overline{\mathbf J}
\left(
\mathbf p_0,
\boldsymbol\xi
\right)
:
\boldsymbol\xi\in\Xi
\right\},
\qquad
r_{\star}\left(\mathbf p_0\right)
=
\dim
\mathcal O_{\star}\left(\mathbf p_0\right).
\label{eq:app_observable_dimension}
\end{equation}
Thus,~\(r_{\star}\) is the number of independent scene directions that can be constrained locally by the admissible reference measurements after accounting for unknown reference pose.

A finite reference configuration~\(\mathcal C\) is locally sufficient when
\begin{equation*}
\operatorname{rank}
\overline{\mathbf J}_{\mathcal C}
=
r_{\star}\left(\mathbf p_0\right).
\end{equation*}
The minimum sufficient reference budget is therefore
\begin{equation}
M^{\star}\left(\mathbf p_0\right)
=
\min_{
\substack{
\mathcal C\subset\Xi\\
\left|\mathcal C\right|<\infty
}
}
\left\{
\left|\mathcal C\right|
:
\operatorname{rank}
\overline{\mathbf J}_{\mathcal C}
=
r_{\star}\left(\mathbf p_0\right)
\right\}.
\label{eq:app_minimum_budget}
\end{equation}
Since~\(\mathcal O_{\star}\left(\mathbf p_0\right)\) is finite dimensional, a finite collection of admissible reference row spaces can span it whenever~\(r_{\star}>0\).

\textbf{Reference Capacity and Minimum Budget.}
For the spatial spectra used by~\ourSystem, each reference is formed from the per element phases of an~\(N\) element array.
Let
\begin{equation*}
\widetilde{\mathbf x}
\left(
\boldsymbol\phi
\right)
=
\begin{bmatrix}
e^{j\phi_1}
&
\cdots
&
e^{j\phi_N}
\end{bmatrix}^{\mathsf T},
\qquad
\boldsymbol\phi\in\mathbb R^N.
\end{equation*}
Its~Bartlett power is, up to a fixed scale,
\begin{equation}
P_{\boldsymbol\phi}
\left(
\Omega
\right)
=
g\left(
\Omega
\right)
\left|
\mathbf a\left(
\Omega
\right)^{\mathsf H}
\widetilde{\mathbf x}
\left(
\boldsymbol\phi
\right)
\right|^2.
\label{eq:app_phase_bartlett}
\end{equation}
Let~\(\mathcal H\left(\boldsymbol\phi\right)\) denote the differentiable map from the phase vector to the vectorized spatial spectrum, including the square root and per spectrum normalization.
The reference map therefore factors locally as
\begin{equation*}
\mathcal F
\left(
\mathbf p,
\boldsymbol\xi
\right)
=
\mathcal H
\left(
\boldsymbol\phi
\left(
\mathbf p,
\boldsymbol\xi
\right)
\right).
\end{equation*}

\begin{proof}[Proof of Proposition~\ref{prop:reference_budget}]
For any scalar phase~\(\alpha\),
\begin{equation*}
\widetilde{\mathbf x}
\left(
\boldsymbol\phi
+
\alpha\mathbf 1
\right)
=
e^{j\alpha}
\widetilde{\mathbf x}
\left(
\boldsymbol\phi
\right).
\end{equation*}
The common phase factor disappears from the magnitude in~\eqref{eq:app_phase_bartlett}, so
\begin{equation*}
\mathcal H
\left(
\boldsymbol\phi
+
\alpha\mathbf 1
\right)
=
\mathcal H
\left(
\boldsymbol\phi
\right).
\end{equation*}
Differentiating with respect to~\(\alpha\) at~\(\alpha=0\) gives
\begin{equation*}
D\mathcal H
\left(
\boldsymbol\phi
\right)
\mathbf 1
=
\mathbf 0.
\end{equation*}
Hence~\(D\mathcal H\left(\boldsymbol\phi\right)\) has a nontrivial null direction and
\begin{equation*}
\operatorname{rank}
D\mathcal H
\left(
\boldsymbol\phi
\right)
\le
N-1.
\end{equation*}

By the chain rule,
\begin{equation*}
\mathbf J_{\mathrm p,i}
=
D\mathcal H
\left(
\boldsymbol\phi_i
\right)
D_{\mathbf p}
\boldsymbol\phi
\left(
\mathbf p_0,
\boldsymbol\xi_i
\right),
\end{equation*}
and therefore
\begin{equation*}
\operatorname{rank}
\mathbf J_{\mathrm p,i}
\le
N-1.
\end{equation*}
Since~\(\overline{\mathbf J}_i=\boldsymbol\Pi_i^{\perp}\mathbf J_{\mathrm p,i}\) and left multiplication cannot increase rank,
\begin{equation}
\operatorname{rank}
\overline{\mathbf J}_i
\le
N-1.
\label{eq:app_reference_capacity}
\end{equation}

For any configuration~\(\mathcal C\) containing~\(M\) references,
\begin{equation}
\operatorname{rank}
\overline{\mathbf J}_{\mathcal C}
\le
\sum_{i=1}^{M}
\operatorname{rank}
\overline{\mathbf J}_i
\le
M\left(N-1\right).
\label{eq:app_reference_set_capacity}
\end{equation}
If~\(\mathcal C\) is locally sufficient, then
\begin{equation*}
r_{\star}\left(\mathbf p_0\right)
=
\operatorname{rank}
\overline{\mathbf J}_{\mathcal C}
\le
M\left(N-1\right).
\end{equation*}
Taking the minimum over all locally sufficient configurations gives
\begin{equation}
M^{\star}\left(\mathbf p_0\right)
\ge
\left\lceil
\frac{
r_{\star}\left(\mathbf p_0\right)
}{
N-1
}
\right\rceil,
\label{eq:app_reference_budget_bound}
\end{equation}
which proves Proposition~\ref{prop:reference_budget}.
\end{proof}

\textbf{Reference Complementarity.}
The bound in~\eqref{eq:app_reference_budget_bound} is necessary but not sufficient because different references may constrain overlapping scene directions.
For an existing configuration~\(\mathcal C\) and a candidate reference at~\(\boldsymbol\xi\), define its marginal observable dimension as
\begin{equation}
\Delta r
\left(
\boldsymbol\xi
\mid
\mathcal C
\right)
=
\operatorname{rank}
\begin{bmatrix}
\overline{\mathbf J}_{\mathcal C}
\\
\overline{\mathbf J}
\left(
\mathbf p_0,
\boldsymbol\xi
\right)
\end{bmatrix}
-
\operatorname{rank}
\overline{\mathbf J}_{\mathcal C}.
\label{eq:app_marginal_rank}
\end{equation}
By the dimension formula for sums of subspaces,
\begin{align*}
\Delta r
\left(
\boldsymbol\xi
\mid
\mathcal C
\right)
&=
\operatorname{rank}
\overline{\mathbf J}
\left(
\mathbf p_0,
\boldsymbol\xi
\right)
\\
&\quad
-
\dim
\left(
\operatorname{row}
\overline{\mathbf J}
\left(
\mathbf p_0,
\boldsymbol\xi
\right)
\cap
\operatorname{row}
\overline{\mathbf J}_{\mathcal C}
\right).
\end{align*}
Hence
\begin{equation*}
0
\le
\Delta r
\left(
\boldsymbol\xi
\mid
\mathcal C
\right)
\le
N-1.
\end{equation*}
A reference with~\(\Delta r=0\) adds no new first order observable direction, while a larger~\(\Delta r\) indicates greater complementarity with the existing references.
Thus, reference utility depends on overlap between effective observable subspaces rather than reference count alone.

\textbf{Scope.}
The analysis is local around~\(\mathbf p_0\) and characterizes first order observability rather than global uniqueness of the nonlinear propagation state.
It is also noiseless, so a reference with zero marginal rank may still improve statistical precision under measurement noise.
The nuisance projection accounts for ambiguity caused by unknown reference poses and is independent of the estimator used to infer~\(\mathbf p\).
If the reference poses are observed, the known pose case follows by setting~\(\mathbf J_{\xi,i}=\mathbf 0\), giving~\(\overline{\mathbf J}_i=\mathbf J_{\mathrm p,i}\).
Finally,~\(M^{\star}\) is a local identifiability budget for the reference measurements, not a prediction of the number of references required by an amortized learned estimator such as~\ourSystem to achieve a particular reconstruction accuracy.

\section{Propagation-to-Spectrum Error Bound}
\label{app:errobound}

We prove Proposition~\ref{prop:rendering_error}, which characterizes how propagation estimation errors affect the rendered spatial spectrum.
The proof separates representation mismatch from component power and direction errors, whose sensitivities are determined by the fixed receive array.

\textbf{Setup and Representation Error.}
Let~\(P\) denote the target power spectrum and~\(\widehat P=\mathcal B\left(\widehat{\mathbf z}\right)\) the spectrum rendered from the estimated propagation parameters.
For
\begin{equation*}
\mathbf z
=
\left\{
\left(
\Omega_n,m_n
\right)
\right\}_{n=0}^{K},
\end{equation*}
where~\(n=0\) denotes the~LOS component, the analytic renderer is
\begin{equation}
\mathcal B
\left(
\mathbf z
\right)
\left(
\Omega
\right)
=
g\left(
\Omega
\right)
\sum_{n=0}^{K}
m_n
\left|
\mathbf a\left(
\Omega
\right)^{\mathsf H}
\mathbf a\left(
\Omega_n
\right)
\right|^2.
\label{eq:app_renderer}
\end{equation}
The~LOS direction is fixed by the query geometry, while its power is allowed to vary in the admissible renderer class.

Let~\(\mathcal Z_K\) be a compact set of admissible renderer states with bounded nonnegative powers and admissible component directions.
Choose
\begin{equation}
\mathbf z^\star
\in
\arg\min_{\mathbf z\in\mathcal Z_K}
\left\|
P-\mathcal B\left(\mathbf z\right)
\right\|_\infty,
\qquad
P^\star
=
\mathcal B\left(\mathbf z^\star\right),
\qquad
\varepsilon_{\mathrm{phys}}
=
\left\|
P-P^\star
\right\|_\infty.
\label{eq:app_physics_error}
\end{equation}
Such a minimizer exists because~\(\mathcal Z_K\) is compact and~\(\mathcal B\) is continuous on the finite angular grid.
The term~\(\varepsilon_{\mathrm{phys}}\) captures spectrum structure that cannot be represented by the analytic renderer, including finite component approximation and measurement effects outside the renderer model.

\textbf{Array Response Sensitivity.}
Define the unit power~Bartlett response
\begin{equation*}
\chi
\left(
\Omega,\Omega'
\right)
=
\left|
\mathbf a\left(
\Omega
\right)^{\mathsf H}
\mathbf a\left(
\Omega'
\right)
\right|^2.
\end{equation*}
Assume
\begin{equation}
0
\le
g\left(
\Omega
\right)
\le
g_{\max},
\qquad
\left\|
\mathbf a\left(
\Omega
\right)
\right\|_2
\le
A,
\qquad
\left\|
\mathbf a\left(
\Omega
\right)
-
\mathbf a\left(
\Omega'
\right)
\right\|_2
\le
L_{\mathrm a}
d_{\mathrm a}
\left(
\Omega,\Omega'
\right),
\label{eq:app_array_assumptions}
\end{equation}
where~\(d_{\mathrm a}\) is a response relevant direction discrepancy.

\begin{lemma}[Array Response Sensitivity]
\label{lem:array_sensitivity}
Under~\eqref{eq:app_array_assumptions},
\begin{equation}
\chi
\left(
\Omega,\Omega'
\right)
\le
A^4,
\qquad
\left|
\chi
\left(
\Omega,\Omega_1
\right)
-
\chi
\left(
\Omega,\Omega_2
\right)
\right|
\le
2A^3L_{\mathrm a}
d_{\mathrm a}
\left(
\Omega_1,\Omega_2
\right).
\label{eq:app_chi_bound}
\end{equation}
\end{lemma}

\begin{proof}
By~Cauchy--Schwarz,
\begin{equation*}
\left|
\mathbf a\left(
\Omega
\right)^{\mathsf H}
\mathbf a\left(
\Omega'
\right)
\right|
\le
\left\|
\mathbf a\left(
\Omega
\right)
\right\|_2
\left\|
\mathbf a\left(
\Omega'
\right)
\right\|_2
\le
A^2,
\end{equation*}
which gives the first inequality.

For
\begin{equation*}
x
=
\mathbf a\left(
\Omega
\right)^{\mathsf H}
\mathbf a\left(
\Omega_1
\right),
\qquad
y
=
\mathbf a\left(
\Omega
\right)^{\mathsf H}
\mathbf a\left(
\Omega_2
\right),
\end{equation*}
we have~\(\left|x\right|,\left|y\right|\le A^2\), and therefore
\begin{align*}
\left|
\left|x\right|^2
-
\left|y\right|^2
\right|
&\le
\left(
\left|x\right|
+
\left|y\right|
\right)
\left|
x-y
\right|
\\
&\le
2A^3
\left\|
\mathbf a\left(
\Omega_1
\right)
-
\mathbf a\left(
\Omega_2
\right)
\right\|_2
\\
&\le
2A^3L_{\mathrm a}
d_{\mathrm a}
\left(
\Omega_1,\Omega_2
\right).
\end{align*}
\end{proof}

For a planar unit modulus array with element locations~\(\mathbf d_u\in\mathbb R^2\),
\begin{equation*}
a_u\left(
\Omega
\right)
=
e^{-j\kappa
\mathbf d_u^{\mathsf T}
\mathbf u_{\parallel}\left(
\Omega
\right)},
\qquad
\kappa
=
\frac{2\pi}{\lambda},
\end{equation*}
where~\(\mathbf u_{\parallel}\left(\Omega\right)\) is the projection of the arrival direction onto the array plane.
Define
\begin{equation*}
d_{\parallel}
\left(
\Omega,\Omega'
\right)
=
\left\|
\mathbf u_{\parallel}\left(
\Omega
\right)
-
\mathbf u_{\parallel}\left(
\Omega'
\right)
\right\|_2.
\end{equation*}
If~\(\mathbf D\) contains the centered element locations as its rows, then~\(\left|e^{jx}-e^{jy}\right|\le\left|x-y\right|\) gives
\begin{equation*}
\left\|
\mathbf a\left(
\Omega
\right)
-
\mathbf a\left(
\Omega'
\right)
\right\|_2
\le
\kappa
\left\|
\mathbf D
\right\|_2
d_{\parallel}
\left(
\Omega,\Omega'
\right).
\end{equation*}
Hence one may choose
\begin{equation*}
A
=
\sqrt{N},
\qquad
L_{\mathrm a}
=
\kappa
\left\|
\mathbf D
\right\|_2.
\end{equation*}
Moreover, projection is nonexpansive, so if~\(d_{\mathrm{geo}}\) is the geodesic angle between two arrival directions,
\begin{equation*}
d_{\parallel}
\left(
\Omega,\Omega'
\right)
\le
2
\sin
\left(
\frac{
d_{\mathrm{geo}}
\left(
\Omega,\Omega'
\right)
}{
2
}
\right)
\le
d_{\mathrm{geo}}
\left(
\Omega,\Omega'
\right).
\end{equation*}
Thus, the same sensitivity bound also holds with geodesic angular error.

\textbf{Propagation to Spectrum Error.}
Write
\begin{equation*}
\widehat{\mathbf z}
=
\left\{
\left(
\widehat\Omega_n,\widehat m_n
\right)
\right\}_{n=0}^{K},
\qquad
\mathbf z^\star
=
\left\{
\left(
\Omega_n^\star,m_n^\star
\right)
\right\}_{n=0}^{K}.
\end{equation*}
Both~\(\widehat\Omega_0\) and~\(\Omega_0^\star\) equal the known~LOS direction.
Because the learned components are unordered, let~\(\upsilon\) denote a matching permutation over~\(\left\{1,\ldots,K\right\}\).

Define the matched component power error and power weighted direction error as
\begin{align*}
E_m
&=
\left|
\widehat m_0-m_0^\star
\right|
+
\sum_{n=1}^{K}
\left|
\widehat m_{\upsilon\left(n\right)}
-
m_n^\star
\right|,
\\
E_{\Omega}
&=
\sum_{n=1}^{K}
m_n^\star
d_{\mathrm a}
\left(
\widehat\Omega_{\upsilon\left(n\right)},
\Omega_n^\star
\right).
\end{align*}

\begin{proof}[Proof of Proposition~\ref{prop:rendering_error}]
By the triangle inequality,
\begin{equation*}
\left\|
\widehat P-P
\right\|_\infty
\le
\varepsilon_{\mathrm{phys}}
+
\left\|
\widehat P-P^\star
\right\|_\infty.
\end{equation*}

For each learned component,
\begin{align*}
&
\left|
\widehat m_{\upsilon\left(n\right)}
\chi
\left(
\Omega,\widehat\Omega_{\upsilon\left(n\right)}
\right)
-
m_n^\star
\chi
\left(
\Omega,\Omega_n^\star
\right)
\right|
\\
&\quad\le
\left|
\widehat m_{\upsilon\left(n\right)}
-
m_n^\star
\right|
\chi
\left(
\Omega,\widehat\Omega_{\upsilon\left(n\right)}
\right)
\\
&\qquad+
m_n^\star
\left|
\chi
\left(
\Omega,\widehat\Omega_{\upsilon\left(n\right)}
\right)
-
\chi
\left(
\Omega,\Omega_n^\star
\right)
\right|
\\
&\quad\le
A^4
\left|
\widehat m_{\upsilon\left(n\right)}
-
m_n^\star
\right|
+
2A^3L_{\mathrm a}
m_n^\star
d_{\mathrm a}
\left(
\widehat\Omega_{\upsilon\left(n\right)},
\Omega_n^\star
\right),
\end{align*}
where the last inequality follows from Lemma~\ref{lem:array_sensitivity}.
For the~LOS component, the direction is fixed, so its contribution is bounded by~\(A^4\left|\widehat m_0-m_0^\star\right|\).

Summing over the components, using~\(g\left(\Omega\right)\le g_{\max}\), and taking the maximum over the angular grid gives
\begin{equation*}
\left\|
\widehat P-P^\star
\right\|_\infty
\le
g_{\max}A^4E_m
+
2g_{\max}A^3L_{\mathrm a}E_{\Omega}.
\end{equation*}
Therefore,
\begin{equation}
\boxed{
\left\|
\widehat P-P
\right\|_\infty
\le
\varepsilon_{\mathrm{phys}}
+
C_mE_m
+
C_{\Omega}E_{\Omega},
}
\qquad
C_m
=
g_{\max}A^4,
\qquad
C_{\Omega}
=
2g_{\max}A^3L_{\mathrm a}.
\label{eq:app_final_renderer_bound}
\end{equation}
This proves Proposition~\ref{prop:rendering_error}.
\end{proof}

For the planar unit modulus array above, the sensitivity constants become
\begin{equation}
C_m
=
g_{\max}N^2,
\qquad
C_{\Omega}
=
2g_{\max}
N^{3/2}
\kappa
\left\|
\mathbf D
\right\|_2.
\label{eq:app_explicit_C}
\end{equation}
Thus, for a fixed receive array, the conversion from propagation errors to spectrum error is determined by fixed array physics rather than the scene.

\textbf{Magnitude and Normalized Spectrum.}
The preceding result is stated in the power domain because the analytic renderer is additive in the component powers.
Let
\begin{equation*}
S
=
\sqrt{P},
\qquad
\widehat S_{\mathrm u}
=
\sqrt{\widehat P}
\end{equation*}
denote the target and rendered unnormalized magnitude spectra.
Since~\(\left|\sqrt{x}-\sqrt{y}\right|^2\le\left|x-y\right|\) for~\(x,y\ge0\),
\begin{equation}
\left\|
\widehat S_{\mathrm u}-S
\right\|_\infty
\le
\sqrt{
\left\|
\widehat P-P
\right\|_\infty
}.
\label{eq:app_spectrum_error}
\end{equation}

For a nonconstant spectrum~\(u\), define
\begin{equation*}
\mathcal N\left(u\right)
=
\frac{
u-\min u
}{
\max u-\min u
},
\qquad
r_u
=
\max u-\min u.
\end{equation*}

\begin{lemma}[Normalization Stability]
\label{lem:normalization_stability}
For any two nonconstant spectra~\(u\) and~\(v\) on the finite angular grid,
\begin{equation}
\left\|
\mathcal N\left(u\right)
-
\mathcal N\left(v\right)
\right\|_\infty
\le
\frac{
2
}{
\max\left\{
r_u,r_v
\right\}
}
\left\|
u-v
\right\|_\infty.
\label{eq:app_norm_lipschitz}
\end{equation}
\end{lemma}

\begin{proof}
Set~\(\varepsilon=\left\|u-v\right\|_\infty\),~\(a=\min u\),~\(b=\max u\),~\(a'=\min v\), and~\(b'=\max v\).
Then~\(\left|a-a'\right|\le\varepsilon\) and~\(\left|b-b'\right|\le\varepsilon\).
At any grid point, write~\(v=a'+tr_v\) with~\(t\in\left[0,1\right]\) and~\(u=v+\delta\) with~\(\left|\delta\right|\le\varepsilon\).
It follows that
\begin{equation*}
\left|
\mathcal N\left(u\right)
-
\mathcal N\left(v\right)
\right|
=
\frac{
\left|
\left(1-t\right)\left(a'-a\right)
+
t\left(b'-b\right)
+
\delta
\right|
}{
r_u
}
\le
\frac{
2\varepsilon
}{
r_u
}.
\end{equation*}
Interchanging~\(u\) and~\(v\) gives the same bound with~\(r_v\), and taking the smaller of the two bounds proves~\eqref{eq:app_norm_lipschitz}.
\end{proof}

Applying Lemma~\ref{lem:normalization_stability} to~\(\widehat S_{\mathrm u}\) and~\(S\), and then using~\eqref{eq:app_spectrum_error}, gives
\begin{equation}
\left\|
\mathcal N\left(
\widehat S_{\mathrm u}
\right)
-
\mathcal N\left(
S
\right)
\right\|_\infty
\le
\frac{
2
}{
\max\left\{
r_{\widehat S_{\mathrm u}},r_S
\right\}
}
\sqrt{
\left\|
\widehat P-P
\right\|_\infty
}.
\label{eq:app_normalized_spectrum_error}
\end{equation}
Thus, the power domain bound transfers directly to the normalized magnitude spectrum whenever both spectra have nonzero range.

\textbf{Scope.}
The result is a deterministic worst case stability bound and does not assume that~\(\varepsilon_{\mathrm{phys}}\) is small.
It characterizes how specified propagation parameter errors are converted into spectrum error, rather than how accurately the learned model estimates those parameters.
For a planar array,~\(d_{\parallel}\) reflects the intrinsic response ambiguity between directions with identical in plane direction cosines.
The result therefore separates scene-dependent propagation inference from the fixed propagation-to-spectrum sensitivity of the analytic renderer.


\section{Implementation Details}
\label{app:implementation}

This appendix describes the training objectives and hyperparameters~(\S\ref{app:implementation:training}), 
model architecture~(\S\ref{app:implementation:model}), 
optional power refinement~(\S\ref{app:implementation:power_adaptation}), 
receiver array calibration~(\S\ref{app:implementation:array_calibration}),
and dataset~(\S\ref{app:implementation:data}).

\subsection{Training Objectives and Hyperparameters}
\label{app:implementation:training}

\textbf{Training Protocol.}
Training proceeds in two stages.
Stage~I pretrains the spectrum encoder~\(E_\phi\) and reference contextualizer~\(T_\psi\), after which these modules are frozen.
Stage~II trains the reference-to-anchor projection, propagation estimation, query conditioning, and uncertainty modules.
Unless stated otherwise, each training instance uses~\(M=32\) references, with the queries excluded from the reference set.
Unless stated otherwise, the canonical main-paper checkpoints use random seed~\(0\).
Training-randomness estimates repeat these checkpoints with seeds~\(\{0,1,2\}\) on all~\(12\) folds, and the coordinate and pretraining ablations also use three seeds per condition; these exceptions are identified with their results.

\textbf{Stage I: Reference Representation Pretraining.}
We pretrain~\(E_\phi\) and~\(T_\psi\) using supervised contrastive learning over scene identity.
Reference sets from the same scene are treated as positives, while those from different scenes are treated as negatives.
Each training step draws~\(B_{\mathrm s}=4\) reference sets from every training scene.
Each set~\(b\) is encoded by~\(E_\phi\) and~\(T_\psi\), mean-pooled over its~\(M\) contextualized tokens, passed through a projection head, and~\(\ell_2\)-normalized to obtain~\(\mathbf z_b\).
Let~\(y_b\) denote the scene label of set~\(b\), let~\(\mathcal P\left(b\right)=\left\{b'\neq b:y_{b'}=y_b\right\}\) denote its positive sets, and let~\(B\) be the total number of reference sets in the training step.
The supervised contrastive objective is
\begin{equation}
\mathcal L_{\mathrm{sup}}
=
-\frac{1}{B}
\sum_{b=1}^{B}
\frac{1}{\left|\mathcal P\left(b\right)\right|}
\sum_{b'\in\mathcal P\left(b\right)}
\log
\frac{
\exp\left(
\mathbf z_b^{\top}\mathbf z_{b'} / \tau_{\mathrm c}
\right)
}{
\sum_{b''\neq b}
\exp\left(
\mathbf z_b^{\top}\mathbf z_{b''} / \tau_{\mathrm c}
\right)
}.
\label{eq:app_supcon}
\end{equation}
We use the contrastive temperature~\(\tau_{\mathrm c}=0.1\).
Stage~I is trained for~\(1500\) steps using Adam with learning rate~\(3\times10^{-4}\).

\textbf{Stage II: Propagation Field Training.}
We freeze the Stage~I modules and train the remaining model using query spectra outside the reference set.
The full objective is
\begin{equation}
\mathcal L
=
\mathcal L_{\mathrm{rec}}
+
\lambda_{\mathrm{cu}}\mathcal L_{\mathrm{cu}}
+
\lambda_{\mathrm{mar}}\mathcal L_{\mathrm{mar}}
+
\lambda_{\mathrm{unc}}\mathcal L_{\mathrm{unc}}.
\label{eq:app_total_loss}
\end{equation}
Here,~\(\mathcal L_{\mathrm{rec}}\) is the spectrum reconstruction loss, \(\mathcal L_{\mathrm{cu}}\) is the content-use loss, \(\mathcal L_{\mathrm{mar}}\) is the reference-separation margin loss, and~\(\mathcal L_{\mathrm{unc}}\) is the uncertainty loss.
The coefficients~\(\lambda_{\mathrm{cu}}\), \(\lambda_{\mathrm{mar}}\), and~\(\lambda_{\mathrm{unc}}\) weight the corresponding auxiliary objectives.

\textit{1) Reconstruction Loss~\(\mathcal L_{\mathrm{rec}}\).}
The reconstruction objective combines amplitude-weighted pointwise error, summed-amplitude consistency, structural similarity, and normalized energy error:
\begin{equation}
\mathcal L_{\mathrm{rec}}
=
\mathcal L_{\mathrm{amp}}
+
\lambda_{\mathrm e}\mathcal L_{\mathrm{energy}}
+
\lambda_{\mathrm s}\mathcal L_{\mathrm{SSIM}}
+
\lambda_{\mathrm n}\mathcal L_{\mathrm{norm}}.
\label{eq:app_reconstruction_loss}
\end{equation}
All four terms compare the predicted spectrum~\(\widehat S\) and target spectrum~\(S\) on the~\(H\times W\) angular grid~\(\mathcal G\), using the per-measurement normalization described in~\S\ref{app:implementation:data}.
For the amplitude weighting, we use~\(\gamma_{\mathrm a}=0.7\) and~\(\varepsilon_{\mathrm a}=0.02\).
The reconstruction-loss weights are~\(\lambda_{\mathrm e}=0.2\), \(\lambda_{\mathrm s}=0.6\), and~\(\lambda_{\mathrm n}=0.5\).

\textit{2) Amplitude-Weighted Loss~\(\mathcal L_{\mathrm{amp}}\).}
The amplitude-weighted pointwise loss emphasizes high-response regions of the target spectrum:
\begin{equation}
\mathcal L_{\mathrm{amp}}
=
\frac{
\sum_{\Omega\in\mathcal G}
w\left(\Omega\right)
\left|
\widehat S\left(\Omega\right)
-
S\left(\Omega\right)
\right|
}{
\sum_{\Omega\in\mathcal G}
w\left(\Omega\right)
},
\qquad
w\left(\Omega\right)
=
S\left(\Omega\right)^{\gamma_{\mathrm a}}
+
\varepsilon_{\mathrm a}.
\label{eq:app_loss_amp}
\end{equation}

\textit{3) Summed-Amplitude Loss~\(\mathcal L_{\mathrm{energy}}\).}
This loss matches the summed spectrum amplitude and normalizes the error by the number of angular cells:
\begin{equation}
\mathcal L_{\mathrm{energy}}
=
\frac{1}{HW}
\left|
\sum_{\Omega\in\mathcal G}
\widehat S\left(\Omega\right)
-
\sum_{\Omega\in\mathcal G}
S\left(\Omega\right)
\right|.
\label{eq:app_loss_energy}
\end{equation}

\textit{4) Structural Loss~\(\mathcal L_{\mathrm{SSIM}}\).}
The structural loss is one minus the mean~SSIM map, computed using a uniform~\(7\times7\) window and the standard constants~\(C_1=0.01^2\) and~\(C_2=0.03^2\):
\begin{equation}
\mathcal L_{\mathrm{SSIM}}
=
1
-
\frac{1}{HW}
\sum_{\Omega\in\mathcal G}
\frac{
\left(
2\mu_{\widehat S}\mu_S+C_1
\right)
\left(
2\sigma_{\widehat S S}+C_2
\right)
}{
\left(
\mu_{\widehat S}^2+\mu_S^2+C_1
\right)
\left(
\sigma_{\widehat S}^2+\sigma_S^2+C_2
\right)
}.
\label{eq:app_loss_ssim}
\end{equation}
Here,~\(\mu_{\widehat S}\) and~\(\mu_S\) denote the local means, \(\sigma_{\widehat S}^2\) and~\(\sigma_S^2\) denote the local variances, and~\(\sigma_{\widehat S S}\) denotes the local covariance at~\(\Omega\).

\textit{5) Normalized Energy Loss~\(\mathcal L_{\mathrm{norm}}\).}
The normalized energy loss measures squared reconstruction error relative to the target energy and corresponds to the~NMSE metric:
\begin{equation}
\mathcal L_{\mathrm{norm}}
=
\frac{
\sum_{\Omega\in\mathcal G}
\left(
\widehat S\left(\Omega\right)
-
S\left(\Omega\right)
\right)^2
}{
\sum_{\Omega\in\mathcal G}
S\left(\Omega\right)^2
+
\varepsilon
},
\qquad
\varepsilon=10^{-8}.
\label{eq:app_loss_norm}
\end{equation}

\textit{6) Reference-Use Regularization~\(\mathcal L_{\mathrm{cu}},\mathcal L_{\mathrm{mar}}\).}
To encourage the conditioned field to use the target-scene references, we introduce a content-use loss and a reference-separation margin loss.
The content-use loss reconstructs the spectra contained in the reference set:
\begin{equation}
\mathcal L_{\mathrm{cu}}
=
\frac{1}{M}
\sum_{i=1}^{M}
\mathcal L_{\mathrm{rec}}
\left(
\widehat S_{\pi_i},
S_{\pi_i}
\right).
\label{eq:app_content_use}
\end{equation}
We additionally sample a mismatched reference set~\(\mathcal C'\) from another scene and encourage the matched reference set~\(\mathcal C\) to produce a lower query reconstruction loss:
\begin{equation}
\mathcal L_{\mathrm{mar}}
=
\left[
\mathcal L_{\mathrm{rec}}
\left(
\widehat S_{\pi_\star}^{\mathcal C},
S_{\pi_\star}
\right)
-
\mathcal L_{\mathrm{rec}}
\left(
\widehat S_{\pi_\star}^{\mathcal C'},
S_{\pi_\star}
\right)
+
\delta_{\mathrm m}
\right]_+.
\label{eq:app_reference_margin}
\end{equation}
We use~\(\lambda_{\mathrm{cu}}=0.5\), \(\lambda_{\mathrm{mar}}=1.0\), and~\(\delta_{\mathrm m}=0.02\).
The content-use loss is introduced with a~\(400\)-step warm-in.
Both reference-use losses are enabled only for cross-scene training and are disabled for the per-scene setting.

\textit{7) Uncertainty Loss~\(\mathcal L_{\mathrm{unc}}\).}
The uncertainty branch is trained using a heteroscedastic Gaussian negative log-likelihood based on the analytically propagated spectrum variance.
Gradients from this objective are detached from the mean-prediction branch.
Let~\(\sigma^2\left(\Omega\right)\) denote the per-cell predictive variance of the normalized spectrum.
The uncertainty objective is
\begin{equation}
\mathcal L_{\mathrm{unc}}
=
\frac{1}{2HW}
\sum_{\Omega\in\mathcal G}
\left[
\frac{
\left(
S\left(\Omega\right)
-
\operatorname{sg}\!\left[
\widehat S\left(\Omega\right)
\right]
\right)^2
}{
\sigma^2\left(\Omega\right)
+
\varepsilon_{\mathrm u}
}
+
\log
\left(
\sigma^2\left(\Omega\right)
+
\varepsilon_{\mathrm u}
\right)
\right].
\label{eq:app_uncertainty_loss}
\end{equation}
Here,~\(\operatorname{sg}\left[\cdot\right]\) denotes the stop-gradient operator.
The calibrated predictive variance is
\begin{equation}
\sigma^2\left(\Omega\right)
=
\tau_{\mathrm u}^2
\sigma_{\mathrm{phys}}^2\left(\Omega\right)
+
\varepsilon_0,
\label{eq:app_uncertainty_variance}
\end{equation}
where~\(\sigma_{\mathrm{phys}}^2\left(\Omega\right)\) is the analytically propagated variance from~\Eqref{eq:spectrum_uncertainty}, expressed in the units of the normalized spectrum.
For unconstrained optimization, we parameterize the global uncertainty temperature and variance floor as~\(\tau_{\mathrm u}=e^{\vartheta_\tau}\) and~\(\varepsilon_0=\operatorname{softplus}\left(\vartheta_\varepsilon\right)\), with~\(\vartheta_\tau\) and~\(\vartheta_\varepsilon\) initialized to~\(0\) and~\(-4.6\), respectively.
We additionally use a fixed numerical floor~\(\varepsilon_{\mathrm u}=10^{-4}\).
The uncertainty objective updates only the concentration head that produces the component variances and the global calibration parameters~\(\tau_{\mathrm u}\) and~\(\varepsilon_0\) through their unconstrained parameters~\(\vartheta_\tau\) and~\(\vartheta_\varepsilon\).
The mean prediction, canonical anchor field, and analytic array response enter this objective with stopped gradients.
We use~\(\lambda_{\mathrm{unc}}=1.0\).


\subsection{Model Architecture}
\label{app:implementation:model}

We use the same architecture across all reported experiments unless stated otherwise.
\autoref{tab:architecture} summarizes the main architectural configuration.
The implementation follows the cross-scene propagation field in~\S\ref{sec:cond} and the analytic array physics in~\S\ref{sec:decoder}.

\begin{table}[t]
\centering
\small
\caption{\textit{Model architecture of \ourSystem used in reported experiments.}}
\label{tab:architecture}
\setlength{\tabcolsep}{5pt}
\begin{tabular}{lll}
\toprule
Module & Setting & Value \\
\midrule
Canonical anchor field
& Anchor lattice
& \(20\times20\times10=4{,}000\) \\
& Canonical extent
& \(\left[-1,1\right]^3\) \\
& Physical extent
& \(\left[-5,5\right]\times\left[-4,4\right]\times\left[0,3\right]\,\mathrm m\) \\
& Backbone width
& \(512\) \\
& Anchor outputs
& \(\Delta\mathbf x_n,\alpha_n,\mathbf e_n^0\) \\
\midrule
Reference encoder
& Spectrum CNN
& Conv~\(32\)-\(64\)-\(128\)-\(128\) \\
& Pooling
& AdaptiveAvgPool8 \\
& Contextualizer
& Transformer, dim~\(256\) \\
& Attention
& \(4\) heads, \(2\) layers \\
& FFN width
& \(512\) \\
& Positional encoding
& none \\
\midrule
Anchor conditioning
& Cross-attention
& dim~\(512\), \(4\) heads \\
& Bearing-aligned CNN
& \(32\) channels, angular grid preserved \\
\midrule
Query conditioning
& Model width
& \(64\) \\
& Fourier frequencies
& \(6\) \\
& Modulation
& affine FiLM \\
\midrule
Uncertainty
& Parameterization
& \(v_n=m_n^2/\kappa_n\) \\
& Calibration
& global~\(\tau_{\mathrm u}\), variance floor~\(\varepsilon_0\) \\
\midrule
Renderer
& Receiver array
& \(4\times4\) at~\(920\,\mathrm{MHz}\) \\
& Rendering
& analytic Bartlett with explicit LOS \\
& LOS power
& global, \(\operatorname{softplus}\left(\cdot\right)^2\) \\
& Uncertainty
& analytic moment propagation \\
\midrule
Model size
& Total
& \(12.7\)M \\
& Reference encoder
& \(3.8\)M \\
& Field and conditioning
& \(8.9\)M \\
\bottomrule
\end{tabular}
\end{table}

\textbf{Reference Encoder and Contextualizer.}
Each reference spectrum is processed independently by a four-stage~2D CNN with channel widths~\(32\), \(64\), \(128\), and~\(128\), followed by adaptive average pooling.
The resulting reference token is contextualized by a two-layer Transformer with embedding dimension~\(256\), four attention heads, and feed-forward width~\(512\).
No positional encoding is applied to the reference sequence, so the contextualizer remains permutation equivariant over the unordered reference set.
In parallel, a shallow bearing-aligned CNN with~\(32\) channels preserves the angular grid and provides direction-specific reference features.

\textbf{Canonical Anchor Field.}
The shared field contains~\(K=20\times20\times10=4{,}000\) anchors over the canonical extent~\(\left[-1,1\right]^3\), corresponding to the normalized physical extent~\(\left[-5,5\right]\times\left[-4,4\right]\times\left[0,3\right]\,\mathrm m\).
Each anchor uses a~\(512\)-dimensional backbone representation and cross-attends to the contextualized reference tokens using four attention heads.
The resulting global reference-conditioned feature is fused with the bearing-aligned feature before anchor decoding.
The anchor decoder predicts the position refinement~\(\Delta\mathbf x_n\), occupancy parameter~\(\alpha_n\), and base emission~\(\mathbf e_n^0\).
The position refinement is constrained within half of one anchor cell as described in~\S\ref{sec:cond}.

\textbf{Query and Uncertainty Heads.}
Query and anchor geometry in the receiver-array frame are encoded using six Fourier frequencies.
A width-64 network produces affine FiLM modulation of the base anchor emission, yielding the query-conditioned emission used to compute the mean component power~\(m_n\).
A separate uncertainty head predicts a positive concentration~\(\kappa_n\), giving component-power variance~\(v_n=m_n^2/\kappa_n\).
A learned global temperature~\(\tau_u\) and variance floor~\(\epsilon_0\) calibrate the propagated component-power variance.
Uncertainty gradients are detached from the mean branch, leaving the synthesized mean spectrum unchanged.

\textbf{Analytic Renderer.}
The learned model outputs the propagation parameters~\(\left\{\left(\Omega_n,m_n,v_n\right)\right\}_{n=1}^{K}\), which are rendered through the analytic receiver-array model in~\S\ref{sec:decoder}.
For the simulated arrays, the receiver elements are isotropic and the element power pattern is~\(g\equiv1\).
The explicit LOS power~\(m_{\mathrm{los}}\) is parameterized as a squared softplus and shared across scenes and queries.
Because the targets are normalized independently, this scalar represents a normalized direct-path prior rather than absolute received power.
The renderer constructs the receiver covariance, evaluates the Bartlett response, and propagates component-power uncertainty analytically through the same array response.

\textbf{Implementation Compatibility.}
The released checkpoint format inherits a degree-\(9\) angular emission head and scale and rotation outputs from the Gaussian-field implementation on which the code was built.
CoRF reads only the zeroth-order complex emission coefficient; the other emission coefficients and the scale and rotation outputs are disconnected from the covariance, spectrum, uncertainty, and loss.
They are retained solely for checkpoint compatibility and can be pruned without changing any prediction.
The reported checkpoint parameter totals conservatively include this dormant head, while the mathematical model and analyses use only the active variables defined in~\S\ref{sec:cond}.

\subsection{Target-Scene Power Refinement}
\label{app:implementation:power_adaptation}

Power refinement is an optional target-scene optimization that provides a lightweight correction while keeping the pretrained propagation model fixed.
For a target scene, we freeze the spectrum encoder, reference contextualizer, canonical anchor field, propagation estimation modules, and analytic array model.
Only the component powers are refined through lightweight scene-specific and query-dependent residual corrections.
Refinement uses only the same~\(M\) target-scene references already used to condition the canonical anchor field and requires no additional measurements, but it is distinct from feed-forward conditioning.

Although the cross-scene model transfers the propagation structure to an unseen scene, the predicted component powers can still exhibit scene-specific mismatch because attenuation and reflection strengths vary across environments.
Rather than fitting a new scene-specific propagation field, we use the available target-scene references to apply a lightweight correction to~\(\left\{m_n\right\}_{n=1}^{K}\) while preserving the learned anchor geometry and arrival directions.
After this one-time refinement, the corrected model is fixed and reused for all queries in the target scene.

\textbf{Residual Power Correction.}
For each target scene, we introduce a lightweight residual model that corrects the pretrained component powers while leaving the propagation geometry fixed.
For anchor~\(n\) and query configuration~\(\pi_\star\), the residual correction is
\begin{equation}
\delta_n\left(\pi_\star\right)
=
b_n^{\mathrm{PR}}
+
r_\zeta\left(
\Gamma_{\mathrm F}\left(\pi_\star\right)
\right)_n,
\label{eq:app_power_residual}
\end{equation}
where~\(b_n^{\mathrm{PR}}\) is a learned per-anchor power bias and~\(r_\zeta\) is a residual MLP applied to a Fourier encoding~\(\Gamma_{\mathrm F}\left(\pi_\star\right)\) of the query position, array-local direction, and range.
The bias~\(b_n^{\mathrm{PR}}\) captures a scene-specific correction shared across all queries, whereas~\(r_\zeta\left(\Gamma_{\mathrm F}\left(\pi_\star\right)\right)_n\) captures an additional query-dependent correction.
Both~\(\left\{b_n^{\mathrm{PR}}\right\}_{n=1}^{K}\) and~\(r_\zeta\) are fitted once for the target scene using the same~\(M\) reference measurements used for scene conditioning.
We apply this residual in the pre-softplus domain so that the corrected component power remains nonnegative:
\begin{equation}
\widetilde m_n\left(\pi_\star\right)
=
\operatorname{softplus}\left(
\operatorname{softplus}^{-1}\left(
m_n\left(\pi_\star\right)
\right)
+
\delta_n\left(\pi_\star\right)
\right).
\label{eq:app_power_adaptation}
\end{equation}
The anchor biases and the output layer of~\(r_\zeta\) are initialized to zero, so~\(\widetilde m_n=m_n\) before refinement.
After fitting, the residual model is fixed and reused for all queries in the target scene.
Because~\Eqref{eq:app_power_adaptation} modifies only the non-LOS component powers, the canonical anchor geometry and arrival directions~\(\left\{\Omega_n\right\}_{n=1}^{K}\) remain unchanged.
In the per-scene setting, we additionally allow a scalar correction to the explicit LOS power, as described below.

\textbf{Residual-Map Architecture.}
The residual-map capacity depends on the amount of target-scene data available for refinement.
In the few-shot setting,~\(r_\zeta\) is a two-layer MLP with width~\(64\), applied to a two-octave Fourier encoding of the query geometry with~\(35\) input dimensions, for~\(0.266\)M parameters.
Together with the~\(4{,}000\) anchor biases, the full few-shot correction adapts~\(0.270\)M parameters per target scene.
When the full scene training split is available, we use a three-layer MLP with width~\(512\) and an eight-octave Fourier encoding with~\(119\) input dimensions, for approximately~\(2.6\)M parameters.
Before refinement, the pretrained model is evaluated once to cache the base component powers~\(\left\{m_n\left(\pi\right)\right\}_{n=1}^{K}\) for the refinement measurements.
The residual model then learns only the correction in~\Eqref{eq:app_power_residual}, while these pretrained powers remain fixed.

\textbf{Protocol.}
For each target scene, we first use the~\(M=32\) references to construct the scene-conditioned canonical anchor field with the pretrained model.
We then fit one residual model~\(\left\{\left\{b_n^{\mathrm{PR}}\right\}_{n=1}^{K},r_\zeta\right\}\) for that scene using the same references while keeping the pretrained model frozen.
After refinement, the fitted residual model is fixed and reused for all test queries in the scene without further optimization.
We optimize the residual map using the reconstruction objective~\(\mathcal L_{\mathrm{rec}}\) defined in~\S\ref{app:implementation:training}.
Adam is used with learning rate~\(10^{-3}\) for~\(r_\zeta\) and~\(5\times10^{-3}\) for the anchor biases~\(\left\{b_n^{\mathrm{PR}}\right\}_{n=1}^{K}\), together with a cosine learning-rate schedule, batch size~\(8\), and gradient clipping at~\(1.0\).

In the few-shot setting, the same~\(M=32\) target-scene references constitute the complete refinement budget.
All~\(32\) references are used to condition the canonical anchor field, while~\(24\) are used to fit the residual model and the remaining~\(8\) are reserved for step selection.
We optimize for up to~\(3000\) steps with weight decay~\(10^{-4}\), evaluate the selection references every~\(100\) steps, and retain the step with the highest mean~PSNR.
Thus, optional power refinement requires no target-scene measurements beyond the original~\(M=32\) references.

For the per-scene setting, the scene training split is used for fitting and the validation split for step selection.
We optimize for up to~\(16{,}000\) steps without weight decay and evaluate validation performance every~\(500\) steps.
In this setting, the residual map additionally includes a linear head on its hidden features that rescales the explicit LOS power by a bounded exponential factor initialized to one.
The step with the highest validation~PSNR is retained.

\subsection{Receiver Array Calibration}
\label{app:implementation:array_calibration}

Receiver array calibration adapts the analytic receiver response when the deployed array differs from its nominal specification.
Unlike power refinement, calibration leaves the predicted propagation parameters~\(\left\{\left(\Omega_n,m_n,v_n\right)\right\}_{n=1}^{K}\) unchanged and modifies only the receiver response used by the analytic renderer.
Such calibration is needed because errors in array orientation, element geometry, or element response can distort the mapping from the predicted propagation components to the observed spatial spectrum.
The spectrum encoder, reference contextualizer, canonical anchor field, and propagation estimation modules therefore remain frozen throughout calibration.
Calibration uses a separate set of~\(N_{\mathrm C}\) measurements and is performed once for each receiver, after which the calibrated response is reused for all queries collected with that receiver.

\textbf{Calibration Formulation.}
For a receiver array with nominal steering response~\(\mathbf a\left(\Omega\right)\), we introduce receiver-specific calibration parameters~\(\boldsymbol\nu_{\mathrm{arr}}\) and replace the nominal response with
\begin{equation}
\mathbf a\left(\Omega\right)
\longrightarrow
\mathbf a_{\boldsymbol\nu_{\mathrm{arr}}}\left(\Omega\right).
\label{eq:app_calibrated_array_response}
\end{equation}
We estimate~\(\boldsymbol\nu_{\mathrm{arr}}\) from the~\(N_{\mathrm C}\) calibration measurements while keeping the learned propagation field fixed.
The calibration measurements are disjoint from the~\(M\) references used for scene conditioning and from the test queries.

\textbf{Simulated Receiver Arrays.}
For a simulated receiver array, we construct the analytic receiver response from its nominal geometry and fit the three-parameter calibration vector
\begin{equation}
\boldsymbol\nu_{\mathrm{arr}}
=
\left(
\Delta\varphi_{\mathrm{arr}},
s_{\mathrm{arr},x},
s_{\mathrm{arr},y}
\right),
\label{eq:app_array_calibration}
\end{equation}
where~\(\Delta\varphi_{\mathrm{arr}}\) corrects the in-plane array orientation, and~\(s_{\mathrm{arr},x}\) and~\(s_{\mathrm{arr},y}\) correct the element spacing along the two array axes.
We estimate~\(\boldsymbol\nu_{\mathrm{arr}}\) from~\(N_{\mathrm C}\) calibration measurements while keeping the learned propagation field fixed.
After calibration, the corrected response~\(\mathbf a_{\boldsymbol\nu_{\mathrm{arr}}}\left(\Omega\right)\) is fixed and reused for all subsequent queries collected with that receiver.

\textbf{Real Measured Receiver.}
For the real measured receiver, we keep the learned propagation model frozen and estimate only receiver-specific calibration parameters from~\(N_{\mathrm C}\) calibration measurements.
Because a physical receiver can exhibit element-dependent gain, geometry, and directional-response mismatch, we model the calibrated response using a complex gain per element, an in-plane position offset per element, and an angle-dependent element response.

For receiver element~\(u\), the complex gain is
\begin{equation}
g_u
=
\exp\left(
\ell_u+j\phi_u
\right),
\label{eq:app_element_gain}
\end{equation}
where~\(\ell_u\) and~\(\phi_u\) parameterize the log-amplitude and phase, respectively.
The log-amplitude~\(\ell_u\) is constrained to~\(\left[-3,3\right]\), and the element gains are normalized to unit mean magnitude.
The position correction~\(\Delta\mathbf p_u\in\mathbb R^2\) is added to the nominal element position when computing the steering phase.

The angle-dependent element response is
\begin{equation}
c_u\left(\varphi,\theta\right)
=
\exp\left(
\boldsymbol\alpha_u^{\top}
\mathbf f\left(\varphi,\theta\right)
+
j\,
\boldsymbol\beta_u^{\top}
\mathbf f\left(\varphi,\theta\right)
\right),
\label{eq:app_element_pattern}
\end{equation}
where~\(\varphi\) and~\(\theta\) denote the array-local azimuth and elevation.
We use the angular basis
\begin{equation}
\mathbf f\left(\varphi,\theta\right)
=
\left[
1,
\left\{
\cos\left(k\varphi\right),
\sin\left(k\varphi\right)
\right\}_{k=1}^{K_{\mathrm{az}}},
\left\{
\cos\left(j\theta\right),
\sin\left(j\theta\right)
\right\}_{j=1}^{K_{\mathrm{el}}}
\right].
\label{eq:app_element_basis}
\end{equation}
We set~\(K_{\mathrm{az}}=12\) and~\(K_{\mathrm{el}}=4\), giving~\(33\) coefficients for the log-amplitude~\(\boldsymbol\alpha_u\) and~\(33\) coefficients for the phase~\(\boldsymbol\beta_u\) of each of the~\(16\) receiver elements.
Including complex gain and the two-dimensional position correction, the real-array model fits~\(70\) scalar parameters per element, or~\(1{,}120\) in total.
The resulting log-amplitude of the angle-dependent response is constrained to~\(\left[-4,4\right]\).

The calibrated steering response evaluates the steering phase at the corrected element positions and multiplies the response of element~\(u\) by~\(g_u c_u\left(\Omega\right)\).
The same calibrated response is applied to both the anchor directions and the explicit line-of-sight direction.
All calibration parameters are initialized to reproduce the nominal receiver response, with~\(g_u=1\), \(\Delta\mathbf p_u=\mathbf 0\), and~\(c_u\left(\Omega\right)=1\).
No additional regularization is applied beyond the parameter bounds described above.

\textbf{Protocol.}
We use~\(N_{\mathrm C}=64\) calibration measurements by default.
These measurements are separate from the~\(M=32\) references used to condition the target-scene field and from the test queries.

For the simulated arrays, the three calibration parameters in~\Eqref{eq:app_array_calibration} are fitted by coordinate search because rebuilding the analytic response for a new array geometry is not differentiable in our implementation.
We perform three coordinate-search sweeps over the in-plane rotation from~\(-12^\circ\) to~\(12^\circ\) and the two axis scalings from~\(0.86\) to~\(1.16\), using~\(13\) candidate values for each parameter.
Within each sweep, the parameters are optimized sequentially by retaining the candidate that maximizes the mean~SSIM over the~\(N_{\mathrm C}\) calibration measurements.

For the real receiver, we optimize the calibration parameters using the reconstruction objective~\(\mathcal L_{\mathrm{rec}}\) defined in~\S\ref{app:implementation:training}, while keeping the conditioned propagation parameters fixed.
We use Adam with learning rate~\(5\times10^{-3}\), a cosine learning-rate schedule, and gradient clipping at~\(1.0\), and optimize for up to~\(8000\) steps.
We reserve one quarter of the~\(N_{\mathrm C}\) measurements for step selection and use the remaining three quarters for fitting.
With the default~\(N_{\mathrm C}=64\), this corresponds to~\(48\) fitting measurements and~\(16\) selection measurements.
The calibration checkpoint with the highest mean~PSNR on the selection measurements is retained.

After fitting, the receiver calibration is fixed and applied unchanged to all test queries.
The test queries are not used for calibration or model selection.
The default calibration-only protocol above uses disjoint conditioning and calibration sets, for a total of~\(32+64=96\) target-domain spectra.
For the joint measurement-matched real-data result in~\S\ref{app:real}, the scene conditioner uses~\(M=32\) references, and a disjoint set of~\(N_{\mathrm C}=64\) spectra calibrates the element-wise receiver response.
The power-refinement map defined in~\S\ref{app:implementation:power_adaptation} reuses both sets without additional measurements; a held-out subset of the same~\(96\)-spectrum budget is used for checkpoint selection.
The pretrained propagation model remains frozen, and the test split remains untouched.
This combined real-data deployment protocol is denoted by~\ourSystemPRCAL.
The real experiment evaluates one initialization and one measured receiver, and therefore does not characterize sensitivity to calibration initialization, measurement noise, or unmodeled mutual coupling.

\subsection{Dataset and Preprocessing}
\label{app:implementation:data}

\textbf{Simulated Dataset.}
We generate~\(35\) simulated scenes using the~Sionna RT differentiable ray tracer~\citep{hoydis2023sionnart}.
The corpus contains seven scene categories: \texttt{conference room}, \texttt{classroom}, \texttt{bedroom}, \texttt{corridor}, \texttt{open office}, \texttt{laboratory}, and \texttt{lounge}, with five variants per category.
We refer to these scenes as~S1--S35 in category order: S1--S5 are conference rooms, S6--S10 classrooms, S11--S15 bedrooms, S16--S20 corridors, S21--S25 open offices, S26--S30 laboratories, and S31--S35 lounges.
The room dimensions are~\(5\times4\times2.8\,\mathrm m\) for bedrooms,~\(6.5\times5.5\times3.2\,\mathrm m\) for lounges,~\(7\times7\times3\,\mathrm m\) for laboratories,~\(8\times6\times3\,\mathrm m\) for conference rooms,~\(10\times7\times3\,\mathrm m\) for classrooms,~\(12\times2.2\times3\,\mathrm m\) for corridors, and~\(13\times9\times3\,\mathrm m\) for open offices.
Surfaces use fixed permittivity and conductivity values at~\(920\,\mathrm{MHz}\) adapted from the~ITU-R P.2040 building-material models, with concrete and brick from revision~3~\citep{itur2023p2040} and plasterboard and glass from revision~1~\citep{itur2015p2040}.
The five variants of each category differ in furniture layout and clutter.
Ray tracing uses up to five interactions per path with specular reflection and refraction.
Each scene contains~\(4.1\)k--\(5.9\)k transmitter locations, with approximately~\(4.9\)k on average, distributed throughout the room volume.
The per-scene median~SNR is sampled uniformly from~\(8\) to~\(18\,\mathrm{dB}\), with a realized range of~\(8.2\) to~\(17.9\,\mathrm{dB}\) and a mean of~\(12.6\,\mathrm{dB}\).
Complex Gaussian noise is added to the per-element channel response using the scene-level median signal to determine the noise level.
All simulated scenes share the same generation pipeline, carrier frequency, material-value table, and interaction model.
The category-exclusion protocol therefore measures semantic scene transfer inside this distribution rather than transfer across simulators or physical propagation regimes.

\textbf{Real Measured Dataset.}
We use the publicly released~NeRF$^2$ RFID measurements~\citep{zhao2023nerf2}, denoted~S36.
The dataset contains~\(6{,}123\) measurements collected with a physical~\(4\times4\) receiver array operating at~\(915\,\mathrm{MHz}\).
Following the released processing pipeline, we construct Bartlett spectra using the same angular-grid convention as for the simulated data.
The real measurements are not used to train the cross-scene model and serve only for simulation-to-real evaluation.

\textbf{Spectrum Construction and Scene Normalization.}
The simulated receiver is a fixed~\(4\times4\) planar array with~\(16\) elements and~\(0.16\,\mathrm m\) spacing.
At~\(920\,\mathrm{MHz}\), the wavelength is~\(\lambda=0.326\,\mathrm m\), corresponding to an element spacing of approximately~\(0.49\lambda\).
For each transmitter, we construct the single-snapshot, phase-only Bartlett magnitude
\begin{equation}
S\left(\Omega\right)
=
\left|
\frac{1}{N}
\sum_{u=1}^{N}
\exp\left(
j\left(
\varphi_u^{\mathrm{st}}\left(\Omega\right)
+
\angle H_u
\right)
\right)
\right|,
\label{eq:app_measurement_spectrum}
\end{equation}
where~\(H_u\) is the channel response at receiver element~\(u\), and~\(\varphi_u^{\mathrm{st}}\left(\Omega\right)\) is its steering phase.
Following~NeRF$^2$~\citep{zhao2023nerf2}, element magnitudes are discarded.
Each spectrum is sampled on a~\(90\times360\) angular grid covering elevations from~\(1^\circ\) to~\(90^\circ\) and azimuths from~\(0^\circ\) to~\(359^\circ\), and is min-max normalized per measurement.
The target spectrum is constructed from a coherent single snapshot, whereas the analytic renderer in~\S\ref{sec:decoder} models propagation components through an incoherent covariance sum.
We treat this difference as a modeling approximation.

Each scene contains the receiver position and orientation~\(\mathbf Q\), together with all transmitter positions.
These coordinates define the query configurations and the canonical normalization, but the reference encoder receives only measured spectra and not transmitter positions associated with individual references.
Before constructing the canonical anchor field, we isotropically rescale scene geometry as
\begin{equation}
\widetilde{\mathbf x}
=
c_s\mathbf x,
\qquad
c_s
=
0.9
\min\left(
\frac{5}{\max_j\left|x_j\right|},
\frac{4}{\max_j\left|y_j\right|},
\frac{3}{\max_j z_j}
\right),
\label{eq:app_scene_scale}
\end{equation}
where~\(j\) ranges over all transmitter and receiver positions in the scene.
This maps each scene into the canonical extent~\(\left[-5,5\right]\times\left[-4,4\right]\times\left[0,3\right]\,\mathrm m\) with a~\(10\%\) margin.
Because the same isotropic factor is applied to all coordinates, angular bearings and receiver orientation are preserved.

\textbf{Data Splits.}
We use the same per-scene~\(70/10/20\) train, validation, and test split for all methods.
For evaluation, reference sets are sampled from the training split and queries from the test split; the validation split remains disjoint from both.
A simulated scene therefore contains approximately~\(3.4\)k training measurements and~\(1.0\)k test measurements, while the real dataset contains~\(4{,}286/612/1{,}225\) train, validation, and test measurements.
The default reference budget is~\(M=32\).

For generalization to unseen scene variants, we use five folds, each excluding one variant from every scene category.
For generalization to unseen scene categories, we use seven folds, each excluding all five variants of one category.
For simulation-to-real evaluation, the cross-scene model is trained on all~\(35\) simulated scenes and evaluated on the real measurements.
All methods use identical splits and reference sets where applicable.

\section{Additional Experimental Results and Analyses}
\label{app:results}

\textbf{Roadmap.}
The first six subsections provide qualitative results~(\S\ref{app:qualitative}), scene-level comparisons~(\S\ref{app:breakdown}), reliability analysis~(\S\ref{app:reliability_analysis}), receiver-array adaptation~(\S\ref{app:array_error}), a localization case study~(\S\ref{app:localization}), and real-data transfer~(\S\ref{app:real}).
The remaining subsections examine reference evidence~(\S\ref{app:reference_analysis}), representation and renderer choices~(\S\ref{app:representation_choices}), deployment efficiency~(\S\ref{app:deployment_efficiency}), and physical distribution shift~(\S\ref{app:physical_stress}).
The main-paper synthesis tables and their scene-level breakdowns report~\ourSystemPR; other analyses specify the evaluated variant.

\textbf{Shared Protocol.}
Unless stated otherwise, simulated experiments use the same~\(35\) scenes, \(4\times4\) array, \(920\,\mathrm{MHz}\) carrier, and \(90\times360\) normalized-amplitude spectra as the main experiments.
Reported cross-fold deviations reflect variation among held-out scenes.
Across three seeds on all~\(12\) folds, the within-fold PSNR standard deviation averages~\(0.21\,\mathrm{dB}\) and reaches at most~\(0.73\,\mathrm{dB}\).

\begin{figure}[p]
\centering
\subfigure[Per-scene]{\includegraphics[width=\linewidth]{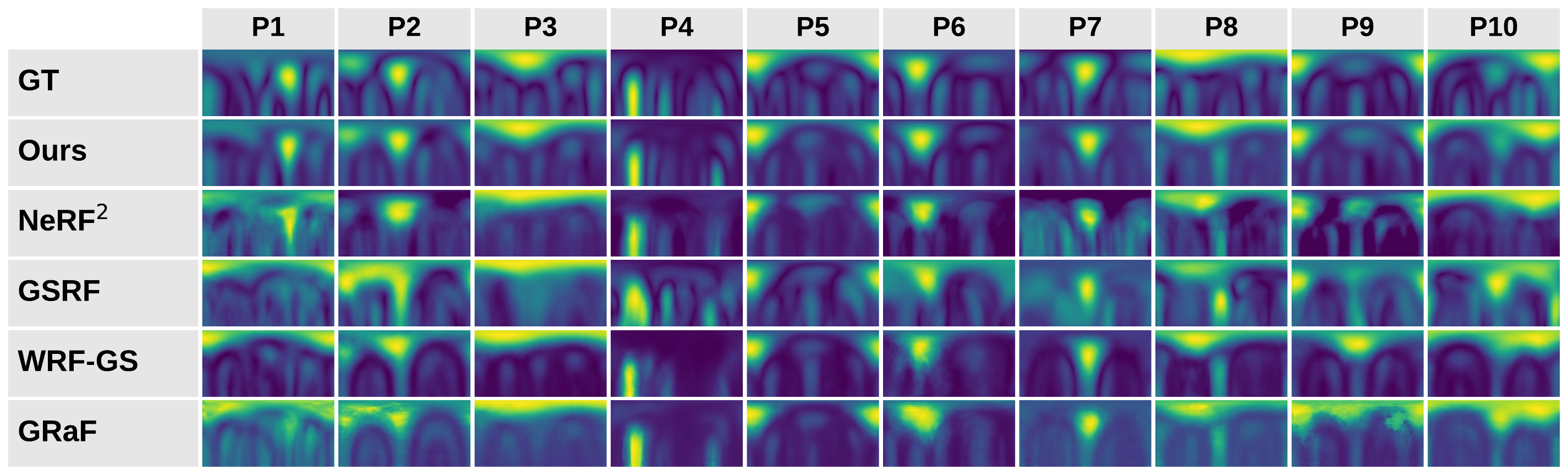}\label{fig:qual_panels_perscene}}\\[4pt]
\subfigure[Unseen scene variants]{\includegraphics[width=\linewidth]{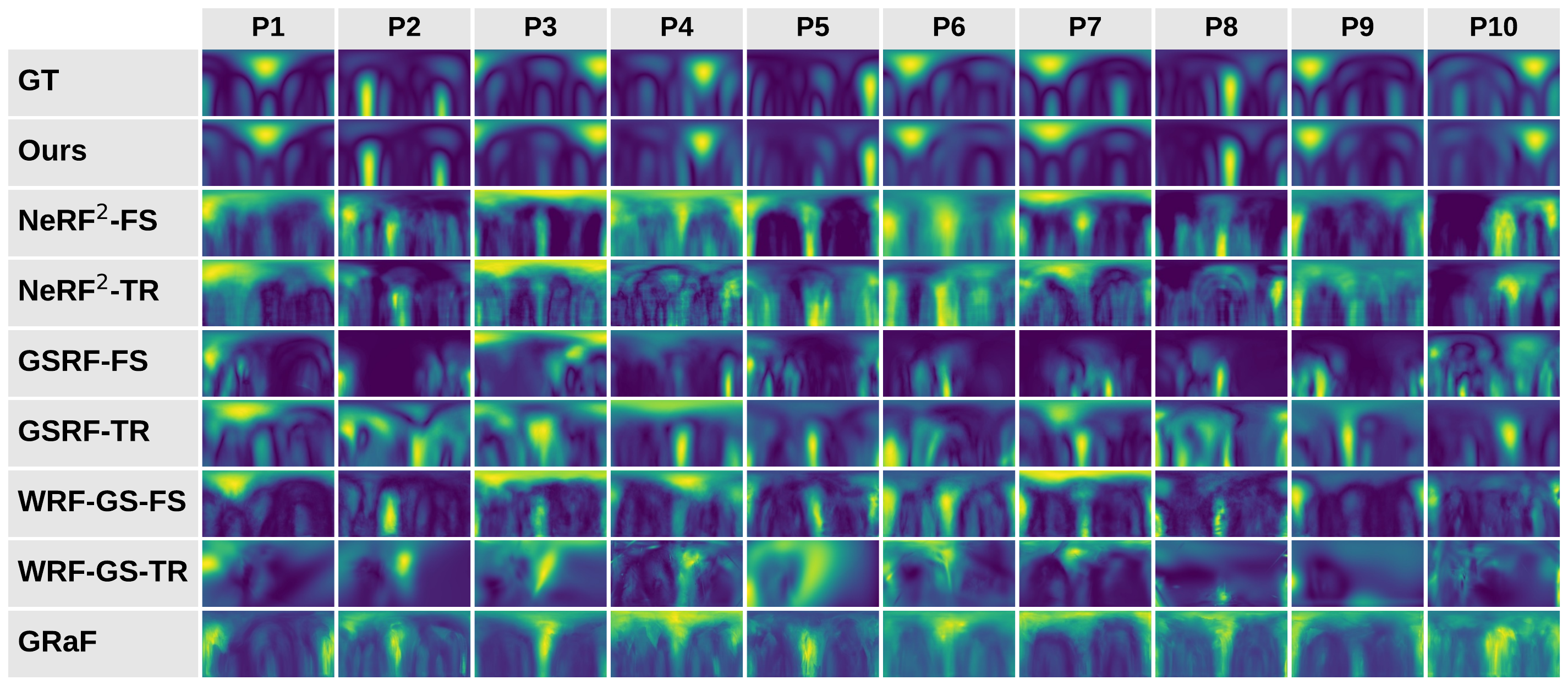}\label{fig:qual_panels_variant}}\\[4pt]
\subfigure[Unseen scene categories]{\includegraphics[width=\linewidth]{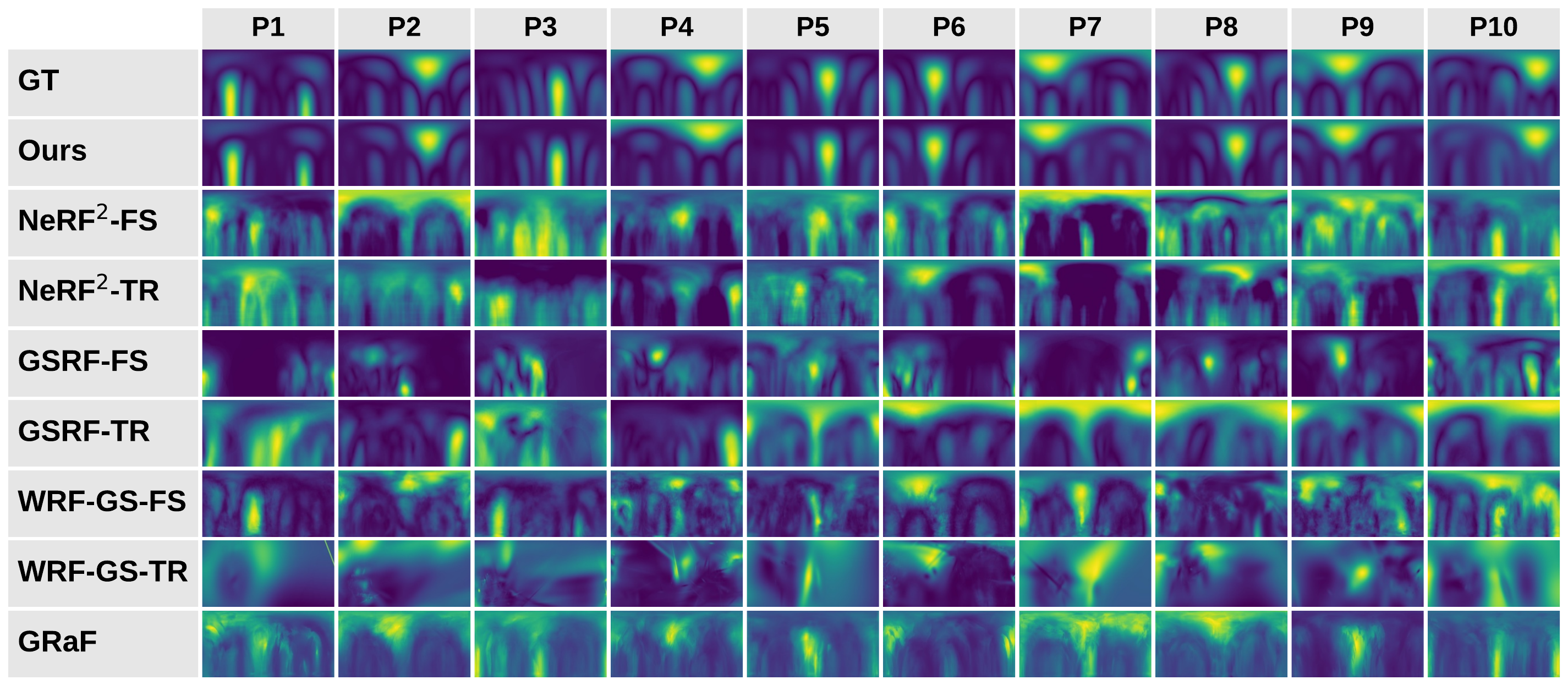}\label{fig:qual_panels_category}}
\caption{\textit{Predicted spatial spectra across the three evaluation settings.}
Rows compare ground truth~(GT), \ourSystemPR~(Ours), and the corresponding baselines.}
\label{fig:qual_panels}
\end{figure}

\subsection{Qualitative Results}
\label{app:qualitative}

This subsection presents qualitative spectra and reflection-dominated failures.

\autoref{fig:qual_panels} compares predictions for per-scene fitting, unseen scene variants, and unseen scene categories.
Each panel shows ten selected test locations~(P1--P10), with at most one per scene.
Rows show the ground-truth~(GT) spectrum, \ourSystemPR~(``Ours''), and the corresponding baselines.
The figures use the saved predictions underlying the quantitative results.
For display, each spectrum is normalized by its maximum; azimuth is on the horizontal axis and elevation on the vertical axis.

\subsubsection{Per-Scene Results}
\label{app:qualitative:perscene}

Figure~\ref{fig:qual_panels_perscene} shows per-scene synthesis, with each method using the same scene-specific training and test splits; location~P1 comes from the measured scene~S36.
The baselines generally recover the dominant lobe but often blur or misplace weaker multipath structure, and some produce spurious high-response bands near the upper elevation boundary.
Across the displayed locations, their~SSIM ranges from~\(0.51\) to~\(0.58\), compared with~\(0.74\) to~\(0.77\) for~\ourSystemPR.

\subsubsection{Unseen Scene Variants}
\label{app:qualitative:variants}

As illustrated in~Figure~\ref{fig:qual_panels_variant}, the few-shot baselines often misplace lobes or collapse multipath structure into a diffuse region with only~\(M=32\) target-scene references.
Transfer from another scene in the same category does not reliably correct these errors.
GRaF produces low-contrast, vertically streaked spectra that sometimes approximately localize the strongest lobe.
In contrast,~\ourSystemPR better preserves dominant and secondary lobe directions; its remaining errors mainly concern their relative strengths and the weakest components.
Across the displayed locations, the strongest baseline achieves~\(0.29\)--\(0.38\)~SSIM.

\subsubsection{Unseen Scene Categories}
\label{app:qualitative:categories}

The qualitative trend persists when an entire scene category is excluded from training, as illustrated in~Figure~\ref{fig:qual_panels_category}.
The baselines, including models transferred from other categories, often misplace or blur multipath lobes and achieve only~\(0.29\)--\(0.37\)~SSIM across the displayed locations.
In contrast,~\ourSystemPR continues to preserve the dominant lobe directions; its remaining errors mainly concern their relative strengths.

\subsubsection{Reflection-Dominated Failure Cases}
\label{app:qualitative:fails}

This diagnostic uses feed-forward~\ourSystem on the first unseen-variant fold; the preceding panels show~\ourSystemPR.
Across~\(6{,}436\) test queries, mean~SSIM is~\(0.804\), with~\(95.0\%\) above~\(0.5\) and only~\(0.3\%\) below~\(0.3\).
We examine the worst~\(5\%\): \(322\) queries with~SSIM~\(\leq0.50\).
\autoref{fig:failure_gallery} shows the ten lowest-SSIM predictions.
All ten targets have multiple narrow lobes and a strongest arrival away from the direct path, while~\ourSystem generally predicts one or two broad lobes.
The strongest baseline also falls from~\(0.60\)~SSIM overall to~\(0.46\) on the worst~\(5\%\).
On the ten displayed queries, it achieves~\(0.23\)--\(0.41\)~SSIM~\versus~\(0.16\)--\(0.26\) for~\ourSystem; neither method captures the full reflected structure.

The worst-\(5\%\) queries are predominantly reflection dominated.
A spectrum formed from the true line-of-sight direction alone has median~SSIM~\(0.32\) on this tail~\versus~\(0.82\) on the remaining queries.
The strongest ground-truth lobe lies a median~\(23.6^\circ\) from line of sight~\versus~\(2.1^\circ\), while the spectral energy within~\(20^\circ\) of line of sight falls from~\(62\%\) to~\(32\%\).
The diffuse floor rises from~\(0.13\) to~\(0.28\).
Across all queries, prediction~SSIM correlates strongly with line-of-sight explainability, with Spearman~\(\rho=0.96\).
The errors also have a consistent structure.
Ground-truth spectra in the worst~\(5\%\) contain a median of three lobes with azimuth width~\(26^\circ\), whereas~\ourSystem typically predicts one broader lobe of width~\(56^\circ\).
Nearest-reference distance is comparable to that of a random control set: \(0.76\)~\versus~\(0.85\,\mathrm m\).
The best-matching reference is less similar for failure cases, with~SSIM~\(0.52\)~\versus~\(0.69\).
Thus, the difficulty is associated more with poorly matched reference spectra than with reference distance alone; narrow reflected arrivals tend to be smoothed into broader components.

\begin{figure}[t]
\centering
\includegraphics[width=\linewidth]{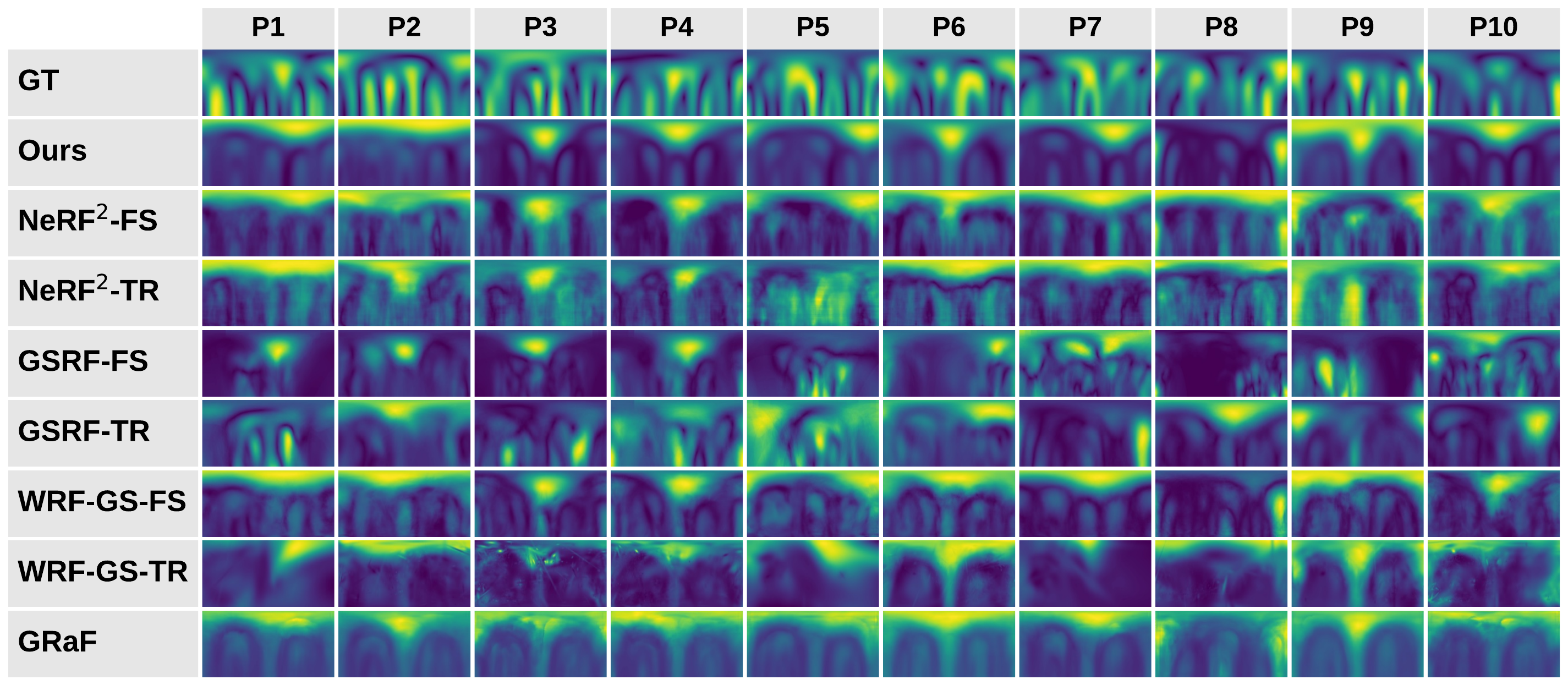}
\caption{\textit{Worst-case predictions.}
The ten lowest-SSIM predictions of feed-forward~\ourSystem in the first unseen-variant fold are shown with all baselines on the same queries.
}
\label{fig:failure_gallery}
\end{figure}

\subsection{Scene-Level Spatial Spectrum Synthesis Results}
\label{app:breakdown}

\textbf{Crosswalk to the Main Results.}
The per-scene averages in~\autoref{tab:e1perscene} reproduce~\autoref{tab:perscene} after rounding.
The unseen-variant and unseen-category breakdowns correspond to the~\ourSystemPR results in~\autoref{tab:seen} and~\autoref{tab:heldout}.
Those main tables aggregate over folds, whereas the breakdowns pool queries within scenes; unweighted averages of the scene rows may therefore differ.

\begin{table}[h!]
\centering\small
\caption{\textit{Per-scene results in the per-scene setting.}
Results are reported for all 36 scenes and all four metrics. S36 denotes the real measured dataset, and the averages correspond to \autoref{tab:perscene}.}
\label{tab:e1perscene}
\setlength{\tabcolsep}{2.5pt}
\resizebox{\linewidth}{!}{%
\begin{tabular}{lcccccccccccccccccccc}
\toprule
& \multicolumn{5}{c}{PSNR\(\uparrow\)} & \multicolumn{5}{c}{SSIM\(\uparrow\)} & \multicolumn{5}{c}{NMSE\(\downarrow\)} & \multicolumn{5}{c}{AoA err (\(^\circ\))\(\downarrow\)} \\
\cmidrule(lr){2-6}\cmidrule(lr){7-11}\cmidrule(lr){12-16}\cmidrule(lr){17-21}
Scene & NeRF\(^2\) & GSRF & WRF-GS & GRaF & \ourSystemPR & NeRF\(^2\) & GSRF & WRF-GS & GRaF & \ourSystemPR & NeRF\(^2\) & GSRF & WRF-GS & GRaF & \ourSystemPR & NeRF\(^2\) & GSRF & WRF-GS & GRaF & \ourSystemPR \\
\midrule
\multicolumn{21}{l}{\emph{Conference}} \\
\quad S1 & 27.13 & 23.97 & 26.84 & 27.06 & \textbf{29.17} & 0.847 & 0.763 & 0.832 & 0.816 & \textbf{0.891} & 0.030 & 0.071 & 0.034 & 0.030 & \textbf{0.021} & 3.10 & 5.89 & 3.29 & 2.47 & \textbf{2.00} \\
\quad S2 & 25.11 & 22.69 & 24.65 & 25.35 & \textbf{27.61} & 0.801 & 0.735 & 0.794 & 0.781 & \textbf{0.865} & 0.057 & 0.097 & 0.059 & 0.053 & \textbf{0.040} & 4.69 & 7.91 & 4.91 & 3.80 & \textbf{3.21} \\
\quad S3 & 19.46 & 19.12 & 20.06 & 19.55 & \textbf{22.03} & 0.644 & 0.642 & 0.638 & 0.612 & \textbf{0.742} & 0.152 & 0.170 & 0.150 & 0.151 & \textbf{0.104} & 11.46 & 11.30 & 11.18 & 9.12 & \textbf{7.25} \\
\quad S4 & 19.27 & 18.73 & 19.15 & 19.56 & \textbf{22.00} & 0.638 & 0.635 & 0.626 & 0.615 & \textbf{0.749} & 0.146 & 0.171 & 0.155 & 0.139 & \textbf{0.096} & 10.53 & 9.64 & 10.74 & 8.14 & \textbf{7.51} \\
\quad S5 & 18.06 & 18.16 & 19.95 & 19.59 & \textbf{21.52} & 0.590 & 0.612 & 0.666 & 0.626 & \textbf{0.731} & 0.219 & 0.224 & 0.162 & 0.158 & \textbf{0.130} & 15.53 & 14.01 & 12.85 & 11.85 & \textbf{11.07} \\
\multicolumn{21}{l}{\emph{Classroom}} \\
\quad S6 & 23.51 & 21.83 & 23.32 & 23.21 & \textbf{26.25} & 0.770 & 0.718 & 0.747 & 0.702 & \textbf{0.841} & 0.063 & 0.092 & 0.067 & 0.067 & \textbf{0.039} & 4.08 & 5.05 & 4.40 & 3.08 & \textbf{2.74} \\
\quad S7 & 29.27 & 23.86 & 29.70 & 29.50 & \textbf{31.83} & 0.875 & 0.749 & 0.877 & 0.855 & \textbf{0.917} & 0.023 & 0.088 & 0.022 & 0.021 & \textbf{0.015} & 3.15 & 7.01 & 2.70 & 2.13 & \textbf{1.77} \\
\quad S8 & 22.63 & 20.89 & 22.45 & 22.64 & \textbf{25.29} & 0.735 & 0.688 & 0.721 & 0.704 & \textbf{0.811} & 0.091 & 0.126 & 0.096 & 0.091 & \textbf{0.068} & 7.47 & 7.59 & 6.85 & 6.81 & \textbf{5.36} \\
\quad S9 & 24.83 & 21.96 & 25.71 & 24.90 & \textbf{27.97} & 0.780 & 0.702 & 0.801 & 0.766 & \textbf{0.854} & 0.090 & 0.133 & 0.069 & 0.073 & \textbf{0.050} & 7.87 & 11.04 & 6.31 & 5.95 & \textbf{4.51} \\
\quad S10 & 21.36 & 19.73 & 21.20 & 20.55 & \textbf{23.78} & 0.691 & 0.651 & 0.666 & 0.616 & \textbf{0.778} & 0.128 & 0.177 & 0.125 & 0.140 & \textbf{0.080} & 11.06 & 11.05 & 9.05 & 8.89 & \textbf{7.42} \\
\multicolumn{21}{l}{\emph{Bedroom}} \\
\quad S11 & 28.40 & 25.81 & 27.01 & 28.41 & \textbf{30.02} & 0.863 & 0.802 & 0.822 & 0.836 & \textbf{0.902} & 0.024 & 0.045 & 0.033 & 0.024 & \textbf{0.017} & 2.64 & 4.43 & 2.92 & 2.00 & \textbf{1.62} \\
\quad S12 & 25.47 & 23.56 & 24.73 & 24.86 & \textbf{27.87} & 0.813 & 0.753 & 0.780 & 0.749 & \textbf{0.868} & 0.043 & 0.076 & 0.050 & 0.050 & \textbf{0.029} & 3.54 & 4.88 & 3.53 & 2.58 & \textbf{2.18} \\
\quad S13 & 23.23 & 22.30 & 23.26 & 22.70 & \textbf{25.63} & 0.760 & 0.724 & 0.746 & 0.695 & \textbf{0.829} & 0.073 & 0.098 & 0.076 & 0.085 & \textbf{0.049} & 5.02 & 5.77 & 4.66 & 4.27 & \textbf{3.61} \\
\quad S14 & 22.86 & 21.84 & 22.50 & 22.71 & \textbf{24.64} & 0.749 & 0.715 & 0.733 & 0.712 & \textbf{0.811} & 0.098 & 0.117 & 0.094 & 0.090 & \textbf{0.066} & 9.24 & 8.80 & 8.87 & 8.78 & \textbf{7.42} \\
\quad S15 & 23.79 & 22.62 & 23.49 & 23.56 & \textbf{26.17} & 0.766 & 0.733 & 0.747 & 0.730 & \textbf{0.832} & 0.073 & 0.093 & 0.077 & 0.076 & \textbf{0.059} & 6.10 & 6.36 & 6.50 & 5.50 & \textbf{4.61} \\
\multicolumn{21}{l}{\emph{Corridor}} \\
\quad S16 & 15.14 & 15.02 & 16.04 & 16.56 & \textbf{17.44} & 0.534 & 0.557 & 0.581 & 0.579 & \textbf{0.655} & 0.284 & 0.305 & 0.236 & 0.206 & \textbf{0.186} & 22.53 & 21.23 & 20.02 & 19.19 & \textbf{19.11} \\
\quad S17 & 21.56 & 20.51 & 21.91 & 21.23 & \textbf{23.45} & 0.690 & 0.662 & 0.670 & 0.629 & \textbf{0.767} & 0.123 & 0.147 & 0.123 & 0.127 & \textbf{0.085} & 19.50 & 21.33 & 24.39 & 15.71 & \textbf{11.87} \\
\quad S18 & 15.55 & 15.54 & 16.89 & 16.96 & \textbf{17.69} & 0.531 & 0.562 & 0.597 & 0.588 & \textbf{0.653} & 0.292 & 0.300 & 0.230 & 0.226 & \textbf{0.204} & 17.35 & 16.44 & 15.93 & 15.48 & \textbf{14.70} \\
\quad S19 & 17.89 & 18.41 & 19.14 & 19.00 & \textbf{20.39} & 0.557 & 0.596 & 0.597 & 0.571 & \textbf{0.676} & 0.244 & 0.219 & 0.210 & 0.193 & \textbf{0.157} & 29.30 & 29.87 & 22.57 & 23.04 & \textbf{20.48} \\
\quad S20 & 17.25 & 17.65 & 17.91 & 18.45 & \textbf{20.35} & 0.537 & 0.575 & 0.556 & 0.554 & \textbf{0.687} & 0.253 & 0.235 & 0.238 & 0.201 & \textbf{0.146} & 29.68 & 26.42 & 20.72 & 21.78 & \textbf{18.57} \\
\multicolumn{21}{l}{\emph{Open office}} \\
\quad S21 & 23.61 & 20.62 & 23.47 & 23.20 & \textbf{25.79} & 0.767 & 0.689 & 0.743 & 0.717 & \textbf{0.829} & 0.069 & 0.143 & 0.077 & 0.078 & \textbf{0.048} & 4.44 & 8.66 & 4.22 & 3.68 & \textbf{3.00} \\
\quad S22 & 25.60 & 22.70 & 24.85 & 24.89 & \textbf{28.12} & 0.806 & 0.735 & 0.781 & 0.745 & \textbf{0.867} & 0.049 & 0.093 & 0.057 & 0.056 & \textbf{0.030} & 4.93 & 8.12 & 4.56 & 3.25 & \textbf{2.43} \\
\quad S23 & 24.20 & 21.68 & 24.28 & 24.04 & \textbf{26.51} & 0.784 & 0.714 & 0.766 & 0.742 & \textbf{0.842} & 0.060 & 0.118 & 0.061 & 0.061 & \textbf{0.042} & 4.76 & 6.13 & 4.75 & 3.70 & \textbf{3.30} \\
\quad S24 & 19.63 & 18.55 & 20.41 & 20.23 & \textbf{23.12} & 0.649 & 0.635 & 0.681 & 0.663 & \textbf{0.773} & 0.147 & 0.191 & 0.137 & 0.136 & \textbf{0.091} & 10.21 & 10.69 & 9.27 & 9.40 & \textbf{7.37} \\
\quad S25 & 19.85 & 18.90 & 21.49 & 20.74 & \textbf{23.73} & 0.628 & 0.625 & 0.682 & 0.626 & \textbf{0.763} & 0.174 & 0.209 & 0.141 & 0.145 & \textbf{0.096} & 18.11 & 22.11 & 12.46 & 16.20 & \textbf{10.19} \\
\multicolumn{21}{l}{\emph{Lab}} \\
\quad S26 & 21.70 & 21.08 & 21.85 & 22.19 & \textbf{24.45} & 0.719 & 0.698 & 0.712 & 0.683 & \textbf{0.804} & 0.102 & 0.115 & 0.101 & 0.094 & \textbf{0.067} & 6.37 & 6.31 & 6.68 & 4.91 & \textbf{4.38} \\
\quad S27 & 19.53 & 19.59 & 20.71 & 20.10 & \textbf{22.69} & 0.637 & 0.659 & 0.672 & 0.619 & \textbf{0.766} & 0.147 & 0.153 & 0.123 & 0.136 & \textbf{0.083} & 8.37 & 8.30 & 7.72 & 7.10 & \textbf{5.66} \\
\quad S28 & 18.70 & 18.64 & 19.10 & 19.39 & \textbf{21.68} & 0.619 & 0.629 & 0.612 & 0.604 & \textbf{0.740} & 0.181 & 0.186 & 0.172 & 0.156 & \textbf{0.111} & 15.10 & 13.65 & 13.21 & 11.29 & \textbf{10.37} \\
\quad S29 & 18.58 & 18.42 & 19.46 & 19.07 & \textbf{21.26} & 0.615 & 0.623 & 0.626 & 0.603 & \textbf{0.728} & 0.176 & 0.183 & 0.160 & 0.166 & \textbf{0.115} & 14.69 & 12.72 & 11.46 & 10.35 & \textbf{9.54} \\
\quad S30 & 15.57 & 16.76 & 17.85 & 18.11 & \textbf{19.42} & 0.477 & 0.566 & 0.591 & 0.587 & \textbf{0.666} & 0.385 & 0.329 & 0.266 & 0.256 & \textbf{0.222} & 28.90 & 27.80 & 29.10 & 27.38 & \textbf{24.04} \\
\multicolumn{21}{l}{\emph{Lounge}} \\
\quad S31 & 25.03 & 22.96 & 25.45 & 25.93 & \textbf{27.83} & 0.808 & 0.742 & 0.806 & 0.805 & \textbf{0.871} & 0.052 & 0.083 & 0.047 & 0.043 & \textbf{0.033} & 3.95 & 5.29 & 3.85 & 2.81 & \textbf{2.35} \\
\quad S32 & 23.75 & 21.68 & 23.37 & 23.90 & \textbf{26.61} & 0.769 & 0.706 & 0.761 & 0.728 & \textbf{0.842} & 0.070 & 0.107 & 0.075 & 0.070 & \textbf{0.047} & 5.15 & 7.84 & 5.92 & 3.73 & \textbf{3.38} \\
\quad S33 & 19.79 & 19.99 & 21.95 & 21.76 & \textbf{24.18} & 0.639 & 0.663 & 0.697 & 0.677 & \textbf{0.794} & 0.148 & 0.158 & 0.109 & 0.108 & \textbf{0.079} & 8.21 & 7.73 & 6.40 & 5.71 & \textbf{5.01} \\
\quad S34 & 20.21 & 20.25 & 21.11 & 21.55 & \textbf{23.60} & 0.680 & 0.685 & 0.716 & 0.710 & \textbf{0.789} & 0.153 & 0.147 & 0.130 & 0.116 & \textbf{0.103} & 15.55 & 14.77 & 14.97 & 12.29 & \textbf{11.95} \\
\quad S35 & 18.14 & 18.47 & 19.04 & 18.66 & \textbf{21.31} & 0.599 & 0.620 & 0.605 & 0.583 & \textbf{0.726} & 0.207 & 0.198 & 0.176 & 0.179 & \textbf{0.130} & 19.10 & 15.31 & 13.48 & 12.68 & \textbf{12.26} \\
\midrule
\multicolumn{21}{l}{\emph{Real measured dataset}} \\
\quad S36 & 21.00 & 21.59 & 19.23 & 20.31 & \textbf{22.18} & 0.755 & 0.803 & 0.732 & 0.776 & \textbf{0.806} & 0.052 & 0.051 & 0.079 & 0.067 & \textbf{0.043} & 5.17 & 5.38 & 5.45 & 5.61 & \textbf{4.12} \\
\bottomrule
\end{tabular}}
\end{table}

\subsubsection{Per-Scene Results}
\label{app:breakdown_perscene}

\autoref{tab:e1perscene} reports all four metrics for each of the~\(36\) scenes, grouped by category.
Queries are pooled within scenes, and the table averages reproduce~\autoref{tab:perscene} after rounding.
Because every method is trained and evaluated within the same scene, this setting measures spectrum-synthesis fidelity without testing cross-scene generalization.

\textbf{The Advantage Is Consistent Across Scenes.}
\ourSystemPR leads on every scene and metric, covering all~\(144\) scene-metric comparisons.
Thus, the aggregate improvement is not driven by a few favorable scenes.
Relative to the strongest baseline for each scene and metric, the average gains are~\(2.02\,\mathrm{dB}\) in~PSNR and~\(0.073\) in~SSIM, with a~\(1.07^\circ\) reduction in~AoA error.
The largest PSNR gain is~\(2.74\,\mathrm{dB}\) on classroom scene~S6, while the largest SSIM gain is~\(0.112\) on corridor scene~S20.
The margin is smallest on measured scene~S36, where~\ourSystemPR reaches~\(0.806\)~SSIM against~\(0.803\) for~GSRF.

\textbf{Scene Difficulty Is Consistent Across Methods.}
Per-scene~SSIM is strongly correlated between~\ourSystemPR and each baseline, with~\(r=0.93\)--\(0.98\).
Corridors are among the most difficult scenes, with~\ourSystemPR achieving~\(0.653\)--\(0.767\)~SSIM, followed by laboratories.
Several bedrooms and classrooms are easier, reaching~\(0.917\)~SSIM on~S7.
These patterns suggest that scene properties contribute substantially to performance variation across methods.

\textbf{Scene Variants Produce Substantial Variation.}
Across categories, the SSIM spread among five variants is~\(0.09\)--\(0.16\), comparable to or larger than the average~\(0.073\) gap between~\ourSystemPR and the strongest baseline.
Layouts and clutter therefore matter even within a category, motivating the scene-level breakdown rather than relying on category averages alone.

\begin{table}[t]
\centering\small
\caption{\textit{Scene-level results for unseen scene variants with few-shot baselines.}
Each scene is evaluated by the within-category fold that excludes it using \(M{=}32\) references.
}
\label{tab:e2perscene}
\setlength{\tabcolsep}{2.5pt}
\resizebox{\linewidth}{!}{%
\begin{tabular}{lcccccccccccccccccccc}
\toprule
& \multicolumn{5}{c}{PSNR\(\uparrow\)} & \multicolumn{5}{c}{SSIM\(\uparrow\)} & \multicolumn{5}{c}{NMSE\(\downarrow\)} & \multicolumn{5}{c}{AoA err (\(^\circ\))\(\downarrow\)} \\
\cmidrule(lr){2-6}\cmidrule(lr){7-11}\cmidrule(lr){12-16}\cmidrule(lr){17-21}
Scene & NeRF\(^2\)-FS & GSRF-FS & WRF-GS-FS & GRaF & \ourSystemPR & NeRF\(^2\)-FS & GSRF-FS & WRF-GS-FS & GRaF & \ourSystemPR & NeRF\(^2\)-FS & GSRF-FS & WRF-GS-FS & GRaF & \ourSystemPR & NeRF\(^2\)-FS & GSRF-FS & WRF-GS-FS & GRaF & \ourSystemPR \\
\midrule
\multicolumn{21}{l}{\emph{Conference}} \\
\quad S1 & 15.18 & 12.09 & 17.13 & 15.53 & \textbf{28.44} & 0.431 & 0.302 & 0.538 & 0.487 & \textbf{0.876} & 0.499 & 0.805 & 0.354 & 0.418 & \textbf{0.023} & 14.40 & 51.26 & 16.62 & 20.27 & \textbf{2.10} \\
\quad S2 & 15.26 & 11.66 & 17.12 & 13.83 & \textbf{25.61} & 0.480 & 0.301 & 0.549 & 0.439 & \textbf{0.811} & 0.427 & 0.850 & 0.287 & 0.558 & \textbf{0.069} & 16.48 & 50.03 & 14.66 & 26.26 & \textbf{4.40} \\
\quad S3 & 12.80 & 12.16 & 15.68 & 13.73 & \textbf{21.12} & 0.371 & 0.336 & 0.485 & 0.441 & \textbf{0.710} & 0.657 & 0.687 & 0.364 & 0.518 & \textbf{0.128} & 25.35 & 54.77 & 21.14 & 25.87 & \textbf{8.17} \\
\quad S4 & 14.11 & 11.99 & 15.17 & 14.10 & \textbf{20.97} & 0.444 & 0.336 & 0.484 & 0.444 & \textbf{0.712} & 0.491 & 0.694 & 0.380 & 0.460 & \textbf{0.117} & 19.19 & 52.14 & 17.87 & 25.50 & \textbf{7.54} \\
\quad S5 & 13.87 & 12.15 & 14.92 & 14.45 & \textbf{19.10} & 0.412 & 0.329 & 0.473 & 0.468 & \textbf{0.643} & 0.538 & 0.740 & 0.424 & 0.452 & \textbf{0.227} & 23.96 & 55.71 & 23.81 & 27.87 & \textbf{14.92} \\
\multicolumn{21}{l}{\emph{Classroom}} \\
\quad S6 & 15.71 & 11.60 & 17.67 & 15.25 & \textbf{25.62} & 0.456 & 0.284 & 0.562 & 0.500 & \textbf{0.824} & 0.420 & 0.798 & 0.275 & 0.451 & \textbf{0.042} & 12.99 & 48.06 & 12.33 & 15.24 & \textbf{2.78} \\
\quad S7 & 14.25 & 11.71 & 17.83 & 13.61 & \textbf{30.40} & 0.412 & 0.281 & 0.566 & 0.443 & \textbf{0.897} & 0.644 & 0.944 & 0.313 & 0.668 & \textbf{0.021} & 23.38 & 67.48 & 14.24 & 24.85 & \textbf{1.92} \\
\quad S8 & 14.80 & 11.45 & 16.47 & 13.53 & \textbf{24.42} & 0.481 & 0.304 & 0.532 & 0.465 & \textbf{0.787} & 0.506 & 0.836 & 0.380 & 0.659 & \textbf{0.082} & 16.42 & 63.04 & 19.94 & 23.45 & \textbf{5.33} \\
\quad S9 & 13.59 & 12.25 & 17.13 & 14.79 & \textbf{25.27} & 0.426 & 0.325 & 0.513 & 0.449 & \textbf{0.803} & 0.786 & 0.853 & 0.360 & 0.579 & \textbf{0.072} & 19.83 & 68.48 & 17.20 & 26.54 & \textbf{4.96} \\
\quad S10 & 10.99 & 11.73 & 15.49 & 13.47 & \textbf{22.89} & 0.300 & 0.323 & 0.463 & 0.433 & \textbf{0.746} & 1.348 & 0.845 & 0.445 & 0.639 & \textbf{0.104} & 34.94 & 57.91 & 23.82 & 25.16 & \textbf{7.67} \\
\multicolumn{21}{l}{\emph{Bedroom}} \\
\quad S11 & 14.89 & 14.72 & 16.96 & 13.92 & \textbf{29.28} & 0.446 & 0.392 & 0.514 & 0.415 & \textbf{0.888} & 0.577 & 0.533 & 0.407 & 0.641 & \textbf{0.020} & 16.31 & 34.36 & 20.66 & 23.42 & \textbf{1.65} \\
\quad S12 & 15.03 & 12.99 & 15.81 & 13.10 & \textbf{26.74} & 0.458 & 0.339 & 0.472 & 0.408 & \textbf{0.848} & 0.533 & 0.690 & 0.434 & 0.706 & \textbf{0.034} & 14.94 & 39.90 & 19.30 & 27.25 & \textbf{2.20} \\
\quad S13 & 14.53 & 13.39 & 16.59 & 13.97 & \textbf{23.94} & 0.433 & 0.375 & 0.512 & 0.430 & \textbf{0.785} & 0.566 & 0.622 & 0.384 & 0.590 & \textbf{0.070} & 19.17 & 35.24 & 16.97 & 25.72 & \textbf{3.80} \\
\quad S14 & 14.73 & 13.62 & 16.31 & 14.24 & \textbf{22.29} & 0.462 & 0.381 & 0.510 & 0.442 & \textbf{0.748} & 0.504 & 0.573 & 0.374 & 0.507 & \textbf{0.106} & 17.06 & 35.44 & 18.41 & 28.48 & \textbf{8.42} \\
\quad S15 & 14.23 & 12.72 & 15.51 & 13.78 & \textbf{25.45} & 0.428 & 0.339 & 0.468 & 0.421 & \textbf{0.806} & 0.567 & 0.713 & 0.423 & 0.606 & \textbf{0.067} & 16.51 & 48.30 & 23.63 & 25.01 & \textbf{4.56} \\
\multicolumn{21}{l}{\emph{Corridor}} \\
\quad S16 & 12.04 & 11.68 & 12.92 & 13.19 & \textbf{15.36} & 0.402 & 0.349 & 0.420 & 0.460 & \textbf{0.573} & 0.621 & 0.579 & 0.494 & 0.477 & \textbf{0.287} & 31.91 & 48.11 & 31.30 & 27.52 & \textbf{22.37} \\
\quad S17 & 14.16 & 14.16 & 16.33 & 14.10 & \textbf{21.90} & 0.411 & 0.395 & 0.487 & 0.441 & \textbf{0.712} & 0.633 & 0.561 & 0.399 & 0.594 & \textbf{0.116} & 40.63 & 55.37 & 33.01 & 46.88 & \textbf{13.23} \\
\quad S18 & 11.99 & 11.25 & 13.86 & 12.79 & \textbf{14.81} & 0.349 & 0.318 & 0.445 & 0.421 & \textbf{0.532} & 0.648 & 0.681 & 0.422 & 0.590 & \textbf{0.348} & 24.90 & 46.52 & 21.72 & 24.85 & \textbf{19.94} \\
\quad S19 & 11.76 & 13.59 & 14.54 & 13.53 & \textbf{19.69} & 0.290 & 0.350 & 0.421 & 0.383 & \textbf{0.643} & 0.871 & 0.521 & 0.454 & 0.558 & \textbf{0.176} & 57.47 & 48.95 & 43.13 & 52.41 & \textbf{21.03} \\
\quad S20 & 12.93 & 12.94 & 14.46 & 13.56 & \textbf{18.21} & 0.342 & 0.360 & 0.417 & 0.384 & \textbf{0.610} & 0.651 & 0.619 & 0.457 & 0.559 & \textbf{0.270} & 38.39 & 61.16 & 34.75 & 45.70 & \textbf{24.02} \\
\multicolumn{21}{l}{\emph{Open office}} \\
\quad S21 & 14.00 & 11.76 & 17.30 & 14.71 & \textbf{25.19} & 0.405 & 0.348 & 0.550 & 0.502 & \textbf{0.813} & 0.562 & 0.767 & 0.271 & 0.478 & \textbf{0.055} & 16.61 & 53.11 & 15.08 & 17.44 & \textbf{3.24} \\
\quad S22 & 13.59 & 13.04 & 17.56 & 15.03 & \textbf{27.33} & 0.452 & 0.398 & 0.559 & 0.478 & \textbf{0.850} & 0.767 & 0.682 & 0.343 & 0.536 & \textbf{0.036} & 15.52 & 42.97 & 16.49 & 25.00 & \textbf{2.55} \\
\quad S23 & 15.05 & 10.79 & 17.91 & 14.02 & \textbf{25.77} & 0.474 & 0.279 & 0.576 & 0.484 & \textbf{0.820} & 0.486 & 0.951 & 0.269 & 0.598 & \textbf{0.047} & 13.85 & 63.81 & 11.32 & 21.62 & \textbf{3.35} \\
\quad S24 & 11.90 & 11.03 & 14.31 & 12.75 & \textbf{20.12} & 0.374 & 0.296 & 0.465 & 0.436 & \textbf{0.686} & 0.817 & 0.830 & 0.466 & 0.677 & \textbf{0.182} & 27.30 & 61.17 & 23.52 & 27.91 & \textbf{12.10} \\
\quad S25 & 13.15 & 12.03 & 15.37 & 14.26 & \textbf{20.93} & 0.380 & 0.349 & 0.464 & 0.444 & \textbf{0.686} & 0.633 & 0.783 & 0.438 & 0.485 & \textbf{0.191} & 43.08 & 62.38 & 29.35 & 32.67 & \textbf{15.18} \\
\multicolumn{21}{l}{\emph{Lab}} \\
\quad S26 & 14.17 & 13.07 & 15.49 & 14.11 & \textbf{23.97} & 0.444 & 0.339 & 0.496 & 0.473 & \textbf{0.775} & 0.626 & 0.617 & 0.459 & 0.607 & \textbf{0.077} & 16.92 & 34.46 & 18.49 & 22.66 & \textbf{4.40} \\
\quad S27 & 14.35 & 12.21 & 16.35 & 13.76 & \textbf{22.04} & 0.447 & 0.326 & 0.514 & 0.435 & \textbf{0.743} & 0.515 & 0.697 & 0.312 & 0.515 & \textbf{0.094} & 18.83 & 48.92 & 16.06 & 26.43 & \textbf{5.80} \\
\quad S28 & 13.65 & 12.31 & 15.57 & 13.94 & \textbf{20.99} & 0.375 & 0.339 & 0.503 & 0.455 & \textbf{0.716} & 0.627 & 0.686 & 0.437 & 0.553 & \textbf{0.136} & 26.91 & 50.64 & 23.77 & 27.24 & \textbf{11.21} \\
\quad S29 & 13.56 & 12.02 & 15.08 & 14.53 & \textbf{20.31} & 0.413 & 0.333 & 0.478 & 0.449 & \textbf{0.699} & 0.576 & 0.707 & 0.400 & 0.430 & \textbf{0.138} & 30.32 & 52.15 & 23.79 & 32.41 & \textbf{9.70} \\
\quad S30 & 11.65 & 11.94 & 12.86 & 12.93 & \textbf{16.59} & 0.321 & 0.315 & 0.375 & 0.394 & \textbf{0.562} & 0.928 & 0.758 & 0.720 & 0.679 & \textbf{0.403} & 49.25 & 58.62 & 45.54 & 46.01 & \textbf{28.68} \\
\multicolumn{21}{l}{\emph{Lounge}} \\
\quad S31 & 15.11 & 12.87 & 17.53 & 14.45 & \textbf{27.36} & 0.501 & 0.332 & 0.562 & 0.452 & \textbf{0.858} & 0.502 & 0.701 & 0.313 & 0.550 & \textbf{0.036} & 14.18 & 42.54 & 12.62 & 17.87 & \textbf{2.48} \\
\quad S32 & 14.27 & 12.10 & 16.21 & 12.85 & \textbf{25.84} & 0.447 & 0.295 & 0.497 & 0.398 & \textbf{0.827} & 0.615 & 0.789 & 0.391 & 0.758 & \textbf{0.052} & 18.40 & 54.48 & 17.10 & 26.83 & \textbf{3.41} \\
\quad S33 & 13.92 & 12.76 & 15.87 & 12.81 & \textbf{23.84} & 0.449 & 0.358 & 0.509 & 0.420 & \textbf{0.776} & 0.644 & 0.666 & 0.419 & 0.844 & \textbf{0.086} & 17.37 & 47.18 & 16.07 & 26.13 & \textbf{5.11} \\
\quad S34 & 12.89 & 12.98 & 15.12 & 13.09 & \textbf{22.31} & 0.387 & 0.347 & 0.485 & 0.392 & \textbf{0.733} & 0.736 & 0.626 & 0.422 & 0.682 & \textbf{0.129} & 34.83 & 46.72 & 31.17 & 40.58 & \textbf{11.50} \\
\quad S35 & 13.12 & 13.23 & 14.22 & 12.98 & \textbf{20.48} & 0.381 & 0.372 & 0.437 & 0.403 & \textbf{0.703} & 0.618 & 0.534 & 0.472 & 0.634 & \textbf{0.150} & 40.61 & 47.04 & 36.32 & 36.94 & \textbf{12.18} \\
\bottomrule
\end{tabular}}
\end{table}

\begin{table}[t]
\centering\small
\caption{\textit{Scene-level results for unseen scene variants with transfer baselines.}
The three transfer baselines fine-tune source-scene models using the same \(M{=}32\) references; GRaF and~\ourSystemPR are repeated unchanged from the few-shot comparison.}
\label{tab:e2perscene_tr}
\setlength{\tabcolsep}{2.5pt}
\resizebox{\linewidth}{!}{%
\begin{tabular}{lcccccccccccccccccccc}
\toprule
& \multicolumn{5}{c}{PSNR\(\uparrow\)} & \multicolumn{5}{c}{SSIM\(\uparrow\)} & \multicolumn{5}{c}{NMSE\(\downarrow\)} & \multicolumn{5}{c}{AoA err (\(^\circ\))\(\downarrow\)} \\
\cmidrule(lr){2-6}\cmidrule(lr){7-11}\cmidrule(lr){12-16}\cmidrule(lr){17-21}
Scene & NeRF\(^2\)-TR & GSRF-TR & WRF-GS-TR & GRaF & \ourSystemPR & NeRF\(^2\)-TR & GSRF-TR & WRF-GS-TR & GRaF & \ourSystemPR & NeRF\(^2\)-TR & GSRF-TR & WRF-GS-TR & GRaF & \ourSystemPR & NeRF\(^2\)-TR & GSRF-TR & WRF-GS-TR & GRaF & \ourSystemPR \\
\midrule
\multicolumn{21}{l}{\emph{Conference}} \\
\quad S1 & 17.33 & 19.22 & 20.94 & 15.53 & \textbf{28.44} & 0.504 & 0.626 & 0.653 & 0.487 & \textbf{0.876} & 0.319 & 0.202 & 0.139 & 0.418 & \textbf{0.023} & 7.11 & 13.57 & 11.58 & 20.27 & \textbf{2.10} \\
\quad S2 & 14.50 & 11.14 & 12.47 & 13.83 & \textbf{25.61} & 0.442 & 0.362 & 0.397 & 0.439 & \textbf{0.811} & 0.483 & 1.004 & 0.710 & 0.558 & \textbf{0.069} & 19.55 & 51.60 & 23.46 & 26.26 & \textbf{4.40} \\
\quad S3 & 12.84 & 12.83 & 13.22 & 13.73 & \textbf{21.12} & 0.375 & 0.428 & 0.429 & 0.441 & \textbf{0.710} & 0.693 & 0.725 & 0.590 & 0.518 & \textbf{0.128} & 24.05 & 30.46 & 22.72 & 25.87 & \textbf{8.17} \\
\quad S4 & 12.19 & 12.45 & 12.69 & 14.10 & \textbf{20.97} & 0.368 & 0.405 & 0.449 & 0.444 & \textbf{0.712} & 0.766 & 0.683 & 0.647 & 0.460 & \textbf{0.117} & 30.25 & 42.80 & 28.89 & 25.50 & \textbf{7.54} \\
\quad S5 & 13.07 & 13.85 & 13.71 & 14.45 & \textbf{19.10} & 0.397 & 0.435 & 0.480 & 0.468 & \textbf{0.643} & 0.661 & 0.591 & 0.523 & 0.452 & \textbf{0.227} & 30.71 & 33.75 & 30.66 & 27.87 & \textbf{14.92} \\
\multicolumn{21}{l}{\emph{Classroom}} \\
\quad S6 & 12.28 & 11.21 & 12.87 & 15.25 & \textbf{25.62} & 0.370 & 0.360 & 0.444 & 0.500 & \textbf{0.824} & 0.864 & 0.950 & 0.735 & 0.451 & \textbf{0.042} & 20.48 & 52.78 & 26.97 & 15.24 & \textbf{2.78} \\
\quad S7 & 14.90 & 12.25 & 13.84 & 13.61 & \textbf{30.40} & 0.412 & 0.381 & 0.460 & 0.443 & \textbf{0.897} & 0.533 & 0.942 & 0.611 & 0.668 & \textbf{0.021} & 19.59 & 48.85 & 42.60 & 24.85 & \textbf{1.92} \\
\quad S8 & 13.51 & 13.11 & 13.92 & 13.53 & \textbf{24.42} & 0.391 & 0.441 & 0.474 & 0.465 & \textbf{0.787} & 0.635 & 0.717 & 0.577 & 0.659 & \textbf{0.082} & 35.24 & 33.33 & 27.19 & 23.45 & \textbf{5.33} \\
\quad S9 & 12.53 & 11.99 & 12.14 & 14.79 & \textbf{25.27} & 0.351 & 0.353 & 0.402 & 0.449 & \textbf{0.803} & 1.004 & 0.994 & 0.844 & 0.579 & \textbf{0.072} & 31.34 & 47.57 & 33.40 & 26.54 & \textbf{4.96} \\
\quad S10 & 13.23 & 12.41 & 13.56 & 13.47 & \textbf{22.89} & 0.427 & 0.400 & 0.453 & 0.433 & \textbf{0.746} & 0.712 & 0.885 & 0.528 & 0.639 & \textbf{0.104} & 27.13 & 47.20 & 43.14 & 25.16 & \textbf{7.67} \\
\multicolumn{21}{l}{\emph{Bedroom}} \\
\quad S11 & 11.52 & 13.03 & 13.61 & 13.92 & \textbf{29.28} & 0.294 & 0.425 & 0.427 & 0.415 & \textbf{0.888} & 1.188 & 0.876 & 0.657 & 0.641 & \textbf{0.020} & 37.81 & 33.11 & 39.07 & 23.42 & \textbf{1.65} \\
\quad S12 & 13.60 & 11.18 & 11.95 & 13.10 & \textbf{26.74} & 0.398 & 0.353 & 0.397 & 0.408 & \textbf{0.848} & 0.747 & 1.120 & 0.887 & 0.706 & \textbf{0.034} & 20.39 & 50.38 & 27.36 & 27.25 & \textbf{2.20} \\
\quad S13 & 13.16 & 14.14 & 13.48 & 13.97 & \textbf{23.94} & 0.354 & 0.455 & 0.448 & 0.430 & \textbf{0.785} & 0.698 & 0.665 & 0.616 & 0.590 & \textbf{0.070} & 24.47 & 25.72 & 24.86 & 25.72 & \textbf{3.80} \\
\quad S14 & 12.89 & 13.02 & 11.69 & 14.24 & \textbf{22.29} & 0.387 & 0.425 & 0.384 & 0.442 & \textbf{0.748} & 0.675 & 0.628 & 0.857 & 0.507 & \textbf{0.106} & 28.13 & 37.98 & 34.26 & 28.48 & \textbf{8.42} \\
\quad S15 & 13.60 & 13.43 & 14.17 & 13.78 & \textbf{25.45} & 0.423 & 0.437 & 0.447 & 0.421 & \textbf{0.806} & 0.651 & 0.672 & 0.526 & 0.606 & \textbf{0.067} & 28.23 & 38.97 & 34.19 & 25.01 & \textbf{4.56} \\
\multicolumn{21}{l}{\emph{Corridor}} \\
\quad S16 & 10.67 & 12.15 & 13.16 & 13.19 & \textbf{15.36} & 0.316 & 0.399 & 0.460 & 0.460 & \textbf{0.573} & 0.848 & 0.553 & 0.472 & 0.477 & \textbf{0.287} & 44.73 & 39.26 & 31.57 & 27.52 & \textbf{22.37} \\
\quad S17 & 13.15 & 11.45 & 13.70 & 14.10 & \textbf{21.90} & 0.376 & 0.372 & 0.465 & 0.441 & \textbf{0.712} & 0.757 & 1.070 & 0.573 & 0.594 & \textbf{0.116} & 46.22 & 59.26 & 56.93 & 46.88 & \textbf{13.23} \\
\quad S18 & 12.59 & 13.15 & 14.37 & 12.79 & \textbf{14.81} & 0.331 & 0.442 & 0.478 & 0.421 & \textbf{0.532} & 0.530 & 0.485 & 0.364 & 0.590 & \textbf{0.348} & 24.68 & 27.10 & 20.47 & 24.85 & \textbf{19.94} \\
\quad S19 & 11.61 & 12.22 & 8.64 & 13.53 & \textbf{19.69} & 0.295 & 0.380 & 0.315 & 0.383 & \textbf{0.643} & 0.880 & 0.776 & 1.781 & 0.558 & \textbf{0.176} & 60.44 & 67.37 & 80.24 & 52.41 & \textbf{21.03} \\
\quad S20 & 14.09 & 14.77 & 14.65 & 13.56 & \textbf{18.21} & 0.433 & 0.483 & 0.467 & 0.384 & \textbf{0.610} & 0.521 & 0.466 & 0.417 & 0.559 & \textbf{0.270} & 40.30 & 38.73 & 33.15 & 45.70 & \textbf{24.02} \\
\multicolumn{21}{l}{\emph{Open office}} \\
\quad S21 & 11.33 & 12.27 & 14.40 & 14.71 & \textbf{25.19} & 0.338 & 0.400 & 0.492 & 0.502 & \textbf{0.813} & 1.082 & 0.793 & 0.523 & 0.478 & \textbf{0.055} & 33.82 & 40.15 & 23.01 & 17.44 & \textbf{3.24} \\
\quad S22 & 13.88 & 12.12 & 16.19 & 15.03 & \textbf{27.33} & 0.446 & 0.402 & 0.512 & 0.478 & \textbf{0.850} & 0.711 & 0.923 & 0.411 & 0.536 & \textbf{0.036} & 20.52 & 41.86 & 21.30 & 25.00 & \textbf{2.55} \\
\quad S23 & 13.69 & 11.87 & 13.83 & 14.02 & \textbf{25.77} & 0.417 & 0.402 & 0.472 & 0.484 & \textbf{0.820} & 0.631 & 0.886 & 0.548 & 0.598 & \textbf{0.047} & 17.44 & 41.03 & 20.72 & 21.62 & \textbf{3.35} \\
\quad S24 & 11.54 & 12.22 & 12.39 & 12.75 & \textbf{20.12} & 0.336 & 0.399 & 0.438 & 0.436 & \textbf{0.686} & 0.914 & 0.728 & 0.675 & 0.677 & \textbf{0.182} & 34.45 & 40.31 & 28.58 & 27.91 & \textbf{12.10} \\
\quad S25 & 12.45 & 13.19 & 12.58 & 14.26 & \textbf{20.93} & 0.367 & 0.420 & 0.412 & 0.444 & \textbf{0.686} & 0.716 & 0.645 & 0.631 & 0.485 & \textbf{0.191} & 54.60 & 50.05 & 43.77 & 32.67 & \textbf{15.18} \\
\multicolumn{21}{l}{\emph{Lab}} \\
\quad S26 & 12.58 & 12.46 & 13.90 & 14.11 & \textbf{23.97} & 0.381 & 0.407 & 0.446 & 0.473 & \textbf{0.775} & 0.826 & 0.765 & 0.565 & 0.607 & \textbf{0.077} & 22.52 & 36.17 & 24.81 & 22.66 & \textbf{4.40} \\
\quad S27 & 13.70 & 11.74 & 13.73 & 13.76 & \textbf{22.04} & 0.442 & 0.389 & 0.470 & 0.435 & \textbf{0.743} & 0.579 & 0.854 & 0.476 & 0.515 & \textbf{0.094} & 21.06 & 60.21 & 36.45 & 26.43 & \textbf{5.80} \\
\quad S28 & 13.67 & 12.97 & 13.28 & 13.94 & \textbf{20.99} & 0.391 & 0.427 & 0.450 & 0.455 & \textbf{0.716} & 0.612 & 0.689 & 0.586 & 0.553 & \textbf{0.136} & 26.96 & 32.39 & 29.86 & 27.24 & \textbf{11.21} \\
\quad S29 & 12.60 & 12.76 & 12.83 & 14.53 & \textbf{20.31} & 0.358 & 0.417 & 0.441 & 0.449 & \textbf{0.699} & 0.702 & 0.675 & 0.613 & 0.430 & \textbf{0.138} & 39.80 & 48.43 & 33.79 & 32.41 & \textbf{9.70} \\
\quad S30 & 11.79 & 12.31 & 12.54 & 12.93 & \textbf{16.59} & 0.337 & 0.394 & 0.401 & 0.394 & \textbf{0.562} & 0.871 & 0.753 & 0.673 & 0.679 & \textbf{0.403} & 49.67 & 60.86 & 61.93 & 46.01 & \textbf{28.68} \\
\multicolumn{21}{l}{\emph{Lounge}} \\
\quad S31 & 12.25 & 11.81 & 13.30 & 14.45 & \textbf{27.36} & 0.337 & 0.384 & 0.419 & 0.452 & \textbf{0.858} & 0.893 & 0.914 & 0.686 & 0.550 & \textbf{0.036} & 25.26 & 49.74 & 26.70 & 17.87 & \textbf{2.48} \\
\quad S32 & 13.54 & 11.12 & 12.59 & 12.85 & \textbf{25.84} & 0.388 & 0.358 & 0.411 & 0.398 & \textbf{0.827} & 0.675 & 1.034 & 0.750 & 0.758 & \textbf{0.052} & 30.20 & 54.78 & 31.82 & 26.83 & \textbf{3.41} \\
\quad S33 & 13.05 & 13.38 & 13.49 & 12.81 & \textbf{23.84} & 0.400 & 0.433 & 0.472 & 0.420 & \textbf{0.776} & 0.759 & 0.731 & 0.660 & 0.844 & \textbf{0.086} & 21.84 & 31.89 & 22.80 & 26.13 & \textbf{5.11} \\
\quad S34 & 12.28 & 12.42 & 12.28 & 13.09 & \textbf{22.31} & 0.362 & 0.397 & 0.413 & 0.392 & \textbf{0.733} & 0.869 & 0.743 & 0.775 & 0.682 & \textbf{0.129} & 49.40 & 63.11 & 82.30 & 40.58 & \textbf{11.50} \\
\quad S35 & 12.78 & 12.51 & 13.54 & 12.98 & \textbf{20.48} & 0.394 & 0.417 & 0.435 & 0.403 & \textbf{0.703} & 0.692 & 0.692 & 0.531 & 0.634 & \textbf{0.150} & 50.71 & 63.57 & 39.64 & 36.94 & \textbf{12.18} \\
\bottomrule
\end{tabular}}
\end{table}

\subsubsection{Cross-Scene Variant Results}
\label{app:breakdown_unseen}

\autoref{tab:e2perscene} and~\autoref{tab:e2perscene_tr} break down the~\ourSystemPR results in~\autoref{tab:seen} by scene, using the same~\(M=32\) target-scene references.
They compare few-shot baselines and baselines transferred from another scene in the same category, respectively.
These tables report scene-level summaries, whereas~\autoref{tab:seen} aggregates over folds; unweighted averages of their rows need not reproduce the main-table means.
Comparison with the per-scene fits in~\S\ref{app:breakdown_perscene} shows the combined effect of sparse references and unseen-scene inference.

\textbf{Per-Scene Baselines Depend Strongly on Target-Scene Fitting.}
Across the~\(35\) simulated scenes, moving from full per-scene fitting to~\(32\) target-scene references reduces baseline~PSNR by~\(6.2\)--\(8.0\,\mathrm{dB}\).
NeRF$^2$ falls from~\(21.59\) to~\(13.75\,\mathrm{dB}\), GSRF from~\(20.41\) to~\(12.40\,\mathrm{dB}\), and WRF-GS from~\(22.01\) to~\(15.85\,\mathrm{dB}\).
Feed-forward~GRaF also falls from~\(22.00\) to~\(13.85\,\mathrm{dB}\) under sparse references.
By comparison,~\ourSystemPR decreases by~\(1.34\,\mathrm{dB}\), from~\(24.33\) to~\(22.99\,\mathrm{dB}\).

\textbf{The Advantage Persists Across Scenes.}
\ourSystemPR leads on all~\(35\) scenes against both few-shot and transfer baselines in these breakdowns.
Its average~PSNR gain over the strongest few-shot competitor on each scene is~\(7.13\,\mathrm{dB}\), with a minimum of~\(0.94\,\mathrm{dB}\) on corridor scene~S18.
Against the strongest transfer competitor, the average gain is~\(8.72\,\mathrm{dB}\), with a minimum of~\(0.44\,\mathrm{dB}\), again on~S18.
The feed-forward and power-refinement variants are compared separately in~\autoref{tab:revision:ffpa}.

\textbf{Source-Scene Transfer Provides Limited Benefit.}
Initializing a baseline from another scene in the same category and adapting it with the~\(32\) references does not consistently help.
GSRF rises from~\(12.40\) to~\(12.69\,\mathrm{dB}\) and improves on~\(20\) scenes.
NeRF$^2$ loses~\(0.77\,\mathrm{dB}\) on average and improves on only~\(9\) scenes; WRF-GS loses~\(2.43\,\mathrm{dB}\) and improves on only~\(4\).
A field fitted to one scene is therefore not generally a reliable initialization for another, even within the same category.

\textbf{Harder Scenes Incur Larger Generalization Losses.}
For~\ourSystemPR, the SSIM decrease from per-scene fitting to unseen-variant inference is largest for corridors, from~\(0.687\) to~\(0.614\), and smallest for lounges, from~\(0.805\) to~\(0.779\).
At the scene level, the gap is negatively correlated with per-scene~SSIM, with~\(r=-0.64\).
The six largest gaps occur on~S18, S30, S5, S24, S16, and~S20, which are also among the more difficult scenes under per-scene fitting.
Thus, scenes that are difficult with full target-scene data tend to lose more when synthesis relies on sparse references in an unseen scene.

\begin{table}[t]
\centering\small
\caption{\textit{Scene-level results with the entire scene category unseen using few-shot baselines.}
Each scene is evaluated by the cross-category fold that excludes its category using \(M{=}32\) references. Few-shot baseline results are unchanged from the unseen-variant setting.
}
\label{tab:e3perscene}
\setlength{\tabcolsep}{2.5pt}
\resizebox{\linewidth}{!}{%
\begin{tabular}{lcccccccccccccccccccc}
\toprule
& \multicolumn{5}{c}{PSNR\(\uparrow\)} & \multicolumn{5}{c}{SSIM\(\uparrow\)} & \multicolumn{5}{c}{NMSE\(\downarrow\)} & \multicolumn{5}{c}{AoA err (\(^\circ\))\(\downarrow\)} \\
\cmidrule(lr){2-6}\cmidrule(lr){7-11}\cmidrule(lr){12-16}\cmidrule(lr){17-21}
Scene & NeRF\(^2\)-FS & GSRF-FS & WRF-GS-FS & GRaF & \ourSystemPR & NeRF\(^2\)-FS & GSRF-FS & WRF-GS-FS & GRaF & \ourSystemPR & NeRF\(^2\)-FS & GSRF-FS & WRF-GS-FS & GRaF & \ourSystemPR & NeRF\(^2\)-FS & GSRF-FS & WRF-GS-FS & GRaF & \ourSystemPR \\
\midrule
\multicolumn{21}{l}{\emph{Conference}} \\
\quad S1 & 15.18 & 12.09 & 17.13 & 14.55 & \textbf{28.24} & 0.431 & 0.302 & 0.538 & 0.461 & \textbf{0.871} & 0.499 & 0.805 & 0.354 & 0.524 & \textbf{0.024} & 14.40 & 51.26 & 16.62 & 21.53 & \textbf{2.10} \\
\quad S2 & 15.26 & 11.66 & 17.12 & 14.48 & \textbf{25.72} & 0.480 & 0.301 & 0.549 & 0.451 & \textbf{0.823} & 0.427 & 0.850 & 0.287 & 0.478 & \textbf{0.060} & 16.48 & 50.03 & 14.66 & 26.08 & \textbf{4.14} \\
\quad S3 & 12.80 & 12.16 & 15.68 & 14.32 & \textbf{20.77} & 0.371 & 0.336 & 0.485 & 0.451 & \textbf{0.698} & 0.657 & 0.687 & 0.364 & 0.459 & \textbf{0.141} & 25.35 & 54.77 & 21.14 & 24.54 & \textbf{9.05} \\
\quad S4 & 14.11 & 11.99 & 15.17 & 14.09 & \textbf{20.34} & 0.444 & 0.336 & 0.484 & 0.448 & \textbf{0.698} & 0.491 & 0.694 & 0.380 & 0.465 & \textbf{0.151} & 19.19 & 52.14 & 17.87 & 24.14 & \textbf{8.13} \\
\quad S5 & 13.87 & 12.15 & 14.92 & 14.60 & \textbf{19.22} & 0.412 & 0.329 & 0.473 & 0.471 & \textbf{0.647} & 0.538 & 0.740 & 0.424 & 0.437 & \textbf{0.217} & 23.96 & 55.71 & 23.81 & 27.42 & \textbf{14.98} \\
\multicolumn{21}{l}{\emph{Classroom}} \\
\quad S6 & 15.71 & 11.60 & 17.67 & 14.04 & \textbf{25.68} & 0.456 & 0.284 & 0.562 & 0.471 & \textbf{0.820} & 0.420 & 0.798 & 0.275 & 0.637 & \textbf{0.044} & 12.99 & 48.06 & 12.33 & 20.36 & \textbf{2.75} \\
\quad S7 & 14.25 & 11.71 & 17.83 & 13.47 & \textbf{30.43} & 0.412 & 0.281 & 0.566 & 0.441 & \textbf{0.897} & 0.644 & 0.944 & 0.313 & 0.725 & \textbf{0.020} & 23.38 & 67.48 & 14.24 & 24.44 & \textbf{1.92} \\
\quad S8 & 14.80 & 11.45 & 16.47 & 13.71 & \textbf{23.79} & 0.481 & 0.304 & 0.532 & 0.472 & \textbf{0.767} & 0.506 & 0.836 & 0.380 & 0.652 & \textbf{0.098} & 16.42 & 63.04 & 19.94 & 22.46 & \textbf{5.84} \\
\quad S9 & 13.59 & 12.25 & 17.13 & 14.31 & \textbf{26.52} & 0.426 & 0.325 & 0.513 & 0.450 & \textbf{0.819} & 0.786 & 0.853 & 0.360 & 0.673 & \textbf{0.068} & 19.83 & 68.48 & 17.20 & 27.22 & \textbf{4.96} \\
\quad S10 & 10.99 & 11.73 & 15.49 & 12.86 & \textbf{22.95} & 0.300 & 0.323 & 0.463 & 0.424 & \textbf{0.742} & 1.348 & 0.845 & 0.445 & 0.757 & \textbf{0.101} & 34.94 & 57.91 & 23.82 & 24.91 & \textbf{7.58} \\
\multicolumn{21}{l}{\emph{Bedroom}} \\
\quad S11 & 14.89 & 14.72 & 16.96 & 12.80 & \textbf{29.49} & 0.446 & 0.392 & 0.514 & 0.371 & \textbf{0.892} & 0.577 & 0.533 & 0.407 & 0.796 & \textbf{0.019} & 16.31 & 34.36 & 20.66 & 37.57 & \textbf{1.65} \\
\quad S12 & 15.03 & 12.99 & 15.81 & 13.59 & \textbf{26.74} & 0.458 & 0.339 & 0.472 & 0.410 & \textbf{0.831} & 0.533 & 0.690 & 0.434 & 0.647 & \textbf{0.035} & 14.94 & 39.90 & 19.30 & 27.51 & \textbf{2.17} \\
\quad S13 & 14.53 & 13.39 & 16.59 & 14.20 & \textbf{25.02} & 0.433 & 0.375 & 0.512 & 0.437 & \textbf{0.811} & 0.566 & 0.622 & 0.384 & 0.542 & \textbf{0.058} & 19.17 & 35.24 & 16.97 & 26.03 & \textbf{3.79} \\
\quad S14 & 14.73 & 13.62 & 16.31 & 13.67 & \textbf{22.67} & 0.462 & 0.381 & 0.510 & 0.426 & \textbf{0.751} & 0.504 & 0.573 & 0.374 & 0.560 & \textbf{0.105} & 17.06 & 35.44 & 18.41 & 29.84 & \textbf{8.39} \\
\quad S15 & 14.23 & 12.72 & 15.51 & 13.24 & \textbf{25.55} & 0.428 & 0.339 & 0.468 & 0.400 & \textbf{0.813} & 0.567 & 0.713 & 0.423 & 0.677 & \textbf{0.065} & 16.51 & 48.30 & 23.63 & 27.73 & \textbf{4.59} \\
\multicolumn{21}{l}{\emph{Corridor}} \\
\quad S16 & 12.04 & 11.68 & 12.92 & 13.16 & \textbf{15.26} & 0.402 & 0.349 & 0.420 & 0.457 & \textbf{0.569} & 0.621 & 0.579 & 0.494 & 0.488 & \textbf{0.294} & 31.91 & 48.11 & 31.30 & 28.04 & \textbf{22.30} \\
\quad S17 & 14.16 & 14.16 & 16.33 & 14.97 & \textbf{22.14} & 0.411 & 0.395 & 0.487 & 0.457 & \textbf{0.695} & 0.633 & 0.561 & 0.399 & 0.485 & \textbf{0.121} & 40.63 & 55.37 & 33.01 & 42.46 & \textbf{13.30} \\
\quad S18 & 11.99 & 11.25 & 13.86 & 13.72 & \textbf{15.76} & 0.349 & 0.318 & 0.445 & 0.450 & \textbf{0.566} & 0.648 & 0.681 & 0.422 & 0.457 & \textbf{0.286} & 24.90 & 46.52 & 21.72 & 23.63 & \textbf{17.03} \\
\quad S19 & 11.76 & 13.59 & 14.54 & 13.55 & \textbf{19.12} & 0.290 & 0.350 & 0.421 & 0.401 & \textbf{0.624} & 0.871 & 0.521 & 0.454 & 0.552 & \textbf{0.195} & 57.47 & 48.95 & 43.13 & 50.13 & \textbf{21.74} \\
\quad S20 & 12.93 & 12.94 & 14.46 & 13.80 & \textbf{18.06} & 0.342 & 0.360 & 0.417 & 0.395 & \textbf{0.603} & 0.651 & 0.619 & 0.457 & 0.530 & \textbf{0.278} & 38.39 & 61.16 & 34.75 & 45.93 & \textbf{26.50} \\
\multicolumn{21}{l}{\emph{Open office}} \\
\quad S21 & 14.00 & 11.76 & 17.30 & 13.93 & \textbf{23.42} & 0.405 & 0.348 & 0.550 & 0.478 & \textbf{0.768} & 0.562 & 0.767 & 0.271 & 0.567 & \textbf{0.086} & 16.61 & 53.11 & 15.08 & 20.60 & \textbf{3.74} \\
\quad S22 & 13.59 & 13.04 & 17.56 & 15.04 & \textbf{26.67} & 0.452 & 0.398 & 0.559 & 0.474 & \textbf{0.839} & 0.767 & 0.682 & 0.343 & 0.526 & \textbf{0.039} & 15.52 & 42.97 & 16.49 & 24.14 & \textbf{2.66} \\
\quad S23 & 15.05 & 10.79 & 17.91 & 14.27 & \textbf{24.01} & 0.474 & 0.279 & 0.576 & 0.492 & \textbf{0.788} & 0.486 & 0.951 & 0.269 & 0.544 & \textbf{0.057} & 13.85 & 63.81 & 11.32 & 21.45 & \textbf{3.41} \\
\quad S24 & 11.90 & 11.03 & 14.31 & 12.58 & \textbf{19.76} & 0.374 & 0.296 & 0.465 & 0.439 & \textbf{0.676} & 0.817 & 0.830 & 0.466 & 0.712 & \textbf{0.193} & 27.30 & 61.17 & 23.52 & 27.74 & \textbf{11.42} \\
\quad S25 & 13.15 & 12.03 & 15.37 & 14.08 & \textbf{19.96} & 0.380 & 0.349 & 0.464 & 0.442 & \textbf{0.665} & 0.633 & 0.783 & 0.438 & 0.499 & \textbf{0.245} & 43.08 & 62.38 & 29.35 & 37.30 & \textbf{16.43} \\
\multicolumn{21}{l}{\emph{Lab}} \\
\quad S26 & 14.17 & 13.07 & 15.49 & 12.71 & \textbf{23.69} & 0.444 & 0.339 & 0.496 & 0.441 & \textbf{0.776} & 0.626 & 0.617 & 0.459 & 0.823 & \textbf{0.080} & 16.92 & 34.46 & 18.49 & 26.47 & \textbf{4.50} \\
\quad S27 & 14.35 & 12.21 & 16.35 & 13.70 & \textbf{22.18} & 0.447 & 0.326 & 0.514 & 0.432 & \textbf{0.743} & 0.515 & 0.697 & 0.312 & 0.537 & \textbf{0.096} & 18.83 & 48.92 & 16.06 & 26.25 & \textbf{5.83} \\
\quad S28 & 13.65 & 12.31 & 15.57 & 13.93 & \textbf{20.58} & 0.375 & 0.339 & 0.503 & 0.460 & \textbf{0.682} & 0.627 & 0.686 & 0.437 & 0.555 & \textbf{0.155} & 26.91 & 50.64 & 23.77 & 27.83 & \textbf{11.18} \\
\quad S29 & 13.56 & 12.02 & 15.08 & 13.71 & \textbf{20.47} & 0.413 & 0.333 & 0.478 & 0.440 & \textbf{0.688} & 0.576 & 0.707 & 0.400 & 0.542 & \textbf{0.137} & 30.32 & 52.15 & 23.79 & 30.23 & \textbf{9.60} \\
\quad S30 & 11.65 & 11.94 & 12.86 & 12.02 & \textbf{16.63} & 0.321 & 0.315 & 0.375 & 0.382 & \textbf{0.547} & 0.928 & 0.758 & 0.720 & 0.833 & \textbf{0.404} & 49.25 & 58.62 & 45.54 & 45.77 & \textbf{28.76} \\
\multicolumn{21}{l}{\emph{Lounge}} \\
\quad S31 & 15.11 & 12.87 & 17.53 & 13.53 & \textbf{27.37} & 0.501 & 0.332 & 0.562 & 0.414 & \textbf{0.858} & 0.502 & 0.701 & 0.313 & 0.644 & \textbf{0.036} & 14.18 & 42.54 & 12.62 & 22.92 & \textbf{2.45} \\
\quad S32 & 14.27 & 12.10 & 16.21 & 12.99 & \textbf{25.84} & 0.447 & 0.295 & 0.497 & 0.386 & \textbf{0.823} & 0.615 & 0.789 & 0.391 & 0.748 & \textbf{0.053} & 18.40 & 54.48 & 17.10 & 26.49 & \textbf{3.39} \\
\quad S33 & 13.92 & 12.76 & 15.87 & 13.30 & \textbf{23.55} & 0.449 & 0.358 & 0.509 & 0.426 & \textbf{0.766} & 0.644 & 0.666 & 0.419 & 0.744 & \textbf{0.091} & 17.37 & 47.18 & 16.07 & 25.12 & \textbf{5.16} \\
\quad S34 & 12.89 & 12.98 & 15.12 & 12.57 & \textbf{21.71} & 0.387 & 0.347 & 0.485 & 0.383 & \textbf{0.718} & 0.736 & 0.626 & 0.422 & 0.762 & \textbf{0.126} & 34.83 & 46.72 & 31.17 & 41.10 & \textbf{11.47} \\
\quad S35 & 13.12 & 13.23 & 14.22 & 12.82 & \textbf{20.24} & 0.381 & 0.372 & 0.437 & 0.392 & \textbf{0.697} & 0.618 & 0.534 & 0.472 & 0.661 & \textbf{0.148} & 40.61 & 47.04 & 36.32 & 38.86 & \textbf{12.24} \\
\bottomrule
\end{tabular}}
\end{table}

\begin{table}[t]
\centering\small
\caption{\textit{Scene-level results with the entire scene category unseen using transfer baselines.}
The three transfer baselines fine-tune models from outside the target category using the same \(M{=}32\) references; GRaF and~\ourSystemPR are unchanged from the few-shot comparison.}
\label{tab:e3perscene_tr}
\setlength{\tabcolsep}{2.5pt}
\resizebox{\linewidth}{!}{%
\begin{tabular}{lcccccccccccccccccccc}
\toprule
& \multicolumn{5}{c}{PSNR\(\uparrow\)} & \multicolumn{5}{c}{SSIM\(\uparrow\)} & \multicolumn{5}{c}{NMSE\(\downarrow\)} & \multicolumn{5}{c}{AoA err (\(^\circ\))\(\downarrow\)} \\
\cmidrule(lr){2-6}\cmidrule(lr){7-11}\cmidrule(lr){12-16}\cmidrule(lr){17-21}
Scene & NeRF\(^2\)-TR & GSRF-TR & WRF-GS-TR & GRaF & \ourSystemPR & NeRF\(^2\)-TR & GSRF-TR & WRF-GS-TR & GRaF & \ourSystemPR & NeRF\(^2\)-TR & GSRF-TR & WRF-GS-TR & GRaF & \ourSystemPR & NeRF\(^2\)-TR & GSRF-TR & WRF-GS-TR & GRaF & \ourSystemPR \\
\midrule
\multicolumn{21}{l}{\emph{Conference}} \\
\quad S1 & 12.77 & 13.22 & 12.21 & 14.55 & \textbf{28.24} & 0.348 & 0.408 & 0.404 & 0.461 & \textbf{0.871} & 0.782 & 0.729 & 0.924 & 0.524 & \textbf{0.024} & 30.64 & 38.85 & 27.99 & 21.53 & \textbf{2.10} \\
\quad S2 & 12.77 & 13.39 & 10.72 & 14.48 & \textbf{25.72} & 0.338 & 0.415 & 0.358 & 0.451 & \textbf{0.823} & 0.681 & 0.657 & 1.102 & 0.478 & \textbf{0.060} & 31.41 & 35.87 & 39.86 & 26.08 & \textbf{4.14} \\
\quad S3 & 12.52 & 13.11 & 10.95 & 14.32 & \textbf{20.77} & 0.371 & 0.421 & 0.406 & 0.451 & \textbf{0.698} & 0.791 & 0.646 & 1.047 & 0.459 & \textbf{0.141} & 33.03 & 39.49 & 36.36 & 24.54 & \textbf{9.05} \\
\quad S4 & 12.22 & 12.45 & 12.40 & 14.09 & \textbf{20.34} & 0.369 & 0.404 & 0.436 & 0.448 & \textbf{0.698} & 0.761 & 0.682 & 0.705 & 0.465 & \textbf{0.151} & 29.71 & 42.88 & 29.38 & 24.14 & \textbf{8.13} \\
\quad S5 & 12.51 & 13.00 & 10.35 & 14.60 & \textbf{19.22} & 0.352 & 0.408 & 0.375 & 0.471 & \textbf{0.647} & 0.792 & 0.706 & 1.150 & 0.437 & \textbf{0.217} & 39.05 & 37.67 & 32.44 & 27.42 & \textbf{14.98} \\
\multicolumn{21}{l}{\emph{Classroom}} \\
\quad S6 & 14.29 & 12.06 & 9.15 & 14.04 & \textbf{25.68} & 0.453 & 0.416 & 0.317 & 0.471 & \textbf{0.820} & 0.593 & 0.916 & 1.547 & 0.637 & \textbf{0.044} & 17.36 & 36.61 & 32.80 & 20.36 & \textbf{2.75} \\
\quad S7 & 14.30 & 12.10 & 12.53 & 13.47 & \textbf{30.43} & 0.422 & 0.389 & 0.425 & 0.441 & \textbf{0.897} & 0.638 & 1.040 & 0.798 & 0.725 & \textbf{0.020} & 20.96 & 35.67 & 40.02 & 24.44 & \textbf{1.92} \\
\quad S8 & 15.36 & 13.99 & 14.46 & 13.71 & \textbf{23.79} & 0.488 & 0.472 & 0.484 & 0.472 & \textbf{0.767} & 0.424 & 0.564 & 0.450 & 0.652 & \textbf{0.098} & 17.56 & 20.56 & 26.77 & 22.46 & \textbf{5.84} \\
\quad S9 & 14.10 & 12.07 & 11.96 & 14.31 & \textbf{26.52} & 0.415 & 0.379 & 0.400 & 0.450 & \textbf{0.819} & 0.683 & 1.089 & 0.964 & 0.673 & \textbf{0.068} & 23.18 & 39.75 & 39.46 & 27.22 & \textbf{4.96} \\
\quad S10 & 12.76 & 12.45 & 12.75 & 12.86 & \textbf{22.95} & 0.396 & 0.420 & 0.450 & 0.424 & \textbf{0.742} & 0.806 & 0.890 & 0.706 & 0.757 & \textbf{0.101} & 32.34 & 35.22 & 39.22 & 24.91 & \textbf{7.58} \\
\multicolumn{21}{l}{\emph{Bedroom}} \\
\quad S11 & 13.93 & 13.41 & 14.16 & 12.80 & \textbf{29.49} & 0.385 & 0.416 & 0.457 & 0.371 & \textbf{0.892} & 0.719 & 0.838 & 0.588 & 0.796 & \textbf{0.019} & 26.56 & 32.28 & 36.20 & 37.57 & \textbf{1.65} \\
\quad S12 & 13.36 & 12.15 & 14.15 & 13.59 & \textbf{26.74} & 0.379 & 0.398 & 0.457 & 0.410 & \textbf{0.831} & 0.728 & 0.959 & 0.530 & 0.647 & \textbf{0.035} & 24.17 & 48.11 & 25.34 & 27.51 & \textbf{2.17} \\
\quad S13 & 16.70 & 17.38 & 16.77 & 14.20 & \textbf{25.02} & 0.515 & 0.573 & 0.546 & 0.437 & \textbf{0.811} & 0.354 & 0.315 & 0.341 & 0.542 & \textbf{0.058} & 12.01 & 11.86 & 21.86 & 26.03 & \textbf{3.79} \\
\quad S14 & 13.83 & 12.52 & 12.70 & 13.67 & \textbf{22.67} & 0.393 & 0.411 & 0.415 & 0.426 & \textbf{0.751} & 0.595 & 0.798 & 0.696 & 0.560 & \textbf{0.105} & 23.39 & 49.58 & 32.96 & 29.84 & \textbf{8.39} \\
\quad S15 & 13.05 & 12.05 & 12.71 & 13.24 & \textbf{25.55} & 0.361 & 0.395 & 0.421 & 0.400 & \textbf{0.813} & 0.728 & 0.938 & 0.731 & 0.677 & \textbf{0.065} & 27.50 & 50.49 & 31.31 & 27.73 & \textbf{4.59} \\
\multicolumn{21}{l}{\emph{Corridor}} \\
\quad S16 & 10.48 & 12.44 & 13.63 & 13.16 & \textbf{15.26} & 0.274 & 0.422 & 0.464 & 0.457 & \textbf{0.569} & 0.873 & 0.503 & 0.400 & 0.488 & \textbf{0.294} & 48.77 & 35.55 & 28.54 & 28.04 & \textbf{22.30} \\
\quad S17 & 11.46 & 13.06 & 14.74 & 14.97 & \textbf{22.14} & 0.277 & 0.425 & 0.493 & 0.457 & \textbf{0.695} & 1.106 & 0.757 & 0.482 & 0.485 & \textbf{0.121} & 59.78 & 43.08 & 39.90 & 42.46 & \textbf{13.30} \\
\quad S18 & 10.94 & 12.05 & 13.88 & 13.72 & \textbf{15.76} & 0.270 & 0.399 & 0.459 & 0.450 & \textbf{0.566} & 0.799 & 0.614 & 0.423 & 0.457 & \textbf{0.286} & 45.62 & 36.06 & 23.04 & 23.63 & \textbf{17.03} \\
\quad S19 & 10.13 & 12.18 & 12.57 & 13.55 & \textbf{19.12} & 0.248 & 0.373 & 0.432 & 0.401 & \textbf{0.624} & 1.280 & 0.861 & 0.660 & 0.552 & \textbf{0.195} & 67.11 & 41.91 & 53.63 & 50.13 & \textbf{21.74} \\
\quad S20 & 10.82 & 11.97 & 13.75 & 13.80 & \textbf{18.06} & 0.273 & 0.381 & 0.462 & 0.395 & \textbf{0.603} & 1.084 & 0.808 & 0.503 & 0.530 & \textbf{0.278} & 56.01 & 54.19 & 56.23 & 45.93 & \textbf{26.50} \\
\multicolumn{21}{l}{\emph{Open office}} \\
\quad S21 & 13.80 & 13.33 & 12.71 & 13.93 & \textbf{23.42} & 0.430 & 0.436 & 0.452 & 0.478 & \textbf{0.768} & 0.568 & 0.612 & 0.708 & 0.567 & \textbf{0.086} & 24.64 & 39.25 & 29.89 & 20.60 & \textbf{3.74} \\
\quad S22 & 13.34 & 14.05 & 12.99 & 15.04 & \textbf{26.67} & 0.406 & 0.446 & 0.369 & 0.474 & \textbf{0.839} & 0.777 & 0.679 & 0.704 & 0.526 & \textbf{0.039} & 29.96 & 41.21 & 23.68 & 24.14 & \textbf{2.66} \\
\quad S23 & 13.17 & 13.33 & 13.32 & 14.27 & \textbf{24.01} & 0.419 & 0.430 & 0.471 & 0.492 & \textbf{0.788} & 0.708 & 0.649 & 0.622 & 0.544 & \textbf{0.057} & 21.18 & 32.94 & 24.24 & 21.45 & \textbf{3.41} \\
\quad S24 & 11.73 & 12.54 & 12.18 & 12.58 & \textbf{19.76} & 0.375 & 0.408 & 0.427 & 0.439 & \textbf{0.676} & 0.917 & 0.661 & 0.799 & 0.712 & \textbf{0.193} & 32.35 & 40.73 & 31.25 & 27.74 & \textbf{11.42} \\
\quad S25 & 12.47 & 13.21 & 12.73 & 14.08 & \textbf{19.96} & 0.370 & 0.421 & 0.388 & 0.442 & \textbf{0.665} & 0.709 & 0.642 & 0.635 & 0.499 & \textbf{0.245} & 54.64 & 49.65 & 40.45 & 37.30 & \textbf{16.43} \\
\multicolumn{21}{l}{\emph{Lab}} \\
\quad S26 & 13.22 & 11.99 & 13.81 & 12.71 & \textbf{23.69} & 0.420 & 0.395 & 0.468 & 0.441 & \textbf{0.776} & 0.713 & 0.891 & 0.591 & 0.823 & \textbf{0.080} & 27.90 & 40.05 & 20.67 & 26.47 & \textbf{4.50} \\
\quad S27 & 16.17 & 14.19 & 15.00 & 13.70 & \textbf{22.18} & 0.526 & 0.463 & 0.478 & 0.432 & \textbf{0.743} & 0.414 & 0.536 & 0.430 & 0.537 & \textbf{0.096} & 17.41 & 34.88 & 24.65 & 26.25 & \textbf{5.83} \\
\quad S28 & 13.23 & 12.38 & 13.63 & 13.93 & \textbf{20.58} & 0.406 & 0.399 & 0.454 & 0.460 & \textbf{0.682} & 0.630 & 0.752 & 0.507 & 0.555 & \textbf{0.155} & 34.57 & 46.65 & 40.68 & 27.83 & \textbf{11.18} \\
\quad S29 & 13.53 & 12.10 & 12.12 & 13.71 & \textbf{20.47} & 0.406 & 0.400 & 0.413 & 0.440 & \textbf{0.688} & 0.576 & 0.800 & 0.840 & 0.542 & \textbf{0.137} & 30.77 & 52.29 & 26.14 & 30.23 & \textbf{9.60} \\
\quad S30 & 11.36 & 11.25 & 13.42 & 12.02 & \textbf{16.63} & 0.345 & 0.353 & 0.418 & 0.382 & \textbf{0.547} & 0.941 & 0.974 & 0.571 & 0.833 & \textbf{0.404} & 57.68 & 54.25 & 56.97 & 45.77 & \textbf{28.76} \\
\multicolumn{21}{l}{\emph{Lounge}} \\
\quad S31 & 13.08 & 12.80 & 11.64 & 13.53 & \textbf{27.37} & 0.373 & 0.404 & 0.377 & 0.414 & \textbf{0.858} & 0.741 & 0.884 & 0.889 & 0.644 & \textbf{0.036} & 23.60 & 31.73 & 33.40 & 22.92 & \textbf{2.45} \\
\quad S32 & 12.98 & 12.47 & 12.48 & 12.99 & \textbf{25.84} & 0.375 & 0.399 & 0.399 & 0.386 & \textbf{0.823} & 0.763 & 0.916 & 0.816 & 0.748 & \textbf{0.053} & 27.00 & 39.34 & 34.64 & 26.49 & \textbf{3.39} \\
\quad S33 & 12.45 & 12.49 & 10.72 & 13.30 & \textbf{23.55} & 0.382 & 0.404 & 0.365 & 0.426 & \textbf{0.766} & 0.867 & 0.873 & 1.075 & 0.744 & \textbf{0.091} & 26.03 & 30.32 & 49.79 & 25.12 & \textbf{5.16} \\
\quad S34 & 12.67 & 11.78 & 10.77 & 12.57 & \textbf{21.71} & 0.351 & 0.378 & 0.378 & 0.383 & \textbf{0.718} & 0.755 & 0.985 & 1.103 & 0.762 & \textbf{0.126} & 38.69 & 44.31 & 50.75 & 41.10 & \textbf{11.47} \\
\quad S35 & 12.19 & 11.63 & 11.22 & 12.82 & \textbf{20.24} & 0.355 & 0.383 & 0.382 & 0.392 & \textbf{0.697} & 0.797 & 0.906 & 1.002 & 0.661 & \textbf{0.148} & 42.09 & 48.45 & 49.75 & 38.86 & \textbf{12.24} \\
\bottomrule
\end{tabular}}
\end{table}

\subsubsection{Cross-Scene Category Results}
\label{app:breakdown_unseentype}

\autoref{tab:heldout} reports~\ourSystemPR results averaged over seven excluded scene categories.
\autoref{tab:e3perscene} and~\autoref{tab:e3perscene_tr} provide scene-level comparisons with few-shot and transfer baselines.
Their unweighted scene averages may differ from the fold-level means in~\autoref{tab:heldout}.
The category comparisons below pool results over each category's five scenes.

\textbf{Excluding an Entire Category Causes a Small Additional Drop.}
Relative to unseen-variant evaluation, the unweighted scene summaries for~\ourSystemPR decrease by~\(0.14\,\mathrm{dB}\) in~PSNR and~\(0.007\) in~SSIM; GRaF changes by~\(0.18\,\mathrm{dB}\) and~\(0.005\), respectively.
For~\ourSystemPR, SSIM falls on~\(24\) scenes and rises on~\(11\), with the largest scene-level decrease of~\(0.045\) on~S21.
The largest category-level change is~\(0.024\) for open offices, while the overall category mean moves from~\(0.750\) to~\(0.743\).
Thus, category exclusion causes a modest average loss, although its effect varies by scene.

\textbf{The Margin over GRaF Remains Large.}
\ourSystemPR exceeds~GRaF on all four metrics in every category.
Its category-level SSIM margin ranges from~\(0.18\) on corridors to~\(0.41\) on bedrooms.
AoA error ranges from~\(4.1^\circ\) to~\(20.2^\circ\) for~\ourSystemPR, compared with~\(23.9^\circ\) to~\(38.0^\circ\) for~GRaF.
Across individual scenes,~\ourSystemPR spans~\(0.547\)--\(0.897\)~SSIM, showing that absolute synthesis quality still varies despite the consistent margin.

\textbf{Source-Scene Transfer Has Mixed Effects.}
With target references fixed, using a source scene outside the target category changes~PSNR by~\(-0.02\,\mathrm{dB}\) for~NeRF$^2$, \(+0.07\,\mathrm{dB}\) for~GSRF, and~\(-0.70\,\mathrm{dB}\) for~WRF-GS relative to within-category transfer.
A same-category source therefore provides no consistent advantage.
Likewise,~\autoref{tab:heldout} shows that transfer improves some GSRF metrics but degrades NeRF$^2$ and WRF-GS relative to their few-shot fits.

\textbf{The Advantage Persists Across Scenes.}
\ourSystemPR leads on all~\(35\) scenes and all four metrics against both baseline families, covering~\(280\) scene-metric comparisons.
The aggregate gain is therefore not concentrated in a few favorable scenes.

\subsubsection{Power Refinement}

This ablation repeats the unseen-variant and unseen-category protocols of~\autoref{tab:seen} and~\autoref{tab:heldout} in an independent run.
All variants condition the same pretrained model on~\(M=32\) target-scene references.
\emph{Feed-forward} denotes~\ourSystem and performs no target-scene fitting.
\emph{Anchor offsets} fits one scene-specific power correction per anchor, shared across queries.
\emph{Query network} fits a network that adjusts component powers for each query.
\emph{Both} fits these two corrections and denotes~\ourSystemPR.
The pretrained model, anchor positions, and arrival directions remain fixed.
Each fitted variant uses~\(24\) references for optimization and~\(8\) for step selection.
Setup time includes conditioning and fitting to the selected step; query time is per synthesized spectrum.

\begin{table}[!htbp]
\centering
\vspace{0.1in}
\caption{\textit{Power-refinement ablation.} Mean~\(\pm\) standard deviation over folds.}
\label{tab:revision:ffpa}
\setlength{\tabcolsep}{3.5pt}
\resizebox{\linewidth}{!}{%
\begin{tabular}{llccccccc}
\toprule
Setting & Variant & PSNR$\uparrow$ & SSIM$\uparrow$ & NMSE$\downarrow$ & AoA$\downarrow$ & Adapted params. & Setup & Query \\
\midrule
\multirow{4}{*}{Unseen variants}
& Feed-forward & 22.95$\pm$1.94 & 0.740$\pm$0.054 & 0.121$\pm$0.050 & 8.40$\pm$3.85 & 0 & 0.01 s & 8.5 ms \\
& Anchor offsets & 23.16$\pm$1.98 & 0.751$\pm$0.052 & 0.116$\pm$0.048 & 8.39$\pm$3.85 & 0.004M & 202 s & 8.5 ms \\
& Query network & 22.90$\pm$2.11 & 0.743$\pm$0.054 & 0.124$\pm$0.052 & 8.72$\pm$3.94 & 0.266M & 47 s & 8.5 ms \\
& Both & 22.92$\pm$2.09 & 0.744$\pm$0.054 & 0.124$\pm$0.052 & 8.70$\pm$4.02 & 0.270M & 80 s & 8.5 ms \\
\midrule
\multirow{4}{*}{Unseen categories}
& Feed-forward & 22.83$\pm$2.69 & 0.733$\pm$0.074 & 0.124$\pm$0.069 & 8.41$\pm$5.18 & 0 & 0.01 s & 8.5 ms \\
& Anchor offsets & 23.01$\pm$2.70 & 0.743$\pm$0.074 & 0.120$\pm$0.067 & 8.40$\pm$5.17 & 0.004M & 202 s & 8.5 ms \\
& Query network & 22.84$\pm$2.73 & 0.740$\pm$0.074 & 0.123$\pm$0.067 & 8.67$\pm$5.25 & 0.266M & 47 s & 8.5 ms \\
& Both & 22.83$\pm$2.72 & 0.739$\pm$0.074 & 0.124$\pm$0.068 & 8.69$\pm$5.32 & 0.270M & 80 s & 8.5 ms \\
\bottomrule
\end{tabular}}
\vspace{0.1in}
\end{table}

As shown in~\autoref{tab:revision:ffpa}, joint refinement changes PSNR relative to feed-forward~\ourSystem by~\(-0.03\,\mathrm{dB}\) on unseen variants and~\(+0.01\,\mathrm{dB}\) on unseen categories.
The respective~\(95\%\) confidence intervals are~\([-0.24,0.16]\) and~\([-0.15,0.16]\).
Anchor offsets give the largest gains, \(0.21\) and~\(0.18\,\mathrm{dB}\), but require~\(202\) seconds of setup.
The query network peaks near~\(50\) optimization steps and then overfits, falling to roughly~\(21.16\,\mathrm{dB}\) after~\(3{,}000\) steps.
The selection gate rejects a fitted correction in~\(28\) of~\(70\) scene evaluations.
These results show that the cross-scene advantage is largely present in feed-forward~\ourSystem.
Power refinement provides limited additional accuracy at a one-time setup cost, while leaving per-query rendering time unchanged.

\begin{figure}[t]
\centering
\includegraphics[width=\linewidth]{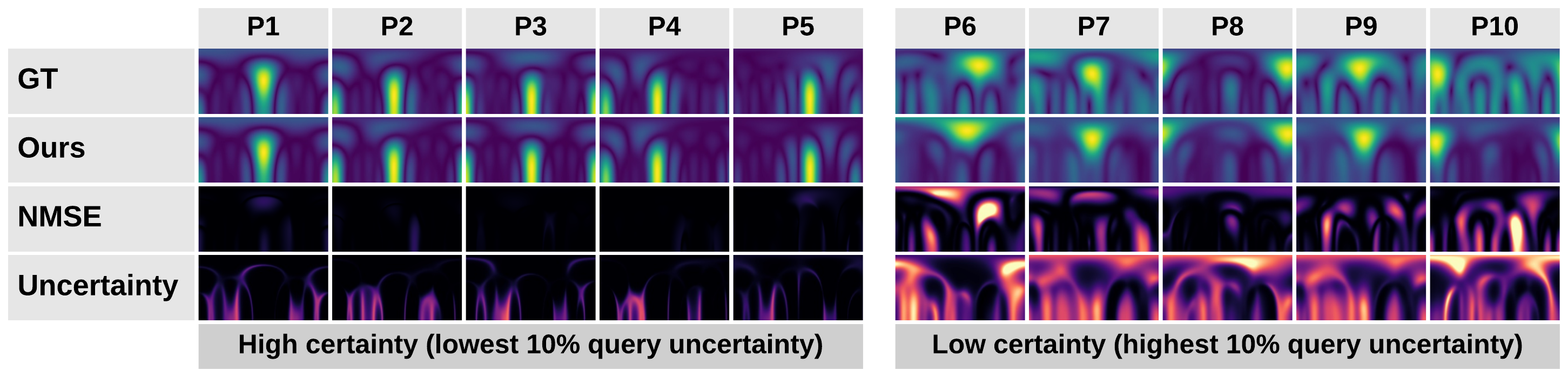}
\caption{\textit{Predictions at low and high uncertainty.}
P1--P5 are low-uncertainty queries and P6--P10 are high-uncertainty queries, paired by scene. Rows show ground truth (GT), predicted mean (Ours), normalized squared error, and predicted standard deviation.}
\label{fig:qual_uncertainty}
\end{figure}

\begin{table}[t]
\centering\small
\caption{\textit{Uncertainty scale and within-spectrum error localization.}
Results are mean~\(\pm\) standard deviation across folds.}
\label{tab:calib}
\begin{tabular}{lcccc}
\toprule
Setting & \(r_{\mathrm{pix}}\) & cov@\(1\sigma\) & cov@\(2\sigma\) & \(c_\sigma\) \\
\midrule
Unseen scene variants & $-0.027$\,{\scriptsize$\pm$0.034} & 0.809\,{\scriptsize$\pm$0.068} & 0.931\,{\scriptsize$\pm$0.033} & 1.11\,{\scriptsize$\pm$0.28} \\
Unseen scene categories & $-0.034$\,{\scriptsize$\pm$0.088} & 0.803\,{\scriptsize$\pm$0.111} & 0.925\,{\scriptsize$\pm$0.062} & 1.12\,{\scriptsize$\pm$0.44} \\
\bottomrule
\end{tabular}
\end{table}

\subsection{Reliability}
\label{app:reliability_analysis}
\suppressfloats[t]

We first examine whether the variance from feed-forward~\ourSystem ranks query difficulty, has an appropriate overall scale, and localizes error within a spectrum.
We then test how optional power refinement affects these properties.

\subsubsection{Feed-Forward Reliability}

The main paper reports the aggregate query-ranking result in~\S\ref{sec:eval:uncertainty}.
\autoref{fig:qual_uncertainty} illustrates the corresponding low- and high-uncertainty predictions.

\textbf{Query-level separation.}
In the first unseen-variant fold, the~\(10\%\) of~\(6{,}436\) test queries with the lowest propagated variance achieve mean~SSIM~\(0.889\) and~NMSE~\(0.017\).
The highest-variance~\(10\%\) achieve~\(0.553\)~SSIM and~\(0.326\)~NMSE.
The separation persists among queries whose ground-truth spectra have at least three lobes: mean~SSIM is~\(0.862\) in the low-uncertainty group and~\(0.539\) in the high-uncertainty group.
For the figure, the two groups are matched by scene, and each column shows the median-SSIM query from that scene within its uncertainty group.
The low-uncertainty examples achieve~\(0.86\)--\(0.89\)~SSIM and~\(0.012\)--\(0.034\)~NMSE.
The high-uncertainty examples achieve~\(0.40\)--\(0.72\)~SSIM and~\(0.11\)--\(0.34\)~NMSE, with larger errors around the dominant lobes.

\textbf{Uncertainty scale and spatial localization.}
We assess the overall scale of predicted uncertainty and whether it identifies angular cells with larger amplitude error~\(\varepsilon=\widehat S-S\).
For each spectrum,~\(r_{\mathrm{pix}}\) is the correlation between predicted standard deviation~\(\sigma\) and~\(\left|\varepsilon\right|\) over~\(400\) randomly sampled angular cells; we then average this correlation across queries.
Computing the correlation within each spectrum prevents differences in query difficulty from driving this spatial measure.
We also report empirical coverage within~\(1\sigma\) and~\(2\sigma\), whose Gaussian reference values are~\(0.68\) and~\(0.95\), and the scale factor
\begin{equation}
c_\sigma
=
\sqrt{
\mathbb E
\left[
\frac{\varepsilon^2}{\sigma^2}
\right]
},
\label{eq:app_uncertainty_scale}
\end{equation}
where~\(c_\sigma\approx1\) indicates approximately correct average scale.

\autoref{tab:calib} shows that the average uncertainty scale is similar in the two simulated generalization settings.
Coverage within~\(2\sigma\) is about~\(0.93\), near the Gaussian reference of~\(0.95\), and the scale factors are~\(1.11\) and~\(1.12\).
Coverage within~\(1\sigma\) is higher than the Gaussian reference, at about~\(0.81\)~\versus~\(0.68\), indicating that the residual distribution is not exactly Gaussian.
Spatial localization is much weaker than query-level ranking.
The within-spectrum correlations are approximately~\(-0.03\) in both settings, so angular cells assigned higher uncertainty are not systematically those with larger errors in the same spectrum.
The propagated variance is therefore most informative as a query-level reliability signal.

\begin{table}[t]
\centering
\small
\caption{\textit{Amplitude-domain reliability before and after power refinement.}}
\label{tab:revision:uqpa}
\resizebox{\linewidth}{!}{%
\begin{tabular}{lrrrrrr}
\toprule
Condition & Spearman & NLL & Cov.~\(1\sigma\) & Cov.~\(2\sigma\) & \(c_\sigma\) & Sharpness \\
\midrule
CoRF & 0.603 & -0.630 & 0.809 & 0.930 & 1.10 & 0.1110 \\
\ourSystemPR mean, CoRF variance & 0.599 & -0.641 & 0.804 & 0.929 & 1.09 & 0.1110 \\
\ourSystemPR variance at refined powers & 0.523 & -0.650 & 0.819 & 0.938 & 1.04 & 0.1217 \\
\ourSystemPR variance, 8-reference recalibration & 0.498 & -0.741 & 0.762 & 0.922 & 1.20 & 0.1057 \\
\bottomrule
\end{tabular}}
\end{table}

\subsubsection{Reliability after Power Refinement}
\label{app:revision:uqpa}

We test whether power refinement preserves query-level ranking and amplitude-domain calibration on the first~\(300\) test queries per scene across all~\(12\) folds.
In~\autoref{tab:revision:uqpa}, Spearman correlates mean predicted variance with mean squared amplitude error across queries, while NLL is pixel-level Gaussian negative log-likelihood.
Coverage is the fraction of angular cells with absolute error within one or two predicted standard deviations.
The scale factor~\(c_\sigma\) is defined in~\Eqref{eq:app_uncertainty_scale}, and sharpness is the mean predicted standard deviation.

Carrying the feed-forward variance unchanged to the refined mean nearly preserves Spearman correlation, from~\(0.603\) to~\(0.599\), and~\(1\sigma\) coverage, from~\(0.809\) to~\(0.804\).
Recomputing variance at refined powers lowers the correlation to~\(0.523\).
Recalibration using eight selection references improves NLL but lowers~\(1\sigma\) coverage to~\(0.762\) and raises~\(c_\sigma\) to~\(1.20\), indicating underestimated error scale.
These results favor carrying the feed-forward variance unchanged alongside a power-refined mean.
Converting variance to the power domain by the delta method requires scale factors of about~\(9.7\)--\(11.8\); the calibration results here concern normalized amplitude.

\begin{table}[t]
\centering\small
\caption{\textit{Tolerance to array-parameter error on \(\mathcal A_2\) (\(8\times2\)).}
Uncalibrated results use the corrupted array parameters directly, while calibrated results fit \(\boldsymbol\nu_{\mathrm{arr}}\) using \(N_{\mathrm C}=64\) measurements per scene.}
\label{tab:array_error}
\setlength{\tabcolsep}{4pt}
\resizebox{\linewidth}{!}{%
\begin{tabular}{llllllllll}
\toprule
& & \multicolumn{4}{l}{Uncalibrated} & \multicolumn{4}{l}{Calibrated} \\
\cmidrule(lr){3-6}\cmidrule(lr){7-10}
Orientation & Spacing & PSNR\(\uparrow\) & SSIM\(\uparrow\) & NMSE\(\downarrow\) & AoA\(\downarrow\) & PSNR\(\uparrow\) & SSIM\(\uparrow\) & NMSE\(\downarrow\) & AoA\(\downarrow\) \\
\midrule
\(\phantom{0}0^\circ\) & ---           & 22.32 & 0.701 & 0.095 & 6.73  & 22.21 & 0.701 & 0.096 & 6.73 \\
\(\phantom{0}2^\circ\) & \(+2/-2\%\)   & 21.22 & 0.671 & 0.107 & 6.74  & 22.17 & 0.700 & 0.096 & 6.73 \\
\(\phantom{0}4^\circ\) & \(+4/-3\%\)   & 19.79 & 0.624 & 0.134 & 7.70  & 22.22 & 0.700 & 0.096 & 6.73 \\
\(\phantom{0}6^\circ\) & \(+7/-6\%\)   & 18.45 & 0.570 & 0.175 & 9.01  & 22.09 & 0.697 & 0.097 & 6.78 \\
\(\phantom{0}9^\circ\) & \(+10/-9\%\)  & 17.05 & 0.514 & 0.235 & 10.44 & 22.14 & 0.697 & 0.096 & 6.77 \\
\(12^\circ\)           & \(+14/-12\%\) & 16.01 & 0.475 & 0.293 & 12.22 & 22.07 & 0.695 & 0.097 & 6.71 \\
\bottomrule
\end{tabular}%
}
\end{table}

\subsection{Tolerance to Array-Parameter Error}
\label{app:array_error}

\S\ref{sec:eval:array} evaluates receiver calibration at one level of geometry mismatch.
Here we vary the nominal array error and test how much accuracy calibration recovers.

\textbf{Deployment Requirements for a New Array.}
The analytic decoder requires receiver element positions, carrier wavelength, and array position and orientation.
Element geometry and wavelength are typically known from the array specification, while receiver position is obtained at installation.
Orientation may be reported by deployed infrastructure or estimated through calibration.
For example, AISG supports reporting antenna azimuth, tilt, and roll~\citep{aisg2024als}, while 3GPP specifies positioning information and array-orientation conventions~\citep{3gpp38455,3gpp37355,3gpp38901}.
When this metadata is unavailable or inaccurate, the receiver response can be calibrated using measurements collected at known locations, as demonstrated on commercial~5G arrays~\citep{pan2023insitu}.
A new array therefore requires either an accurate response specification or receiver-specific calibration measurements.

\textbf{Protocol.}
We use the unseen~\(8\times2\) receiver array~\(\mathcal A_2\) and perturb its in-plane orientation and element spacing over six levels, from exact geometry to~\(12^\circ\) of orientation error with~\(+14\%\) and~\(-12\%\) axis scaling.
At each level, we evaluate the nominal response and a response calibrated with~\(N_{\mathrm C}=64\) measurements.
The propagation model,~\(M=32\) target-scene references, and test queries remain unchanged; calibration measurements are disjoint from the references and queries.

\textbf{Uncalibrated Accuracy Degrades with Array Error.}
\autoref{tab:array_error} reports the results.
As the nominal geometry becomes less accurate, each additional~\(2^\circ\)--\(3^\circ\) of orientation error with its corresponding spacing error reduces~PSNR by approximately~\(1.0\)--\(1.4\,\mathrm{dB}\).
At the largest perturbation, PSNR falls to~\(16.01\,\mathrm{dB}\), approaching the~\(14.96\,\mathrm{dB}\) obtained by retaining the training-array response for~\(\mathcal A_2\).
NMSE increases from~\(0.095\) to~\(0.293\), and~AoA error from~\(6.73^\circ\) to~\(12.22^\circ\).
Thus, receiver-geometry errors distort rendered spectrum even though the predicted propagation components are unchanged.

\textbf{Calibration Recovers In-Family Geometry Errors.}
After calibration, PSNR remains between~\(22.07\) and~\(22.22\,\mathrm{dB}\) across the perturbation range.
NMSE remains between~\(0.096\) and~\(0.097\), and~AoA error between~\(6.71^\circ\) and~\(6.78^\circ\).
The injected errors match the fitted calibration family: one in-plane rotation and two axis scalings.
After fitting, recovered element positions differ from their true positions by~\(0.010\)--\(0.028\) wavelengths on average, with maximum displacement below~\(0.058\lambda\).
These results establish recovery within the modeled geometry family, not for arbitrary receiver mismatch.

\textbf{Calibration Becomes More Valuable as Mismatch Increases.}
Calibration improves~PSNR by~\(0.95\,\mathrm{dB}\) at the smallest nonzero perturbation and~\(6.06\,\mathrm{dB}\) at the largest.
With exact geometry, it changes~PSNR by only~\(-0.11\,\mathrm{dB}\).
Calibration therefore adds little when the receiver specification is accurate but becomes more valuable as its geometry becomes less reliable.

\textbf{Calibration Robustness.}
Under the same in-family perturbation, using~\(N_{\mathrm C}=8,16,32,\) and~\(64\) calibration measurements yields~\(21.96\), \(22.03\), \(22.05\), and~\(22.08\,\mathrm{dB}\), respectively, compared with~\(14.96\,\mathrm{dB}\) without calibration and~\(22.32\,\mathrm{dB}\) with exact geometry.
Ten random initializations at~\(N_{\mathrm C}=64\) yield~\(22.06\pm0.02\,\mathrm{dB}\).
Calibration-measurement SNRs of~\(20\), \(10\), and~\(0\,\mathrm{dB}\) yield~\(22.07\), \(22.00\), and~\(21.33\,\mathrm{dB}\).
Thus, the geometric fit is stable across these measurement budgets and initializations, although severe measurement noise reduces accuracy.

\subsection{Case Study: Angular Localization with Synthesized Spatial Spectra}
\label{app:localization}
\suppressfloats[t]

We evaluate whether spectra synthesized from sparse references can augment training for downstream localization; beam management is studied in~\S\ref{sec:eval:beam}.

\begin{figure}[t]
\centering
\includegraphics[width=0.8\linewidth]{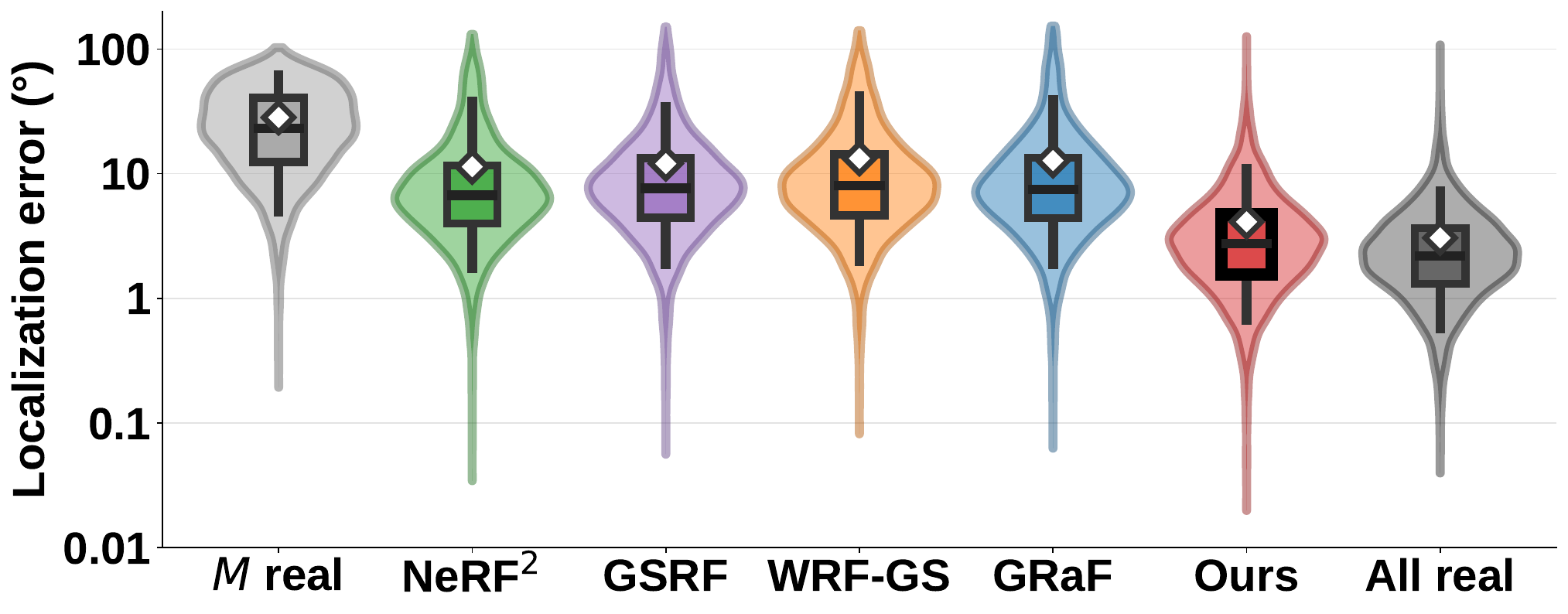}
\caption{\textit{Localization with \(M=32\) reference spectra and synthesized training spectra.}
Only the \(32\) references per scene are treated as available observations; each method synthesizes the remaining training spectra.
Localization error is evaluated on the same test queries across unseen scene variants.
}
\label{fig:localization}
\end{figure}

Following~\citet{zhao2023nerf2}, we train an angular artificial neural network~(AANN) with a ResNet-50 backbone and an MLP head.
Given a spatial spectrum, it predicts the geometric direction from the receiver to the transmitter as a unit vector.
The training loss is cosine distance to the direction computed from the known transmitter and receiver positions.

\textbf{Protocol.}
Only the~\(M=32\) target-scene reference spectra are available to each synthesis method.
Each method generates spectra at the remaining training locations; feed-forward methods condition on the references, while per-scene baselines are fitted from scratch using their few-shot configurations.
For synthesis method~\(m\), the localization model is trained on
\begin{equation}
\mathcal D_{\mathrm{loc}}^{m}
=
\mathcal D_{\mathrm{GT}}^{M}
\cup
\widehat{\mathcal D}_{m},
\label{eq:localization_training}
\end{equation}
where~\(\mathcal D_{\mathrm{GT}}^{M}\) contains the~\(32\) reference spectra and~\(\widehat{\mathcal D}_{m}\) contains spectra synthesized at the remaining training locations.
Synthesized spectra retain the transmitter locations and geometric localization labels of their corresponding ground-truth spectra.

We compare seven training conditions: the~\(32\) references alone; augmentation with NeRF$^2$-FS, GSRF-FS, WRF-GS-FS, GRaF, or~\ourSystem; and the full ground-truth training set.
The evaluation covers seven unseen scene variants in one fold, with one scene from each category.
All AANNs use the same architecture, training configuration, and seed.
Every condition is evaluated on the same~\(6{,}436\) ground-truth test spectra, never on synthesized test spectra.
For test query~\(\pi\), localization error is
\begin{equation}
E^{\mathrm{loc}}_\pi
=
d_{\mathrm{ang}}
\left(
\widehat{\Omega}^{\mathrm{loc}}_\pi,
\Omega^{\mathrm{gt}}_\pi
\right),
\label{eq:localization_error}
\end{equation}
where~\(\widehat{\Omega}^{\mathrm{loc}}_\pi\) is the AANN prediction, \(\Omega^{\mathrm{gt}}_\pi\) is the geometric transmitter direction, and~\(d_{\mathrm{ang}}\) is spherical angular distance.

\textbf{Results.}
As shown in~\autoref{fig:localization}, training on the~\(32\) references alone yields mean angular error~\(28.27^\circ\).
Augmentation with~\ourSystem spectra reduces it to~\(4.11^\circ\), close to the~\(3.05^\circ\) obtained with the full ground-truth training set.
This closes~\(96\%\) of the gap between the reference-only and full-data conditions.
The strongest baseline, NeRF$^2$-FS, achieves~\(11.36^\circ\); GSRF-FS, GRaF, and WRF-GS-FS achieve~\(12.06^\circ\), \(12.75^\circ\), and~\(13.27^\circ\), respectively.
With~\ourSystem augmentation,~\(77\%\) of test queries fall within~\(5^\circ\) of the ground-truth direction, compared with~\(28\%\)--\(35\%\) for the baselines.
\ourSystem also yields the lowest mean localization error in each of the seven evaluated scenes.

\textbf{Synthesis Fidelity and Localization.}
Localization labels always use the true geometric directions, so synthesis errors affect the training inputs rather than the labels.
Misplaced or distorted spectral lobes create a mismatch between synthesized training inputs and ground-truth test spectra.
Across method--scene pairs, localization error is negatively correlated with the SSIM of synthesized training spectra, with Pearson~\(r=-0.73\).
This association is consistent with improved synthesis fidelity contributing to lower localization error.

Overall, augmenting~\(32\) reference spectra with~\ourSystem predictions trains a localization model whose accuracy approaches that obtained with the full ground-truth training set.
Augmentation with baseline-synthesized spectra leaves substantially larger localization error.

\subsection{Simulation-to-Real Transfer Study}
\label{app:real}

We transfer the propagation model trained on~\(35\) simulated scenes to measured scene~S36, which uses a physical~\(4\times4\) receiver array.
Each method has a budget of~\(96\) target-domain spectra, although the methods use these measurements differently.
Both~\ourSystemCAL and~\ourSystemPRCAL condition on~\(M=32\) references and calibrate the receiver with~\(N_{\mathrm C}=64\) disjoint spectra.
\ourSystemPRCAL also reuses both sets for power refinement without additional measurements.

\begin{table}[h!]
\centering
\small
\vspace{0.1in}
\caption{\textit{Measurement-matched transfer to measured scene~S36.}}
\label{tab:revision:real}
\begin{tabular}{lrrrr}
\toprule
Method & PSNR$\uparrow$ & SSIM$\uparrow$ & NMSE$\downarrow$ & AoA$\downarrow$ \\
\midrule
NeRF$^2$ & 12.94 & 0.457 & 0.287 & 16.24 \\
GSRF & 16.85 & 0.587 & 0.145 & 12.09 \\
WRF-GS+ & 16.96 & 0.592 & 0.140 & 9.26 \\
GRaF & 14.24 & 0.510 & 0.219 & \textbf{7.87} \\
\midrule
\ourSystemCAL & 16.74 & 0.620 & 0.125 & 8.12 \\
\ourSystemPRCAL & \textbf{17.45} & \textbf{0.653} & \textbf{0.118} & 7.88 \\
\bottomrule
\end{tabular}
\vspace{0.1in}
\end{table}

In~\autoref{tab:revision:real}, \ourSystemPRCAL leads in~PSNR, SSIM, and~NMSE, exceeding~WRF-GS+ by~\(0.49\,\mathrm{dB}\) in~PSNR and~\(0.061\) in~SSIM.
Even without power refinement, \ourSystemCAL achieves higher~SSIM and lower~NMSE and~AoA error than~WRF-GS+, despite slightly lower~PSNR.
GRaF has nearly the same~AoA error as~\ourSystemPRCAL but substantially worse full-spectrum metrics.
AoA measures the strongest direction, while~SSIM and~NMSE also reflect the locations and relative strengths of secondary spectral structure.

The calibration controls show that modifying the receiver response removes substantial error while the conditioned propagation estimates remain fixed.
Uncalibrated feed-forward~\ourSystem reaches~\(14.79\,\mathrm{dB}\) PSNR.
Correcting array rotation and spacing with three geometric parameters raises this to~\(14.98\,\mathrm{dB}\), whereas the element-wise response model reaches~\(16.74\,\mathrm{dB}\).
The small geometric gain suggests that a global rotation and two spacing corrections cannot capture most of the correctable response discrepancy.
The element-wise model additionally represents gain, phase, position, and angular-response differences, providing a plausible mechanism for its larger gain.
Because these parameters are fitted together, the comparison does not identify which element-wise effect contributes most.

Power refinement addresses a different part of the synthesis pipeline by correcting predicted component strengths while the pretrained propagation model remains frozen.
Under the same measurement budget, \ourSystemPRCAL reaches~\(17.45\,\mathrm{dB}\).
This further improvement is consistent with component-power error remaining after receiver calibration.
The two adaptations therefore address complementary discrepancies: calibration changes the mapping from predicted propagation to the measured spectrum, while refinement changes the powers supplied to that mapping.

The~\(0.49\,\mathrm{dB}\) real-data margin is smaller than the~\(7.16\,\mathrm{dB}\) margin over the strongest few-shot baseline on unseen simulated variants.
These comparisons use different measurement budgets and adaptation settings, so their margins should be interpreted within their respective protocols.
In the measured-scene comparison, WRF-GS+ fits its field to~S36, whereas~\ourSystem retains a simulation-trained propagation model and adapts its receiver response and component powers.
The baseline can therefore fit scene-specific multipath directly, while~\ourSystem must infer it from the transferred model and sparse references.
Calibration and power refinement can correct the response and strengths of predicted arrivals without retraining the learned propagation model, but neither can add a missing propagation component.
Secondary reflected structure may remain unresolved even when the dominant direction is accurate, consistent with the failures in~\S\ref{app:qualitative:fails}.

A separate larger-budget diagnostic uses the~\(612\)-spectrum validation split in addition to the~\(M=32\) conditioning references.
Fitting the receiver response reaches~\(17.39\,\mathrm{dB}\) PSNR, while fitting both the response and power correction reaches~\(19.99\,\mathrm{dB}\).
This shows that additional target-domain data can improve both corrections, but the resulting~\(644\)-spectrum protocol differs from the matched~\(96\)-spectrum comparison in~\autoref{tab:revision:real}.
In a deployment with an accurately specified effective receiver response, receiver-specific calibration measurements could be omitted.
Hardware specifications and network metadata may provide geometry and pose~\citep{aisg2024als,3gpp38455}, but the~\(19.99\,\mathrm{dB}\) experiment also uses validation spectra to fit the scene-specific power correction.
It therefore does not establish the same accuracy from metadata and~\(M=32\) references alone.

\subsection{Role of Target-Scene References}
\label{app:reference_analysis}
\suppressfloats[t]

\begin{figure}[t]
\centering
\subfigure[Unseen scene variants]{\includegraphics[width=\linewidth]{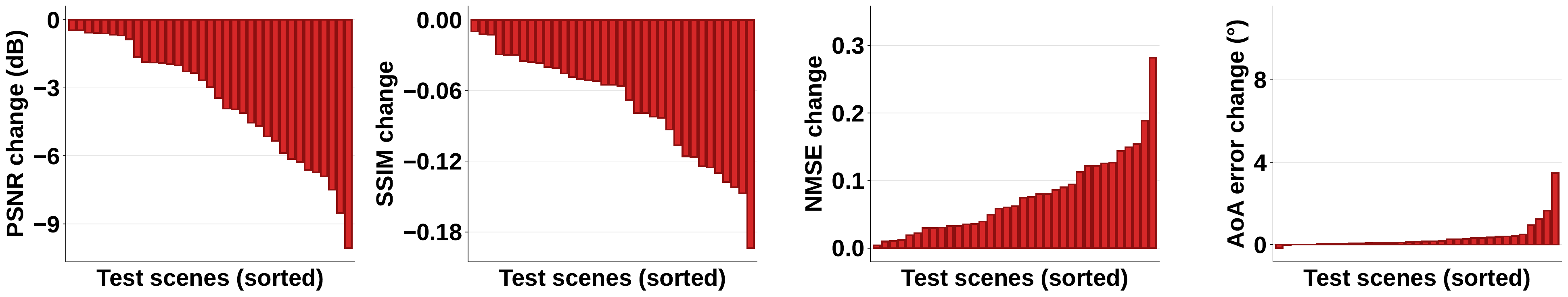}\label{fig:reference_usage_drop_a}}\\[4pt]
\subfigure[Unseen scene categories]{\includegraphics[width=\linewidth]{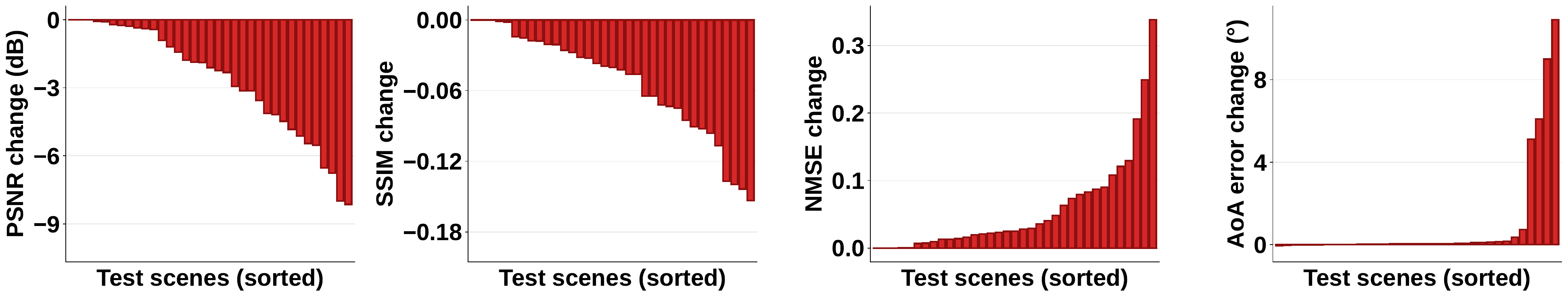}\label{fig:reference_usage_drop_b}}
\caption{\textit{Effect of cross-scene reference replacement.}
For each test scene, a replacement set is selected by the largest PSNR loss on~\(150\) queries.
Each bar shows its mean metric change over all test queries relative to correct-scene references, with scenes sorted separately for each metric.}
\label{fig:reference_usage_drop}
\end{figure}

\subsubsection{Reference Dependence}
\label{app:reference_usage}

A conditional model could achieve strong accuracy while relying primarily on query geometry and ignoring the reference spectra.
We therefore directly test whether~\ourSystem uses the target-scene references by replacing them with references from other scenes while keeping the trained model and test queries fixed.
All model parameters remain fixed throughout this experiment, so any performance change is caused only by changing the reference set.

\textbf{Protocol.}
For each test scene in the unseen-variant and unseen-category evaluations, we replace its~\(M=32\) references with the reference set from each of the other~\(34\) simulated scenes.
For each replacement scene, we first evaluate~\(150\) evenly spaced test queries and identify the reference set that produces the largest decrease in~PSNR.
We then evaluate all test queries using both the correct target-scene references and this selected replacement set.
We additionally average over all~\(34\) replacement scenes to determine whether the effect persists beyond the single reference set producing the largest degradation.

\textbf{Correct Scene References Matter.}
\autoref{fig:reference_usage_drop} shows that predictions respond to which scene supplies the references.
With the correct references, the model achieves~\(22.95\,\mathrm{dB}\) PSNR and~\(0.740\) SSIM on unseen scene variants, and~\(22.83\,\mathrm{dB}\) and~\(0.733\) on unseen scene categories.
Replacing them with the reference set that produces the largest degradation reduces~PSNR by~\(3.62\,\mathrm{dB}\) and~\(2.65\,\mathrm{dB}\), reduces~SSIM by~\(0.072\) and~\(0.053\), and increases~NMSE by~\(63\%\) and~\(45\%\), respectively.
These are stress-test losses from the selected replacements, not typical scene swaps.

The effect occurs broadly across test scenes rather than being driven by a few extreme cases.
For unseen scene variants, every test scene loses at least~\(0.46\,\mathrm{dB}\), with the largest decrease reaching~\(10.1\,\mathrm{dB}\).
For unseen scene categories, the largest decrease reaches~\(8.1\,\mathrm{dB}\).
The effect also remains when averaging over all~\(34\) replacement scenes:~PSNR decreases by~\(0.96\,\mathrm{dB}\) for unseen variants and~\(0.51\,\mathrm{dB}\) for unseen categories.
Scenes that are more sensitive to the selected replacement set are also more sensitive on average, with Spearman correlations of~\(\rho=0.92\) and~\(0.77\), respectively.
These results show that the dependence on the target-scene references is systematic rather than arising from a small number of unfavorable scene pairs.

\textbf{Reference Sensitivity Depends on Scene Structure.}
The magnitude of the degradation varies across target scenes.
For unseen scene variants, the reference-swap PSNR change is negatively correlated with the PSNR obtained using the correct references, with Spearman correlation~\(\rho=-0.83\).
Classroom and bedroom scenes can lose approximately~\(4\)--\(6\,\mathrm{dB}\), whereas the most difficult corridor scenes often lose less than~\(1\,\mathrm{dB}\).
This suggests that when the model captures more scene-specific spectral structure, replacing the corresponding reference evidence removes more useful information.

\textbf{References Primarily Affect Full-Spectrum Structure.}
Changing the references has a much smaller effect on dominant-bearing accuracy than on full-spectrum synthesis.
Mean~AoA error increases by only~\(0.35^\circ\) for unseen variants and~\(0.91^\circ\) for unseen categories, despite the substantially larger changes in~PSNR, SSIM, and~NMSE.
This is consistent with the analytic decoder preserving geometric information about arrival directions while the reference-conditioned field supplies scene-specific information needed to synthesize the complete spatial spectrum.

\begin{table}[t]
\centering
\small
\caption{\textit{Reference-budget sensitivity.}
Entries are median PSNR losses in~\(\mathrm{dB}\) across scenes relative to each scene's mean at~\(M=32\).
The average first averages over ten reference draws; the least favorable result uses the worst draw for each scene.}
\label{tab:reference_budget_summary}
\begin{tabular}{lcccc}
\toprule
& \multicolumn{2}{c}{Unseen variants} & \multicolumn{2}{c}{Unseen categories} \\
\cmidrule(lr){2-3}\cmidrule(lr){4-5}
Budget & Average & Least favorable & Average & Least favorable \\
\midrule
\(M=1\) & 0.62 & 2.40 & 0.41 & 2.10 \\
\(M=8\) & \(<0.05\) & 0.53 & \(<0.05\) & 0.30 \\
\(M=32\) & 0 & 0.16 & 0 & 0.09 \\
\bottomrule
\end{tabular}
\end{table}

\subsubsection{Reference Budget}
\label{app:reference_budget}

The main experiments use a reference budget of~\(M=32\).
Here we vary the number of references available at inference while keeping the trained model fixed, allowing us to isolate how reference coverage affects cross-scene prediction.

\textbf{Protocol.}
For each test scene, we evaluate
\[
M\in\left\{1,2,4,8,16,32,64,128\right\}
\]
using~\(R=10\) reference sets per budget.
The reference sets are nested within each draw, so the references available at a smaller budget are retained when the budget increases.
At~\(M=32\), draw~\(0\) corresponds to the reference set used in the main evaluation.
All unseen-variant and unseen-category folds are evaluated on their complete test splits.
For each scene, we report performance relative to its mean result at~\(M=32\), considering both the average over the~\(10\) draws and the least favorable draw.
At~\(M=32\), the average difference is zero by definition, but the least favorable draw can fall below the scene mean.

\textbf{Average Performance Saturates Quickly.}
\autoref{tab:reference_budget_summary} shows that relatively few references are sufficient to recover most of the average performance.
At~\(M=1\), the median~PSNR decrease across scenes is only~\(0.62\,\mathrm{dB}\) for unseen scene variants and~\(0.41\,\mathrm{dB}\) for unseen scene categories.
By~\(M=8\), the median decrease is below~\(0.05\,\mathrm{dB}\) in both settings.
Increasing the budget beyond~\(M=32\) produces negligible additional improvement.
Thus, the mean prediction quality saturates after a relatively small number of target-scene references.

\begin{figure}[t]
\centering
\subfigure[Unseen scene variants]{\includegraphics[width=\linewidth]{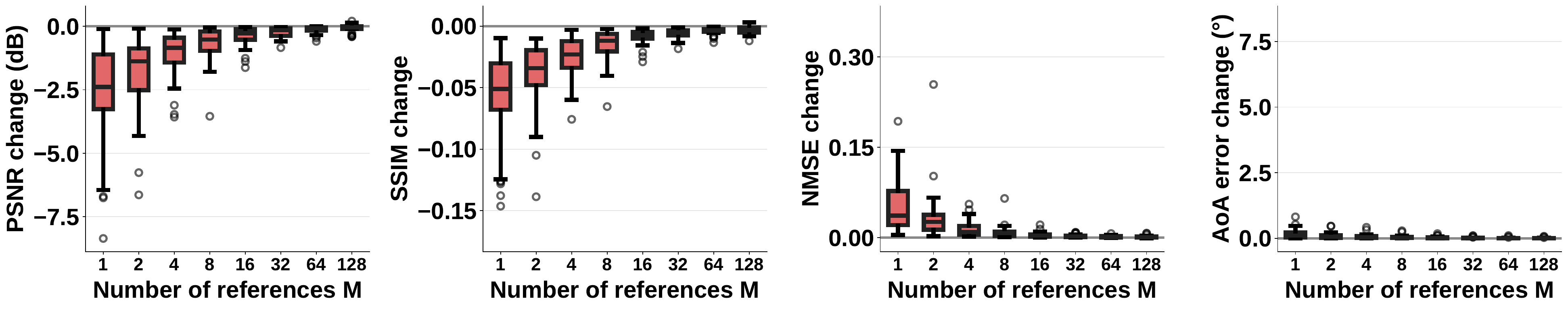}\label{fig:reference_budget_scene_a}}\\[4pt]
\subfigure[Unseen scene categories]{\includegraphics[width=\linewidth]{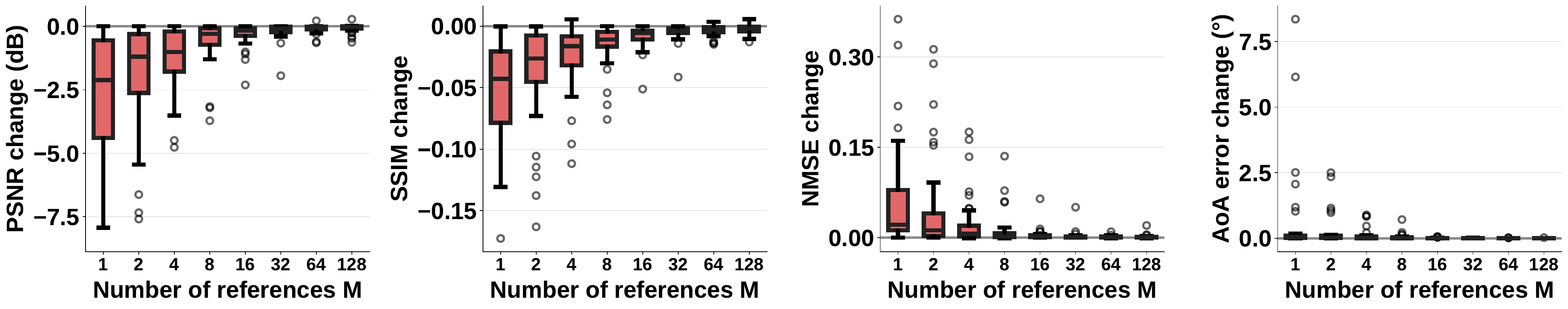}\label{fig:reference_budget_scene_b}}
\caption{\textit{Effect of the reference budget across test scenes.}
For each budget \(M\), each point shows the worst change over ten reference draws relative to the scene's average performance at \(M=32\). Results are shown for unseen scene variants and unseen scene categories.}
\label{fig:reference_budget_scene}
\end{figure}

\textbf{More References Improve Selection Robustness.}
Although the average result saturates quickly, small reference budgets are substantially more sensitive to which measurements are selected.
\autoref{fig:reference_budget_scene} shows the least favorable draw for each scene across reference budgets.
At~\(M=1\), the~\(10\) reference draws span a median of~\(2.4\,\mathrm{dB}\) across unseen-variant scenes and~\(3.0\,\mathrm{dB}\) across unseen-category scenes.
The least favorable draw decreases~PSNR by a median of~\(2.4\,\mathrm{dB}\) and~\(2.1\,\mathrm{dB}\), respectively, with individual scenes losing as much as~\(8.4\,\mathrm{dB}\) and~\(7.9\,\mathrm{dB}\).

This variability decreases rapidly as more references are provided.
At~\(M=4\), the loss under the least favorable draw decreases to~\(0.85\,\mathrm{dB}\) for unseen variants and~\(1.02\,\mathrm{dB}\) for unseen categories.
At~\(M=8\), it further decreases to~\(0.53\,\mathrm{dB}\) and~\(0.30\,\mathrm{dB}\), and at~\(M=32\) to only~\(0.16\,\mathrm{dB}\) and~\(0.09\,\mathrm{dB}\).

\textbf{Reference Sensitivity Varies Across Scenes.}
Scenes that achieve higher fidelity with the default reference set generally exhibit greater sensitivity when only one or a few references are available.
For unseen scene variants, the signed~PSNR change under the least favorable draw at~\(M=1\) is negatively correlated with~PSNR at~\(M=32\), with Spearman correlation~\(\rho=-0.60\).
For example, some classroom scenes lose several~\(\mathrm{dB}\) with an unfavorable single reference, whereas the most difficult corridor scenes change much less.
This pattern is consistent with~\S\ref{app:reference_usage}: scenes for which the conditioned field captures more scene-specific structure also depend more strongly on receiving representative target-scene evidence.

Overall, increasing the reference budget primarily reduces sensitivity to reference selection rather than substantially increasing average accuracy.
A small number of references is often sufficient for a typical prediction, but~\(M=32\) provides substantially more robust performance across reference selections and scenes while additional references beyond this budget yield little further benefit.
This nested-draw experiment begins at~\(M=1\) and does not evaluate reference-free inference.

\begin{figure}[t]
\centering
\subfigure[Unseen scene variants]{\includegraphics[width=\linewidth]{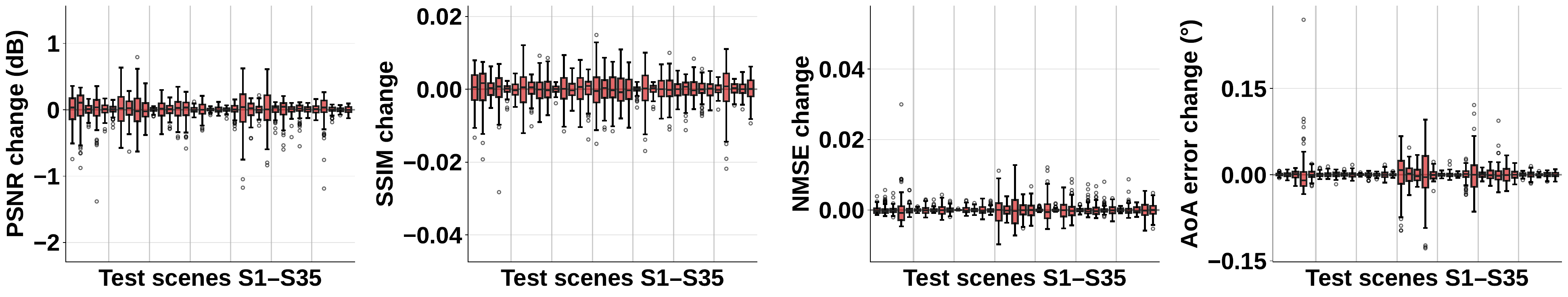}\label{fig:reference_sampling_a}}\\[4pt]
\subfigure[Unseen scene categories]{\includegraphics[width=\linewidth]{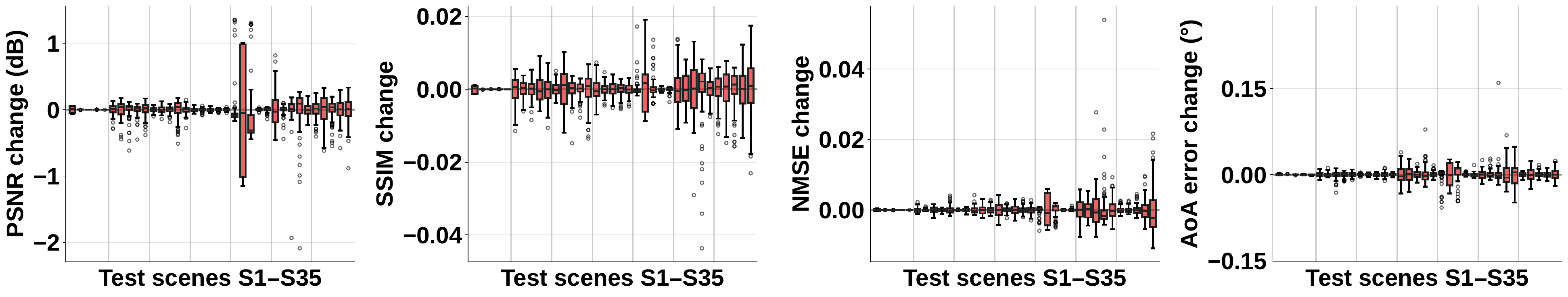}\label{fig:reference_sampling_b}}
\caption{\textit{Variation across \(R=100\) reference draws.}
Each box summarizes the metric change for one test scene relative to its mean across draws at \(M=32\). 
}
\label{fig:reference_sampling}
\end{figure}

\subsubsection{Reference-Set Sampling}
\label{app:reference_sampling}

The main experiments use a fixed budget of~\(M=32\) target-scene references, but the particular measurements available in practice may vary.
We therefore evaluate how sensitive~\ourSystem is to reference selection while keeping the reference budget, trained model, and test queries fixed.

\textbf{Protocol.}
For each test scene~\(s\), let~\(\mathcal R_s\) denote its pool of candidate reference measurements.
We independently sample~\(R=100\) reference sets,
\begin{equation}
\mathcal C_s^{\left(r\right)}
\sim
\operatorname{UniformSubset}
\left(
\mathcal R_s,
M
\right),
\qquad
r=0,\ldots,R-1,
\label{eq:reference_sampling}
\end{equation}
with~\(M=32\).
Sampling is performed without replacement within each set, while the test queries remain fixed and disjoint from all sampled references.
The reference sets are selected before evaluation, and~\(r=0\) corresponds to the original reference selection; this analysis evaluates the fixed feed-forward model, not the power-refined main-table rows.
We evaluate all~\(12\) unseen-variant and unseen-category folds, giving~\(1200\) fold-by-reference-set evaluations in total.
Because the model parameters remain fixed, any variation across repetitions is caused only by the selected references.

\textbf{Performance Is Robust to Reference Selection.}
Averaged over folds and reference sets,~\ourSystem achieves~\(22.95\,\mathrm{dB}\) PSNR, \(0.739\) SSIM, \(0.122\) NMSE, and~\(8.40^\circ\) AoA error on unseen scene variants.
For unseen scene categories, the corresponding values are~\(22.86\,\mathrm{dB}\), \(0.733\), \(0.124\), and~\(8.41^\circ\).
The standard deviation of fold-averaged~PSNR across the~\(100\) reference sets is only~\(0.026\,\mathrm{dB}\) and~\(0.037\,\mathrm{dB}\), respectively, and no reference set changes the fold-averaged~PSNR by more than~\(0.12\,\mathrm{dB}\) relative to the original reference selection.

\autoref{fig:reference_sampling} examines the variation within individual scenes.
For most scenes, reference selection has only a small effect.
Across scenes, the median width of the interquartile range is~\(0.11\,\mathrm{dB}\) for unseen variants and~\(0.07\,\mathrm{dB}\) for unseen categories, while the median span between the 5th and 95th percentiles is~\(0.30\,\mathrm{dB}\) and~\(0.20\,\mathrm{dB}\), respectively.
The reference set used in the main evaluation is also representative, with a median difference from the corresponding scene mean of only~\(-0.02\,\mathrm{dB}\) and~\(0.00\,\mathrm{dB}\).
A small number of scenes exhibit greater sensitivity.
Individual reference sets can shift scene-level~PSNR by as much as~\(1.38\,\mathrm{dB}\) for unseen variants and~\(2.09\,\mathrm{dB}\) for unseen categories.
However, SSIM, NMSE, and especially~AoA remain substantially more stable.
Across all scenes, a single reference selection changes~AoA error by at most~\(0.27^\circ\), indicating that reference selection affects detailed spectrum synthesis much more than the dominant arrival direction.

\textbf{Reference Selection Changes Scene-Level Spectral Structure.}
We further examine the scenes with the largest variation across reference sets.
Changing the references tends to shift performance across many queries in the same scene rather than only at queries near the selected measurements.
Between a scene's lowest- and highest-performing reference sets,~\(63\%\)--\(93\%\) of its test queries improve together, while performance is essentially uncorrelated with distance to the nearest reference.

The dominant spectral peak is also highly stable.
Across these scenes, the predicted peak remains at the same angular cell as in the main reference selection for a median of~\(96\%\)--\(100\%\) of test queries.
Instead, most of the variation appears in the broader spectral background.
Reference sets containing more spectra with rich multipath tend to produce a stronger background response across the scene, whereas reference sets dominated by a strong direct path tend to produce a weaker background.
These relationships are moderate rather than deterministic, indicating that~\ourSystem combines information across the complete reference set rather than relying on individual measurements.

Together with~\S\ref{app:reference_usage}, these results distinguish two properties of the reference-conditioning mechanism.
Changing \emph{which measurements} are selected from the correct target scene usually produces only small variations, whereas replacing them with references from a \emph{different scene} causes substantial degradation.
Thus,~\(M=32\) references provide a stable scene-level representation without requiring a specially selected reference set.

\subsubsection{Local Coverage for GRaF}
\label{app:revision:graf}

GRaF selects nearby reference measurements for each query, so its performance may depend more strongly on local reference coverage than that of \ourSystem.
We examine this effect on unseen-variant fold~1 using~\(150\) test queries per scene and identical nested reference sets for both methods.
The budgets range from~\(M=32\) to~\(M=256\), with each larger set retaining the references in the smaller sets.
We select GRaF's neighbor count from~\(L\in\left\{4,8,16,32\right\}\) using these test queries; the best choice is~\(L=4\).
This test-selected choice favors GRaF, so the experiment diagnoses reference-density sensitivity rather than repeating the main-table protocol.

As shown in~\autoref{tab:revision:graf}, GRaF's PSNR rises by~\(4.24\,\mathrm{dB}\) and its SSIM by~\(0.137\) as the budget increases from~\(M=32\) to~\(M=256\).
Over the same range, \ourSystem's PSNR and SSIM change little.
The PSNR gap narrows from~\(8.50\) to~\(4.30\,\mathrm{dB}\), but does not close even with eight times as many references.
At~\(M=32\), GRaF's PSNR falls from~\(16.79\) to~\(15.96\,\mathrm{dB}\) between queries in the nearest- and farthest-reference quintiles, whereas \ourSystem shows no monotone distance trend.
These observations suggest that GRaF benefits from denser query-local coverage, while \ourSystem remains comparatively stable across the tested budgets.
The experiment covers one fold and one reference draw per budget; unseen-category evaluation and spatially stratified GRaF references were not evaluated.

\begin{table}[t]
\centering
\small
\caption{\textit{Reference-budget sensitivity on one unseen-variant fold.}
GRaF uses the test-selected~\(L=4\).}
\label{tab:revision:graf}
\begin{tabular}{lrrrr}
\toprule
Method & \(M=32\) & \(M=64\) & \(M=128\) & \(M=256\) \\
\midrule
GRaF PSNR & 16.20 & 17.82 & 19.36 & 20.44 \\
GRaF SSIM & 0.525 & 0.576 & 0.627 & 0.662 \\
\midrule
\ourSystem{} PSNR & 24.70 & 24.70 & 24.74 & 24.74 \\
\ourSystem{} SSIM & 0.785 & 0.784 & 0.785 & 0.783 \\
\bottomrule
\end{tabular}
\end{table}

\subsection{Representation and Renderer Diagnostics}
\label{app:representation_choices}
\label{app:error_sources}

\subsubsection{Reference Coordinates}
\label{app:revision:coordinates}

We train spectrum-only, absolute-coordinate, and query-relative-coordinate reference-conditioning variants on one unseen-variant fold and one unseen-category fold, reusing the same pretrained encoder.
Absolute coordinates preserve once-per-scene conditioning, whereas relative coordinates require rebuilding the conditioned field for each query.

Zeroing a trained coordinate branch at evaluation changes accuracy by at most~\(0.13\,\mathrm{dB}\) and~\(0.003\) SSIM, with no consistent direction of change.
Across three seeds, the absolute-coordinate model differs from spectrum-only by~\(-0.05\,\mathrm{dB}\) on the unseen-variant fold and~\(+0.33\,\mathrm{dB}\) on the unseen-category fold, within seed spreads of~\(0.47\) and~\(0.95\,\mathrm{dB}\).
Relative coordinates reduce PSNR by~\(0.70\,\mathrm{dB}\) on one fold and improve it by~\(0.73\,\mathrm{dB}\) on the other, while zeroing their trained coordinate branch again has negligible effect.
On these tested folds, neither coordinate variant provides a consistent accuracy gain over spectrum-only conditioning, which also avoids reference-pose metadata.

\begin{table}[t]
\centering
\small
\caption{\textit{Encoder-training ablation at~\(M=32\).} Results are from seed~0.}
\label{tab:revision:pretraining}
\begin{tabular}{lrrrr}
\toprule
Encoder training & Var. PSNR & Var. SSIM & Cat. PSNR & Cat. SSIM \\
\midrule
Frozen random & 24.22 & 0.778 & 25.06 & 0.794 \\
Synthesis loss only & 24.08 & 0.770 & 25.54 & 0.809 \\
Masked spectrum & 24.75 & 0.793 & 24.84 & 0.790 \\
Supervised contrastive & 24.97 & 0.794 & 25.70 & 0.812 \\
\bottomrule
\end{tabular}
\end{table}

\subsubsection{Pretraining Objective}
\label{app:revision:pretraining}

We compare four encoder-training conditions on one unseen-variant fold and one unseen-category fold, using three seeds and the same~\(5{,}500\) total parameter updates.
The conditions are a frozen random encoder, synthesis-loss-only end-to-end training, masked-spectrum self-supervision, and the supervised contrastive objective used by \ourSystem.

Supervised contrastive pretraining gives the highest seed-0 PSNR and SSIM on both folds in~\autoref{tab:revision:pretraining}.
Across seeds, its PSNR exceeds the best label-free condition by~\(0.24\) and~\(0.35\,\mathrm{dB}\), gains comparable to the seed spreads.
Accuracy alone, however, does not establish that a model uses its references.
In the seed-0 reference-swap probe, seven of the eight trained models respond to a changed reference set, but the synthesis-loss-only model on the unseen-variant fold is exactly reference-blind despite reaching~\(24.08\,\mathrm{dB}\).
Scene-identification accuracy is also insufficient as a reference-use diagnostic: even a frozen random encoder separates scenes using raw spectral statistics.

\subsubsection{Sensitivity to Anchor Resolution}
\label{app:anchor_resolution}

The canonical anchor lattice controls the spatial granularity of the conditioned propagation field.
The default configuration uses~\(K=20\times20\times10=4000\) anchors shared across scenes.
We evaluate whether the reported performance depends strongly on this choice.

\textbf{Protocol.}
We evaluate five lattice resolutions with the same~\(2{:}2{:}1\) aspect ratio:
\[
K\in\left\{500,1372,4000,6912,10976\right\}.
\]
For each resolution, we retrain the propagation model on the first unseen-variant fold while reusing the same pretrained reference encoder.
All other training hyperparameters, the~\(M=32\) references, and the evaluation split remain unchanged.
Because the lattice spacing also determines the allowable anchor-position refinement and the initial spatial scale of each component, changing~\(K\) jointly changes the spatial discretization of the propagation field.

\textbf{The Default Resolution Provides the Best Accuracy.}
As shown in~\autoref{fig:anchor_resolution}, the default~\(K=4000\) lattice achieves the highest~PSNR of~\(24.97\,\mathrm{dB}\) and~SSIM of~\(0.794\), together with~\(0.083\) NMSE.
Coarser lattices reduce~PSNR by~\(0.30\)--\(0.45\,\mathrm{dB}\), indicating that insufficient spatial resolution moderately limits synthesis accuracy.
Increasing the resolution does not improve performance either:~\(K=6912\) and~\(K=10976\) reduce~PSNR by~\(1.29\,\mathrm{dB}\) and~\(0.27\,\mathrm{dB}\), respectively.
AoA error remains nearly unchanged across all resolutions, ranging only from~\(5.06^\circ\) to~\(5.20^\circ\).

\textbf{Higher Resolution Adds Cost Without Consistent Benefit.}
Training time increases steadily with the number of anchors.
On one~H100, training requires approximately~\(50\), \(51\), \(68\), \(92\), and~\(124\) minutes for the five resolutions from smallest to largest.
Thus, the two finer lattices increase training cost by approximately~\(1.4\times\) and~\(1.8\times\) relative to the default without improving synthesis accuracy.

Analysis of the predicted anchor powers also indicates that the model uses only a fraction of the available anchors.
At the default resolution, the anchor-power participation ratio corresponds to approximately~\(244\) effective anchors per query, substantially fewer than the~\(4000\) available anchors.
The finer lattices do not consistently increase this effective capacity: one run concentrates its power on very few anchors, while the other distributes it over substantially more anchors, yet neither improves accuracy.
Because each alternative resolution is trained once, these differences should not be interpreted as an inherent failure mode of finer lattices.
Rather, they show that increasing anchor density beyond the default does not provide a consistent accuracy benefit.

Overall, performance is not highly sensitive to anchor resolution around the default operating point.
A substantially coarser lattice modestly reduces synthesis fidelity, while finer lattices add computation without improving accuracy.
The default~\(K=4000\) therefore provides a practical balance between spatial resolution, synthesis quality, and training cost.

\begin{figure}[t]
\centering
\includegraphics[width=\linewidth]{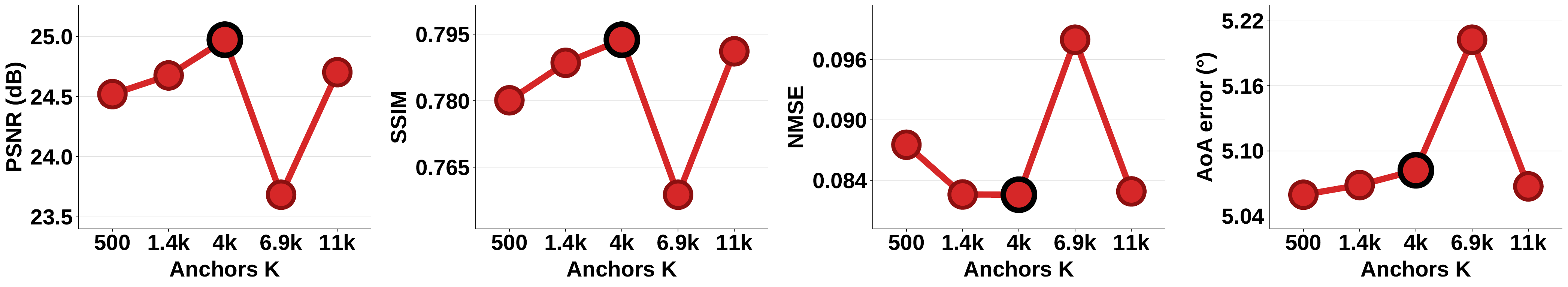}
\caption{\textit{Sensitivity to anchor-lattice resolution.}
Performance is shown across lattice sizes \(K\) with \(M=32\). The default \(K=4{,}000\) corresponds to the reported model.}
\label{fig:anchor_resolution}
\end{figure}

\subsubsection{Oracle Renderer Comparison}
\label{app:revision:oracle}

We ray-trace the first~\(100\) test queries in each simulated scene and render the exact paths coherently, preserving phase cross terms, or incoherently, matching \ourSystem's component-power approximation.
We also evaluate feed-forward~\ourSystem on the same queries.
All predictions are scored against the released noisy targets, so even an exact-path rendering need not match them perfectly.

As shown in~\autoref{tab:revision:oracle}, the incoherent oracle is~\(0.38\,\mathrm{dB}\) below the coherent oracle in average PSNR, while feed-forward~\ourSystem is another~\(0.66\,\mathrm{dB}\) below the incoherent oracle.
The incoherent oracle nevertheless scores slightly better on SSIM, NMSE, and AoA error, so the coherent renderer is not uniformly better against these noisy targets.
The independent noisy coherent sample further illustrates the effect of target noise on the measured scores.

The average PSNR difference conceals a larger approximation error in dense multipath.
When the Ricean factor is below~\(0\,\mathrm{dB}\), the incoherent oracle reaches~\(18.34\,\mathrm{dB}\), compared with~\(19.98\,\mathrm{dB}\) for the coherent oracle.
When at least four paths lie within~\(10\,\mathrm{dB}\) of the strongest, the corresponding scores are~\(16.81\) and~\(18.55\,\mathrm{dB}\).
Thus, omitting coherent interference incurs a modest average PSNR cost on this corpus but a substantially larger cost in its dense-multipath subsets.

\subsection{Deployment Efficiency}
\label{app:deployment_efficiency}

\begin{figure}[t]
\centering
\subfigure[PSNR]{\includegraphics[width=0.49\linewidth]{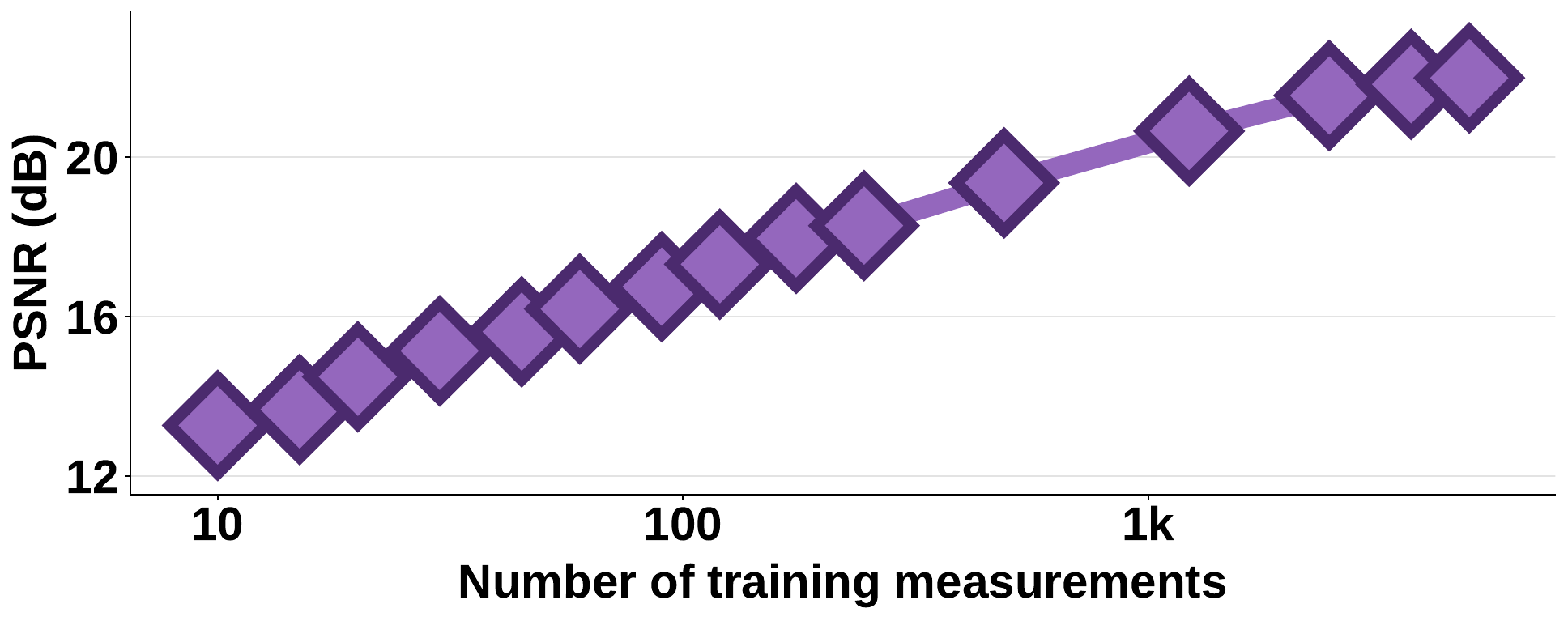}\label{fig:data_efficiency_a}}\hfill
\subfigure[SSIM]{\includegraphics[width=0.49\linewidth]{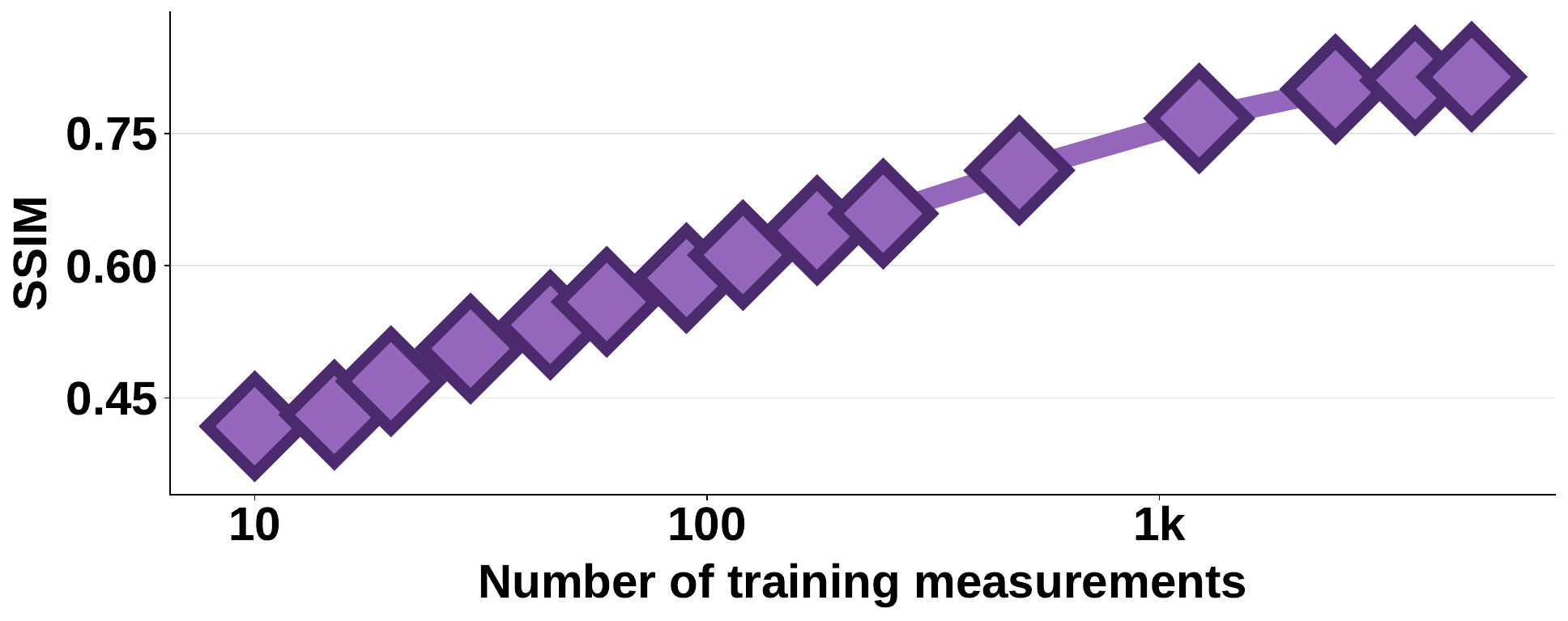}\label{fig:data_efficiency_b}}
\caption{\textit{GSRF data efficiency on S36.}
GSRF is trained on increasing subsets of the S36 measurements and evaluated on the same test set.
}
\label{fig:data_efficiency}
\end{figure}

\begin{table}[t]
\centering
\small
\caption{\textit{Oracle rendering on~\(3{,}500\) simulated queries.} 
}
\label{tab:revision:oracle}
\begin{tabular}{lrrrr}
\toprule
Prediction & PSNR$\uparrow$ & SSIM$\uparrow$ & NMSE$\downarrow$ & AoA$\downarrow$ \\
\midrule
Oracle coherent renderer & \textbf{23.77} & 0.744 & 0.111 & 7.91 \\
Independent noisy coherent sample & 22.62 & 0.720 & 0.129 & 8.90 \\
Oracle incoherent renderer & 23.39 & \textbf{0.746} & \textbf{0.103} & \textbf{7.63} \\
Feed-forward~\ourSystem & 22.73 & 0.734 & 0.127 & 8.67 \\
\bottomrule
\end{tabular}
\end{table}

\subsubsection{Data-Collection Efficiency}
\label{app:efficiency}

\textbf{Measurement Cost.}
Full-data per-scene training uses the target scene's training split: a median of~\(3376\) measurements per scene and a range of~\(2823\)--\(4286\).
In the cross-scene setting,~\ourSystem and~\ourSystemPR use~\(M=32\) target-scene references, approximately two orders of magnitude fewer measurements than full-data fitting.
The reference-budget analysis in~\S\ref{app:reference_budget} shows that this budget already lies in the regime where additional references provide little improvement in average accuracy.

The calibration-only protocol uses~\(M=32\) scene references and~\(N_{\mathrm C}=64\) separate calibration measurements, for a total of~\(96\) target-domain spectra.
The joint protocol conditions on the same~\(32\) references, calibrates with the same separate~\(64\) spectra, and reuses both sets for power refinement.
For an illustrative acquisition time of~\(10\,\mathrm{s}\) to~\(1\,\mathrm{min}\) per location, the upper end reflecting a Wi-Fi site-survey estimate~\citep{li2017turf}, collecting~\(3376\) measurements would take~\(9.4\)--\(56.3\) hours, compared with~\(5.3\)--\(32\) minutes for~\(32\) references.
The~\(64\) separate calibration measurements would add~\(10.7\)--\(64\) minutes under the same assumption; the reported protocol collects them for each scene.

\textbf{Empirical Data Efficiency.}
We directly examine how much target-scene data a per-scene field requires using the real scene~S36.
We train~GSRF on random subsets of its~\(4898\) non-test measurements, using~\(15\) subset sizes from~\(10\) to~\(4898\) and three random seeds per size, and evaluate every model on the same~\(1225\) test measurements.

\autoref{fig:data_efficiency} shows that accuracy improves steadily with the amount of real training data.
With~\(30\) measurements, close to the reference budget of~\ourSystem, GSRF achieves~\(15.14\,\mathrm{dB}\) PSNR and~\(0.506\) SSIM.
With~\(2449\) measurements on~S36, GSRF comes within~\(0.44\,\mathrm{dB}\) PSNR and~\(0.014\) SSIM of its full-data result of~\(21.98\,\mathrm{dB}\) and~\(0.814\).

\subsubsection{Computational Efficiency}
\label{app:compute_efficiency}
\autoref{tab:efficiency} reports wall-clock measurements on one NVIDIA~H100.
A single~\ourSystem model contains~\(12.7\)M inference-time parameters and is shared across scenes.
Training takes approximately~\(86\) minutes per fold, after which feed-forward conditioning on the~\(M=32\) references of a new scene requires~\(10\,\mathrm{ms}\).
This timing describes feed-forward~\ourSystem only.
Depending on whether anchor biases, the residual network, or both are fitted, the refinement optimizes~\(0.004\)--\(0.270\)M parameters and requires~\(47\)--\(202\) seconds of measured scene setup.

\begin{table}[t]
\centering\small
\caption{\textit{Cost of deployment in an unseen scene.}
Runtime is measured on the same NVIDIA~H100.
}
\label{tab:efficiency}
\setlength{\tabcolsep}{5pt}
\begin{tabular}{lccccc}
\toprule
Method & Parameters & Pre-training & Measurements & Setup & Inference \\
\midrule
\textbf{\ourSystem} & \textbf{12.7M} & \textbf{86\,min} & \textbf{32} & \textbf{10\,ms} & \textbf{8.5\,ms} \\
\ourSystemPR & 12.7M + 0.004--0.270M & 86\,min & 32 & 47--202\,s & 8.5\,ms \\
\midrule
GRaF & 12.6M & 177\,min & 32 & none & 2.9\,s \\
\midrule
WRF-GS & 1.3M & none & 3{,}376 & 21--25\,min & 3.7\,ms \\
GSRF & 4.4M & none & 3{,}376 & 45--99\,min & 12\,ms \\
NeRF$^2$ & 0.69M & none & 3{,}376 & 89--172\,min & 125\,ms \\
\bottomrule
\end{tabular}
\end{table}

Rendering requires approximately~\(8.5\,\mathrm{ms}\) per query.
A complete feed-forward test-scene map, containing a median of~\(957\) queries, takes approximately~\(16\)--\(21\,\mathrm{s}\) end to end, including scene conditioning and I/O; power refinement adds its one-time target-scene fit.
Once trained, the per-scene baselines have comparable per-query rendering costs for WRF-GS and GSRF, at approximately~\(3.7\) and~\(12\,\mathrm{ms}\), while NeRF$^2$ requires approximately~\(125\,\mathrm{ms}\).
Their primary deployment cost is scene-specific optimization before rendering:~\(21\)--\(25\) minutes for WRF-GS, \(45\)--\(99\) minutes for GSRF, and~\(89\)--\(172\) minutes for NeRF$^2$ under their released training configurations.

GRaF is also generalizable and therefore does not require per-scene fitting.
Its pretraining takes~\(177\) minutes per fold, but its neighbor retrieval and volumetric rendering require approximately~\(2.9\,\mathrm{s}\) per query.
A complete test-scene map therefore takes approximately~\(48\) minutes, compared with seconds for feed-forward~\ourSystem or a one-time fit plus seconds of rendering for~\ourSystemPR.

\subsection{Physical Distribution Shift}
\label{app:physical_stress}
\label{app:revision:material}

We test feed-forward~\ourSystem under material-property shifts by re-simulating all~\(35\) scenes after scaling the relative permittivity and conductivity of every non-metal material.
Geometry, poses, array, frequency, ray-tracing settings, and per-scene measurement~SNR remain fixed.
Each scene uses its unchanged unseen-variant-fold checkpoint, one fixed draw of~\(M=32\) shifted references, and its first~\(300\) test queries, without parameter adaptation.

The~\((1,1)\) row in~\autoref{tab:revision:material} is a re-simulated control.
Reducing the material parameters yields small PSNR gains in this test, whereas increasing them progressively degrades performance.
At the~\((1.2,1.5)\) shift, PSNR falls by~\(0.71\,\mathrm{dB}\).
At the most severe~\((2,10)\) shift, it falls by~\(5.14\,\mathrm{dB}\), and AoA error rises from~\(8.49^\circ\) to~\(21.66^\circ\).
This last setting serves as a stress test of the fixed propagation model, not a typical deployment perturbation.

\begin{table}[t]
\centering
\small
\caption{Material-property shift for feed-forward~\ourSystem, averaged over the five unseen-variant folds. 
}
\label{tab:revision:material}
\begin{tabular}{lrrrrr}
\toprule
Scale~\((\varepsilon_r,\sigma_{\mathrm{cond}})\) & PSNR & SSIM & NMSE & AoA & \(\Delta\)PSNR \\
\midrule
\((1,1)\) & 22.90 & 0.738 & 0.122 & 8.49 & 0.00 \\
\((0.85,0.67)\) & 22.99 & 0.740 & 0.119 & 8.05 & +0.09 \\
\((0.7,0.33)\) & 23.34 & 0.747 & 0.113 & 7.77 & +0.44 \\
\((1.2,1.5)\) & 22.19 & 0.725 & 0.132 & 9.35 & -0.71 \\
\((1.5,3)\) & 20.43 & 0.679 & 0.168 & 12.92 & -2.46 \\
\((2,10)\) & 17.76 & 0.598 & 0.254 & 21.66 & -5.14 \\
\bottomrule
\end{tabular}
\end{table}


\begin{thebibliography}{54}
\providecommand{\natexlab}[1]{#1}
\providecommand{\url}[1]{\texttt{#1}}
\expandafter\ifx\csname urlstyle\endcsname\relax
  \providecommand{\doi}[1]{doi: #1}\else
  \providecommand{\doi}{doi: \begingroup \urlstyle{rm}\Url}\fi

\bibitem[{3GPP}(2019)]{3gpp38901}
{3GPP}.
\newblock {Study on Channel Model for Frequencies from 0.5 to 100 GHz}.
\newblock Technical Report 38.901, V16.1.0 (Release 16), 3rd Generation Partnership Project, December 2019.
\newblock URL \url{https://www.3gpp.org/ftp/Specs/archive/38_series/38.901/38901-g10.zip}.

\bibitem[{3GPP}(2024)]{3gpp38455}
{3GPP}.
\newblock {NG-RAN; NR Positioning Protocol a (NRPPa)}.
\newblock Technical Specification 38.455, V18.4.0 (Release 18), 3rd Generation Partnership Project, December 2024.
\newblock URL \url{https://www.3gpp.org/ftp/Specs/archive/38_series/38.455/38455-i40.zip}.

\bibitem[{3GPP}(2026)]{3gpp37355}
{3GPP}.
\newblock {LTE Positioning Protocol (LPP)}.
\newblock Technical Specification 37.355, V19.3.0 (Release 19), 3rd Generation Partnership Project, June 2026.
\newblock URL \url{https://www.3gpp.org/ftp/Specs/archive/37_series/37.355/37355-j30.zip}.

\bibitem[{AISG}(2024)]{aisg2024als}
{AISG}.
\newblock {Antenna Location and Orientation Sensor}.
\newblock Subunit Type Standard {AISG-ST-ALS}, v{ALS}3.0.2.2, Antenna Interface Standards Group, June 2024.
\newblock URL \url{https://aisg.org.uk/file/AISG-ALS-vALS3.0.2.pdf}.

\bibitem[Alkhateeb et~al.(2017)Alkhateeb, Nam, Rahman, Zhang, and Heath]{alkhateeb2017beamassociation}
Ahmed Alkhateeb, Young-Han Nam, Md~Saifur Rahman, Jianzhong Zhang, and Robert~W. Heath.
\newblock {Initial Beam Association in Millimeter Wave Cellular Systems: Analysis and Design Insights}.
\newblock \emph{IEEE Transactions on Wireless Communications}, 16\penalty0 (5):\penalty0 2807--2821, 2017.

\bibitem[Berman et~al.(2025)Berman, Gannot, and Tirer]{spnet2025doa}
Lioz Berman, Sharon Gannot, and Tom Tirer.
\newblock {(SP)$^2$-Net: A Neural Spatial Spectrum Method for DOA Estimation}.
\newblock \emph{arXiv preprint arXiv:2509.15475}, 2025.

\bibitem[Bian et~al.(2025)Bian, Tao, Sun, Zhang, and Yu]{bian2025genert}
Kejia Bian, Meixia Tao, Shu Sun, Tongjia Zhang, and Jun Yu.
\newblock {GeNeRT: A Physics-Informed Approach to Intelligent Wireless Channel Modeling via Generalizable Neural Ray Tracing}.
\newblock \emph{arXiv preprint arXiv:2506.18295}, 2025.

\bibitem[Charatan et~al.(2024)Charatan, Li, Tagliasacchi, and Sitzmann]{charatan2024pixelsplat}
David Charatan, Sizhe~Lester Li, Andrea Tagliasacchi, and Vincent Sitzmann.
\newblock {pixelSplat: 3D Gaussian Splats from Image Pairs for Scalable Generalizable 3D Reconstruction}.
\newblock In \emph{IEEE/CVF Conference on Computer Vision and Pattern Recognition~(CVPR)}, 2024.

\bibitem[Chen et~al.(2021)Chen, Xu, Zhao, Zhang, Xiang, Yu, and Su]{chen2021mvsnerf}
Anpei Chen, Zexiang Xu, Fuqiang Zhao, Xiaoshuai Zhang, Fanbo Xiang, Jingyi Yu, and Hao Su.
\newblock {MVSNeRF: Fast Generalizable Radiance Field Reconstruction from Multi-View Stereo}.
\newblock In \emph{IEEE/CVF International Conference on Computer Vision~(ICCV)}, 2021.

\bibitem[Chen et~al.(2024{\natexlab{a}})Chen, Feng, Sun, Qian, and Zhang]{chen2024rfcanvas}
Xingyu Chen, Zihao Feng, Ke~Sun, Kun Qian, and Xinyu Zhang.
\newblock {RFCanvas: Modeling RF Channel by Fusing Visual Priors and Few-Shot RF Measurements}.
\newblock In \emph{ACM International Conference on Embedded Networked Sensor Systems~(SenSys)}, 2024{\natexlab{a}}.

\bibitem[Chen et~al.(2024{\natexlab{b}})Chen, Xu, Zheng, Zhuang, Pollefeys, Geiger, Cham, and Cai]{chen2024mvsplat}
Yuedong Chen, Haofei Xu, Chuanxia Zheng, Bohan Zhuang, Marc Pollefeys, Andreas Geiger, Tat-Jen Cham, and Jianfei Cai.
\newblock {MVSplat: Efficient 3D Gaussian Splatting from Sparse Multi-View Images}.
\newblock In \emph{European Conference on Computer Vision~(ECCV)}, 2024{\natexlab{b}}.

\bibitem[Garnelo et~al.(2018)Garnelo, Rosenbaum, Maddison, Ramalho, Saxton, Shanahan, Teh, Rezende, and Eslami]{garnelo2018cnp}
Marta Garnelo, Dan Rosenbaum, Chris~J. Maddison, Tiago Ramalho, David Saxton, Murray Shanahan, Yee~Whye Teh, Danilo~J. Rezende, and S.~M.~Ali Eslami.
\newblock {Conditional Neural Processes}.
\newblock In \emph{International Conference on Machine Learning~(ICML)}, 2018.

\bibitem[Guo et~al.(2026)Guo, Zhong, Tong, Lyu, and Zhang]{guo2025nbf}
Keqiang Guo, Yuheng Zhong, Xin Tong, Jiangbin Lyu, and Rui Zhang.
\newblock {Neural Beam Field for Spatial Beam RSRP Prediction}.
\newblock In \emph{IEEE Wireless Communications and Networking Conference~(WCNC)}, 2026.

\bibitem[Hoydis et~al.(2023)Hoydis, A{\"\i}t~Aoudia, Cammerer, Nimier-David, Binder, Marcus, and Keller]{hoydis2023sionnart}
Jakob Hoydis, Fay{\c{c}}al A{\"\i}t~Aoudia, Sebastian Cammerer, Merlin Nimier-David, Nikolaus Binder, Guillermo Marcus, and Alexander Keller.
\newblock {Sionna RT: Differentiable Ray Tracing for Radio Propagation Modeling}.
\newblock In \emph{IEEE Globecom Workshops (GC Wkshps)}, 2023.

\bibitem[Huang et~al.(2024)Huang, Miller, Prabhakara, Jin, Laroia, Kolter, and Rowe]{huang2024dart}
Tianshu Huang, John Miller, Akarsh Prabhakara, Tao Jin, Tarana Laroia, Zico Kolter, and Anthony Rowe.
\newblock {DART: Implicit Doppler Tomography for Radar Novel View Synthesis}.
\newblock In \emph{IEEE/CVF Conference on Computer Vision and Pattern Recognition~(CVPR)}, 2024.

\bibitem[{ITU-R}(2015)]{itur2015p2040}
{ITU-R}.
\newblock {Effects of Building Materials and Structures on Radiowave Propagation above about 100~MHz}.
\newblock Recommendation ITU-R P.2040-1, International Telecommunication Union, July 2015.
\newblock URL \url{https://www.itu.int/rec/R-REC-P.2040-1-201507-S/en}.

\bibitem[{ITU-R}(2023)]{itur2023p2040}
{ITU-R}.
\newblock {Effects of Building Materials and Structures on Radiowave Propagation above about 100~MHz}.
\newblock Recommendation ITU-R P.2040-3, International Telecommunication Union, August 2023.
\newblock URL \url{https://www.itu.int/rec/R-REC-P.2040-3-202308-S/en}.

\bibitem[Ji et~al.(2024)Ji, Mao, Xi, and Chen]{shmuel2024transmusic}
Junkai Ji, Wei Mao, Feng Xi, and Shengyao Chen.
\newblock {TransMUSIC: A Transformer-Aided Subspace Method for DOA Estimation with Low-Resolution ADCs}.
\newblock In \emph{IEEE International Conference on Acoustics, Speech and Signal Processing~(ICASSP)}, 2024.

\bibitem[Jin et~al.(2025)Jin, Jiang, Tan, Zhang, Bi, Zhang, Luan, Snavely, and Xu]{jin2025lvsm}
Haian Jin, Hanwen Jiang, Hao Tan, Kai Zhang, Sai Bi, Tianyuan Zhang, Fujun Luan, Noah Snavely, and Zexiang Xu.
\newblock {LVSM: A Large View Synthesis Model with Minimal 3D Inductive Bias}.
\newblock In \emph{International Conference on Learning Representations~(ICLR)}, 2025.

\bibitem[Kim et~al.(2019)Kim, Mnih, Schwarz, Garnelo, Eslami, Rosenbaum, Vinyals, and Teh]{kim2019anp}
Hyunjik Kim, Andriy Mnih, Jonathan Schwarz, Marta Garnelo, Ali Eslami, Dan Rosenbaum, Oriol Vinyals, and Yee~Whye Teh.
\newblock {Attentive Neural Processes}.
\newblock In \emph{International Conference on Learning Representations~(ICLR)}, 2019.

\bibitem[Krim \& Viberg(1996)Krim and Viberg]{krim1996twodecades}
Hamid Krim and Mats Viberg.
\newblock {Two Decades of Array Signal Processing Research: The Parametric Approach}.
\newblock \emph{IEEE Signal Processing Magazine}, 13\penalty0 (4):\penalty0 67--94, 1996.

\bibitem[Lee et~al.(2019)Lee, Lee, Kim, Kosiorek, Choi, and Teh]{lee2019settransformer}
Juho Lee, Yoonho Lee, Jungtaek Kim, Adam~R. Kosiorek, Seungjin Choi, and Yee~Whye Teh.
\newblock {Set Transformer: A Framework for Attention-Based Permutation-Invariant Neural Networks}.
\newblock In \emph{International Conference on Machine Learning~(ICML)}, 2019.

\bibitem[Li et~al.(2017)Li, Xu, Gong, and Zheng]{li2017turf}
Chenhe Li, Qiang Xu, Zhe Gong, and Rong Zheng.
\newblock {TuRF: Fast Data Collection for Fingerprint-Based Indoor Localization}.
\newblock In \emph{International Conference on Indoor Positioning and Indoor Navigation~(IPIN)}, 2017.

\bibitem[Liu et~al.(2026)Liu, Zhang, Yang, Cao, Xu, Xu, Sun, Dai, and Guan]{sudo2025swiftwrf}
Mufan Liu, Cixiao Zhang, Qi~Yang, Yujie Cao, Yiling Xu, Yin Xu, Shu Sun, Mingzeng Dai, and Yunfeng Guan.
\newblock {Deformable 2D Gaussian Splatting for Efficient Wireless Radiance Field Rendering}.
\newblock \emph{IEEE Transactions on Visualization and Computer Graphics}, 32\penalty0 (7):\penalty0 6363--6378, 2026.

\bibitem[Liu et~al.(2018)Liu, Zhang, and Yu]{liu2018dnnimperfections}
Zhang-Meng Liu, Chenwei Zhang, and Philip~S. Yu.
\newblock {Direction-Of-Arrival Estimation Based on Deep Neural Networks with Robustness to Array Imperfections}.
\newblock \emph{IEEE Transactions on Antennas and Propagation}, 66\penalty0 (12):\penalty0 7315--7327, 2018.

\bibitem[Lu et~al.(2024)Lu, Vattheuer, Mirzasoleiman, and Abari]{lu2024newrf}
Haofan Lu, Christopher Vattheuer, Baharan Mirzasoleiman, and Omid Abari.
\newblock {NeWRF: A Deep Learning Framework for Wireless Radiation Field Reconstruction and Channel Prediction}.
\newblock In \emph{International Conference on Machine Learning~(ICML)}, 2024.

\bibitem[Merkofer et~al.(2024)Merkofer, Revach, Shlezinger, Routtenberg, and van Sloun]{merkofer2024damusic}
Julian~P. Merkofer, Guy Revach, Nir Shlezinger, Tirza Routtenberg, and Ruud J.~G. van Sloun.
\newblock {DA-MUSIC: Data-Driven DoA Estimation via Deep Augmented MUSIC Algorithm}.
\newblock \emph{IEEE Transactions on Vehicular Technology}, 73\penalty0 (2):\penalty0 2771--2785, 2024.

\bibitem[{Mukund Varma T} et~al.(2023){Mukund Varma T}, Wang, Chen, Chen, Venugopalan, and Wang]{wang2023gnt}
{Mukund Varma T}, Peihao Wang, Xuxi Chen, Tianlong Chen, Subhashini Venugopalan, and Zhangyang Wang.
\newblock {Is Attention All That NeRF Needs?}
\newblock In \emph{International Conference on Learning Representations~(ICLR)}, 2023.

\bibitem[Nguyen \& Grover(2022)Nguyen and Grover]{nguyen2022tnp}
Tung Nguyen and Aditya Grover.
\newblock {Transformer Neural Processes: Uncertainty-Aware Meta Learning via Sequence Modeling}.
\newblock In \emph{International Conference on Machine Learning~(ICML)}, 2022.

\bibitem[Nukapotula et~al.(2025)Nukapotula, Tripathi, Pregler, Kalathil, Shakkottai, and Rappaport]{nukapotula2025gsparc}
Bhavya~Sai Nukapotula, Rishabh Tripathi, Seth Pregler, Dileep Kalathil, Srinivas Shakkottai, and Theodore~S. Rappaport.
\newblock {GSpaRC: Gaussian Splatting for Real-Time Reconstruction of RF Channels}.
\newblock \emph{arXiv preprint arXiv:2511.22793}, 2025.

\bibitem[Orekondy et~al.(2023)Orekondy, Pratik, Kadambi, Ye, Soriaga, and Behboodi]{orekondy2023winert}
Tribhuvanesh Orekondy, Kumar Pratik, Shreya Kadambi, Hao Ye, Joseph Soriaga, and Arash Behboodi.
\newblock {WiNeRT: Towards Neural Ray Tracing for Wireless Channel Modelling and Differentiable Simulations}.
\newblock In \emph{International Conference on Learning Representations~(ICLR)}, 2023.

\bibitem[Pan et~al.(2023)Pan, Liu, Liu, Qi, Huang, Zheng, Wu, and Gardill]{pan2023insitu}
Mengguan Pan, Shengheng Liu, Peng Liu, Wangdong Qi, Yongming Huang, Wang Zheng, Qihui Wu, and Markus Gardill.
\newblock {In Situ Calibration of Antenna Arrays for Positioning with 5G Networks}.
\newblock \emph{IEEE Transactions on Microwave Theory and Techniques}, 71\penalty0 (10):\penalty0 4600--4613, 2023.

\bibitem[Perez et~al.(2018)Perez, Strub, de~Vries, Dumoulin, and Courville]{perez2018film}
Ethan Perez, Florian Strub, Harm de~Vries, Vincent Dumoulin, and Aaron Courville.
\newblock {FiLM: Visual Reasoning with a General Conditioning Layer}.
\newblock In \emph{AAAI Conference on Artificial Intelligence~(AAAI)}, 2018.

\bibitem[Sajjadi et~al.(2022)Sajjadi, Meyer, Pot, Bergmann, Greff, Radwan, Vora, Lu{\v{c}}i{\'c}, Duckworth, Dosovitskiy, Uszkoreit, Funkhouser, and Tagliasacchi]{sajjadi2022srt}
Mehdi S.~M. Sajjadi, Henning Meyer, Etienne Pot, Urs Bergmann, Klaus Greff, Noha Radwan, Suhani Vora, Mario Lu{\v{c}}i{\'c}, Daniel Duckworth, Alexey Dosovitskiy, Jakob Uszkoreit, Thomas Funkhouser, and Andrea Tagliasacchi.
\newblock {Scene Representation Transformer: Geometry-Free Novel View Synthesis through Set-Latent Scene Representations}.
\newblock In \emph{IEEE/CVF Conference on Computer Vision and Pattern Recognition~(CVPR)}, 2022.

\bibitem[Schmidt(1986)]{schmidt1986music}
Ralph~O. Schmidt.
\newblock {Multiple Emitter Location and Signal Parameter Estimation}.
\newblock \emph{IEEE Transactions on Antennas and Propagation}, 34\penalty0 (3):\penalty0 276--280, 1986.

\bibitem[Shen et~al.(2026)Shen, Lago~Enamorado, Mao, and Wang]{shen2026gainerf}
Jingzhou Shen, Luis Lago~Enamorado, Shiwen Mao, and Xuyu Wang.
\newblock {A Geometric Algebra-Informed NeRF Framework for Generalizable Wireless Channel Prediction}.
\newblock In \emph{IEEE Conference on Computer Communications~(INFOCOM)}, 2026.

\bibitem[Shmuel et~al.(2025)Shmuel, Merkofer, Revach, van Sloun, and Shlezinger]{shmuel2023subspacenet}
Dor~H. Shmuel, Julian~P. Merkofer, Guy Revach, Ruud J.~G. van Sloun, and Nir Shlezinger.
\newblock {SubspaceNet: Deep Learning-Aided Subspace Methods for DoA Estimation}.
\newblock \emph{IEEE Transactions on Vehicular Technology}, 74\penalty0 (3):\penalty0 4962--4976, 2025.

\bibitem[Stoica et~al.(2011)Stoica, Babu, and Li]{stoica2011spice}
Petre Stoica, Prabhu Babu, and Jian Li.
\newblock {SPICE: A Sparse Covariance-Based Estimation Method for Array Processing}.
\newblock \emph{IEEE Transactions on Signal Processing}, 59\penalty0 (2):\penalty0 629--638, 2011.

\bibitem[Wang et~al.(2025)Wang, Chen, Karaev, Vedaldi, Rupprecht, and Novotny]{wang2025vggt}
Jianyuan Wang, Minghao Chen, Nikita Karaev, Andrea Vedaldi, Christian Rupprecht, and David Novotny.
\newblock {VGGT: Visual Geometry Grounded Transformer}.
\newblock In \emph{IEEE/CVF Conference on Computer Vision and Pattern Recognition~(CVPR)}, 2025.

\bibitem[Wang et~al.(2021)Wang, Wang, Genova, Srinivasan, Zhou, Barron, Martin-Brualla, Snavely, and Funkhouser]{wang2021ibrnet}
Qianqian Wang, Zhicheng Wang, Kyle Genova, Pratul Srinivasan, Howard Zhou, Jonathan~T. Barron, Ricardo Martin-Brualla, Noah Snavely, and Thomas Funkhouser.
\newblock {IBRNet: Learning Multi-View Image-Based Rendering}.
\newblock In \emph{IEEE/CVF Conference on Computer Vision and Pattern Recognition~(CVPR)}, 2021.

\bibitem[Wang et~al.(2026)Wang, Huang, and Cheng]{wang2026radiodiffv2}
Xiucheng Wang, Junxi Huang, and Nan Cheng.
\newblock {RadioDiff-v2: Generative Angular Radio Maps for Multi-Beam Selection and Localization}.
\newblock \emph{arXiv preprint arXiv:2607.08045}, 2026.

\bibitem[Wen et~al.(2025)Wen, Tong, Hu, Lin, and Zhang]{wen2025wrfgs}
Chaozheng Wen, Jingwen Tong, Yingdong Hu, Zehong Lin, and Jun Zhang.
\newblock {WRF-GS: Wireless Radiation Field Reconstruction with 3D Gaussian Splatting}.
\newblock In \emph{IEEE Conference on Computer Communications~(INFOCOM)}, 2025.

\bibitem[Wen et~al.(2026)Wen, Tong, Hu, Lin, and Zhang]{wen2024wrfgs}
Chaozheng Wen, Jingwen Tong, Yingdong Hu, Zehong Lin, and Jun Zhang.
\newblock {Neural Representation for Wireless Radiation Field Reconstruction: A 3D Gaussian Splatting Approach}.
\newblock \emph{IEEE Transactions on Wireless Communications}, 25:\penalty0 7490--7504, 2026.

\bibitem[Xue et~al.(2024)Xue, Ji, Ma, Guo, Xu, Chen, and Zhang]{xue2024beammanagement}
Qing Xue, Chengwang Ji, Shaodan Ma, Jiajia Guo, Yongjun Xu, Qianbin Chen, and Wei Zhang.
\newblock {A Survey of Beam Management for mmWave and THz Communications towards 6G}.
\newblock \emph{IEEE Communications Surveys and Tutorials}, 26\penalty0 (3):\penalty0 1520--1559, 2024.

\bibitem[Yang et~al.(2025)Yang, Dong, Ji, Du, and Srivastava]{yang2025gsrf}
Kang Yang, Gaofeng Dong, Sijie Ji, Wan Du, and Mani Srivastava.
\newblock {GSRF: Complex-Valued 3D Gaussian Splatting for Efficient Radio-Frequency Data Synthesis}.
\newblock In \emph{Conference on Neural Information Processing Systems~(NeurIPS)}, 2025.

\bibitem[Yang et~al.(2026)Yang, Chen, and Du]{yang2026graf}
Kang Yang, Yuning Chen, and Wan Du.
\newblock {Generalizable Radio-Frequency Radiance Fields for Spatial Spectrum Synthesis}.
\newblock In \emph{IEEE/CVF Conference on Computer Vision and Pattern Recognition~(CVPR)}, 2026.

\bibitem[Yang et~al.(2013)Yang, Xie, and Zhang]{yang2013offgridsbl}
Zai Yang, Lihua Xie, and Cishen Zhang.
\newblock {Off-Grid Direction of Arrival Estimation Using Sparse Bayesian Inference}.
\newblock \emph{IEEE Transactions on Signal Processing}, 61\penalty0 (1):\penalty0 38--43, 2013.

\bibitem[Yu et~al.(2021)Yu, Ye, Tancik, and Kanazawa]{yu2021pixelnerf}
Alex Yu, Vickie Ye, Matthew Tancik, and Angjoo Kanazawa.
\newblock {pixelNeRF: Neural Radiance Fields from One or Few Images}.
\newblock In \emph{IEEE/CVF Conference on Computer Vision and Pattern Recognition~(CVPR)}, 2021.

\bibitem[Zaheer et~al.(2017)Zaheer, Kottur, Ravanbakhsh, P{\'o}czos, Salakhutdinov, and Smola]{zaheer2017deepsets}
Manzil Zaheer, Satwik Kottur, Siamak Ravanbakhsh, Barnab{\'a}s P{\'o}czos, Ruslan Salakhutdinov, and Alexander Smola.
\newblock {Deep Sets}.
\newblock In \emph{Conference on Neural Information Processing Systems~(NeurIPS)}, 2017.

\bibitem[Zeng et~al.(2024)Zeng, Chen, Xu, Wu, Xu, Jin, Gao, Gesbert, Cui, and Zhang]{zeng2024ckm}
Yong Zeng, Junting Chen, Jie Xu, Di~Wu, Xiaoli Xu, Shi Jin, Xiqi Gao, David Gesbert, Shuguang Cui, and Rui Zhang.
\newblock {A Tutorial on Environment-Aware Communications via Channel Knowledge Map for 6G}.
\newblock \emph{IEEE Communications Surveys and Tutorials}, 26\penalty0 (3):\penalty0 1478--1519, 2024.

\bibitem[Zhang et~al.(2025)Zhang, Li, and Sun]{rfpgs2025}
Lihao Zhang, Zongtan Li, and Haijian Sun.
\newblock {RF-PGS: Fully-Structured Spatial Wireless Channel Representation with Planar Gaussian Splatting}.
\newblock \emph{arXiv preprint arXiv:2508.16849}, 2025.

\bibitem[Zhang et~al.(2026{\natexlab{a}})Zhang, Sun, Berweger, Gentile, and Hu]{zhang2024rf3dgs}
Lihao Zhang, Haijian Sun, Samuel Berweger, Camillo Gentile, and Rose~Qingyang Hu.
\newblock {RF-3DGS: Wireless Channel Modeling with Radio Radiance Field and 3D Gaussian Splatting}.
\newblock \emph{IEEE Transactions on Wireless Communications}, 25:\penalty0 10419--10433, 2026{\natexlab{a}}.

\bibitem[Zhang et~al.(2026{\natexlab{b}})Zhang, Zhao, Liu, Alkhateeb, Agrawal, and Qu]{zhang2026radtwin}
Yuru Zhang, Ming Zhao, Qiang Liu, Ahmed Alkhateeb, Abhishek~K. Agrawal, and Qi~Qu.
\newblock {RadTwin: Generalizable Wireless Digital Twin for Dynamic Environments}.
\newblock In \emph{International Conference on Computer Communications and Networks~(ICCCN)}, 2026{\natexlab{b}}.

\bibitem[Zhao et~al.(2023)Zhao, An, Pan, and Yang]{zhao2023nerf2}
Xiaopeng Zhao, Zhenlin An, Qingrui Pan, and Lei Yang.
\newblock {NeRF$^2$: Neural Radio-Frequency Radiance Fields}.
\newblock In \emph{Annual International Conference on Mobile Computing and Networking~(MobiCom)}, 2023.

\end{thebibliography}
\end{document}